\documentclass[12pt,letterpaper]{JHEP3}

\usepackage{amscd,amsmath,amssymb,amsfonts,xspace,mathrsfs,amsthm}
\usepackage{color}
\let\normalcolor\relax
\usepackage{bbold}
\usepackage{latexsym}
\usepackage{graphicx}
\usepackage{dsfont}
\usepackage{color}
\usepackage{longtable}

 \newtheorem{theorem}{Theorem}%[section]
  \newtheorem{proposition}{Proposition}%[section]
  \newtheorem{lemma}{Lemma}
  
\numberwithin{equation}{section}

\def\lfig#1#2#3#4#5{
\begin{figure}[t]
 \centerline{\includegraphics[width=#3]{#2}}
 \vspace{#5}
  \caption{#1 \label{#4}}
 \end{figure}
}

\def\det{\,{\rm det}\, }

\def\sign{{\rm sgn}}

\def\Ch{{\rm Ch}}

\def\Sym{\,{\rm Sym}\, }

\def\Im{\,{\rm Im}\,}
\def\Re{\,{\rm Re}\,}

\def\({\left(}
\def\){\right)}
\def\[{\left[}
\def\]{\right]}
\def\<{\left\langle}
\def\>{\right\rangle}
\def\hf{{1\over 2}}
\def\haf{\textstyle{1\over 2}}

\newcommand{\de}{\mathrm{d}}

\newcommand{\I}{\mathrm{i}}

\newcommand{\cL}{\mathcal{L}}
\newcommand{\cD}{\mathcal{D}}

\newcommand{\p}{\partial}

\newcommand{\cV}{\mathcal{V}}

\newcommand{\cG}{\mathcal{G}}

\newcommand{\cM}{\mathcal{M}}

\newcommand{\cN}{\mathcal{N}}
\newcommand{\cE}{\mathcal{E}}
\newcommand{\cX}{\mathcal{X}}

\newcommand{\cT}{\mathcal{T}}
\newcommand{\cJ}{\mathcal{J}}

\DeclareSymbolFont{AMSa}{U}{msa}{m}{n}
\DeclareSymbolFont{AMSb}{U}{msb}{m}{n}
\DeclareMathSymbol{\fieldR}{\mathalpha}{AMSb}{"52}

\newcommand{\kahler}{{K\"ahler}\xspace}

\newcommand{\qk}{{quaternion-K\"ahler}\xspace}

\newcommand{\cZ}{\mathcal{Z}}
\newcommand{\cI}{\mathcal{I}}
\newcommand{\cO}{\mathcal{O}}

\newcommand{\cH}{\mathcal{H}}

\newcommand{\pa}{\partial}
\newcommand{\nn}{\nonumber}

\newcommand{\eps}{\epsilon}

\newcommand{\IT}{\mathds{T}}
\newcommand{\IR}{\mathds{R}}
\newcommand{\IC}{\mathds{C}}
\newcommand{\IZ}{\mathds{Z}}
\newcommand{\IQ}{\mathds{Q}}

\newcommand{\IP}{\mathds{P}}

\newcommand{\Nint}{\mathds{N}}
\newcommand{\Tr}{\mbox{Tr}}

\newcommand{\sgn}{\mbox{sgn}}

\newcommand{\q}{\mbox{q}}

\newcommand{\tc}{\tilde c}

\newcommand{\tPhi}{\tilde\Phi}
\newcommand{\vth}{\vartheta}

\newcommand{\tM}{M}

\def\bea{\begin{eqnarray}}
\def\eea{\end{eqnarray}}
\def\be{\begin{equation}}
\def\ee{\end{equation}}
\def\ba{\begin{align}}
\def\ea{\end{align}}
\def\bse{\begin{subequations}}
\def\ese{\end{subequations}}

\def\ba{\bar a}

\def\by{\bar y}

\def\bw{\bar w}

\def\bZ{\bar Z}

\def\hPhi{\hat\Phi}

\def\hM{\hat M}

\def\hV{\hat V}

\def\Hcl{H^{\rm cl}}
\def\Hclr{H^{\mbox{\scriptsize ref-cl}}}

\newcommand{\CL}{{\cal{L}}}

\newcommand{\cB}{\mathcal{B}}

\def\cij#1{c}
\def\ci#1{c}

\def\XXint#1#2#3{{\setbox0=\hbox{$#1{#2#3}{\int}$}
\vcenter{\hbox{$#2#3$}}\kern-.5\wd0}}

\def\Lv#1{L(#1)}
\def\Rv#1{R(#1)}

\def\bv{\bar v}

\newcommand{\expe}[1]{{\bf e}\!\left( #1\right)}

\def\gamD#1{\tilde\gamma}

\def\CY{\mathfrak{Y}}
\def\CYm{{\widehat{\mathfrak{Y}}}}

\def\cl0{\tilde c_0}

\newcommand{\bfu}{{\boldsymbol u}}
\newcommand{\bfb}{{\boldsymbol b}}
\newcommand{\bfc}{{\boldsymbol c}}
\newcommand{\bfv}{{\boldsymbol v}}

\newcommand{\bfp}{{\boldsymbol p}}
\newcommand{\bfq}{{\boldsymbol q}}

\newcommand{\bfx}{{\boldsymbol x}}
\newcommand{\bfy}{{\boldsymbol y}}

\newcommand{\bfmu}{{\boldsymbol \mu}}
\newcommand{\bfnu}{{\boldsymbol \nu}}

\newcommand{\bfchi}{{\boldsymbol \chi}}

\newcommand{\gammap}{\gamma}

\def\thetasn{\vartheta^{\rm \,SN}}

\def\bOm{\bar\Omega}

\def\cXsn{\cX^{\rm SN}}
\def\cXt{\cX^{\,\theta}}

\def\cXcl{\cX^{\rm cl}}

\def\Hcl{H^{\rm cl}}

\def\tcV{\tilde\cV}

\def\gf{g^{(0)}}
\def\whh{\widehat h}

\def\whg{\widehat g}

\def\whPhi{\widehat\Phi}

\def\Phip{\Phi^{(+)}}
\def\Phif{\Phi^{(0)}}

\def\tPhif{\tilde\Phi^{(0)}}

\def\gref{g^{\rm ref}}
\def\gfr{g^{\rm ref(0)}}
\def\cEf{\cE^{(0)}}
\def\cEr{\cE^{\rm ref}}
\def\cErf{\cE^{\rm ref(0)}}

\def\cJr{\cJ^{\rm ref}}
\def\cGr{\cG^{\rm ref}}
\def\cXr{\cX^{\rm ref}}
\def\Rr{R^{\rm ref}}
\def\Fref{F^{\rm ref}}
\def\Hr{H^{\rm ref}}
\def\hr{h^{\rm ref}}
\def\hrDT{h^{\rm DT,ref}}
\def\hVW{h^{\rm VW}}
\def\hrVW{h^{\rm VW,ref}}
\def\whhrVW{\widehat{h}^{\rm VW,ref}}
\def\whhr{\widehat h^{\rm ref}}
\def\cMr{\cM_V^{\rm ref}}
\def\whcZr{\widehat\cZ^{\rm ref}}
\def\whcZrp{\widehat\cZ'^{\rm ref}}
\def\vthref{\vartheta^{\rm ref}}
\def\cErp{\cE^{{\rm ref(+)}}}

\def\whcZ{\widehat\cZ}

\def\tPhi{\widetilde\Phi}
\def\tPhif{\widetilde\Phi^{(0)}}
\def\intPhi{\Phi^{\scriptscriptstyle\,\int}}

\def\trPhi{\Phi^{\rm tr}}

\def\cEPhi{\Phi^{\,\cE}}

\def\tcEPhi{\widetilde\Phi^{\,\cE}}

\def\cEp{\cE^{(+)}}

\def\bOmMSW{{\bar \Omega}^{\rm MSW}}

\def\Zv{\mathscr{Z}}

\def\Zv{\mathscr{Z}}

\def\Lat{\mathbf{\Lambda}}

\def\ptt{\mathfrak{p}}

\def\ver{\mathfrak{v}}

\newcommand{\gtr}{g_{{\rm tr},n}}

\def\cs{S}
\def\cl{c^{(\ell)}}
\def\vz{\mathbb{z}}
\def\vu{\mathbb{u}}

\def\vchi{\mathbb{\chi}}
\def\tbfv{\tilde\bfv}

\def\hbfc{\hat\bfc}
\def\hbfb{\hat\bfb}
\def\hbfv{\hat\bfv}

\def\gama{\check\gamma}

\def\fp{\mathscr{D}}

\def\Vop{\mathbb{V}}

\makeatletter
\DeclareRobustCommand{\cev}[1]{%
  {\mathpalette\do@cev{#1}}%
}
\newcommand{\do@cev}[2]{%
  \vbox{\offinterlineskip
    \sbox\z@{$\m@th#1 x$}%
    \ialign{##\cr
      \hidewidth\reflectbox{$\m@th#1\vec{}\mkern4mu$}\hidewidth\cr
      \noalign{\kern-\ht\z@}
      $\m@th#1#2$\cr
    }%
  }%
}
\makeatother

\def\under#1#2{\mathop{#1}\limits_{#2}}

\title{S-duality and  refined BPS indices}

\preprint{arXiv:1910.03098v2}

\author{Sergei Alexandrov$^{1}$, Jan Manschot$^{2,3}$ and Boris Pioline$^{4}$
\\
$^1$ {\it
Laboratoire Charles Coulomb (L2C), Universit\'e de Montpellier,
CNRS, F-34095, Montpellier, France}\\

$^2$ {\it School of Mathematics, Trinity College, Dublin 2, Ireland}\\

$^3$ {\it Hamilton Mathematical Institute, Trinity College, Dublin 2, Ireland}\\

$^4$ {\it Laboratoire de Physique Th\'eorique et Hautes
Energies (LPTHE), UMR 7589 CNRS-Sorbonne Universit\'e,
Campus Pierre et Marie Curie,
4 place Jussieu, F-75005 Paris, France} \\

\vspace*{2mm} {\tt e-mail:
\email{sergey.alexandrov@umontpellier.fr},
\email{manschot@maths.tcd.ie},
\email{pioline@lpthe.jussieu.fr}
}

\vspace*{-3mm}

}

\abstract{Whenever available, {\it refined} BPS indices provide considerably more information
on the spectrum of BPS states than their unrefined version.
Extending earlier work on the modularity of generalized Donaldson-Thomas invariants counting D4-D2-D0 brane bound states in type IIA strings on a Calabi-Yau threefold $\mathfrak{Y}$, we construct the modular completion of generating functions of refined BPS indices supported on a divisor class.
Although for compact $\mathfrak{Y}$ the refined indices are not protected, switching on the refinement
considerably simplifies the construction of the modular completion.
Furthermore, it leads to a non-commutative analogue of the TBA equations, which suggests a quantization of the moduli space consistent with S-duality. In contrast, for a local CY threefold given by
the total space of the
canonical bundle over a complex surface $S$, refined BPS indices are well-defined,
and equal to Vafa-Witten invariants of $S$.  Our construction provides a modular completion
of the generating function of these refined invariants  for arbitrary rank. In cases where
all reducible components of the divisor class are collinear
(which occurs e.g. when $b_2(\mathfrak{Y})=1$, or in the local case), we show that the holomorphic anomaly equation satisfied by  the completed generating function truncates at quadratic order.
In the local case, it agrees with  an earlier proposal by Minahan et al for unrefined invariants,
and extends it to the refined level using the afore-mentioned non-commutative structure.
Finally, we show that these general predictions reproduce known
results for $U(2)$ and $U(3)$ Vafa-Witten theory on $\mathbb{P}^2$, and
make them explicit for $U(4)$.
}

\begin{document}

\section{Introduction}
\label{sec-intro}

Explaining the microscopic origin of the Bekenstein-Hawking entropy of black holes is
one of the primary targets of any theory of quantum gravity. As shown in  \cite{Strominger:1996sh}
and many follow ups, string theory successfully achieves this goal in the case of supersymmetric black holes,
whose micro-states can be counted using D-brane techniques at weak coupling, and then reliably extended to strong coupling.
In string vacua with maximal or half-maximal supersymmetry,
such as type II strings compactified on $T^6$ or $K3\times T^2$,
the index $\Omega(\gamma)$ counting (with signs) micro-states of electromagnetic charge $\gamma$
is given by a Fourier coefficient of a suitable modular form, giving access to its asymptotics as
$|\gamma|\to\infty$ with arbitrary
precision, and allowing for detailed comparisons between microscopic
degeneracy and macroscopic entropy.

In contrast, for type IIA strings compactified on a generic Calabi-Yau (CY) threefold $\CY$, or type IIB
on the mirror $\CYm$, the indices $\Omega(\gamma)$ (defined mathematically as generalized Donaldson-Thomas (DT) invariants
of the derived category of coherent sheaves on $\CY$, or derived category of Lagrangian submanifolds of $\CYm$) are not known in general,
except for particular charge configurations.
One source of complication is that the indices depend on the \kahler moduli $z^a$ (or complex structure moduli of $\CYm$),
and include contributions from bound states of an arbitrary number $n$ of BPS black holes with charges
$\gamma_i$ such that $\sum_{i=1}^n \gamma_i=\gamma$ \cite{Denef:2000nb}.
These bound states are only stable in certain chambers of the \kahler moduli space $\cM_K(\CY)$,
and the index correspondingly jumps across the boundary of these domains,
leading to the so-called wall-crossing phenomenon, well-known both in the physics \cite{Denef:2007vg,Manschot:2010qz}
and  mathematics literature \cite{ks,Joyce:2008pc,Joyce:2009xv}. While for $\CY=K3\times T^2$ only bound states with
two constituents occur \cite{Dabholkar:2009dq, Dabholkar:2012nd}, for a generic threefold $\CY$ the number of
constituents may be arbitrary large (though only a finite number of
bound states occur for a given  charge $\gamma$). As a result, the duality group in $D=4$
(generated by  monodromies in $\cM_K$)  equates $\Omega(\gamma,z^a)$ to  $\Omega(h\cdot \gamma,h\cdot z^a)$,
but not necessarily to $\Omega(h\cdot \gamma,z^a)$, since the monodromy $z^a\to h\cdot z^a$ may encounter walls of marginal stability.
Thus, the constraints from duality are much weaker than for
string vacua with 16 supercharges or more.

Nevertheless, since BPS black holes in $D=4$ induce
instanton corrections to the metric on the vector multiplet moduli space $\cM_V$
after compactification on a circle, and
since  type IIA string theory compactified on $\CY$ is equivalent to M-theory compactified on $\CY\times T^2$,
it is clear that the indices $\Omega(\gamma,z^a)$ are strongly constrained by invariance under the large diffeomorphisms $SL(2,\IZ)$
of the M-theory torus (or equivalently, by S-duality in type IIB string theory) \cite{Alexandrov:2008gh}.
This is particularly so for black holes obtained by wrapping a D4-brane on a divisor $\cD \subset \CY$,
which lifts to an M5-brane wrapped around $\cD\times T^2$ in M-theory on $\CY\times T^2$:  invariance under $S$-duality is then essentially
equivalent to modular invariance of the superconformal field theory obtained by reducing the five-brane along $\cD$ \cite{Maldacena:1997de}.

\medskip

In a series of recent papers \cite{Alexandrov:2012au,Alexandrov:2016tnf,Alexandrov:2017qhn,Alexandrov:2018lgp},
building on related works \cite{Manschot:2009ia,Manschot:2010xp,Alexandrov:2016enp,Alexandrov:2018iao},
we have studied  the modular properties of the generating functions $h_{p,\mu}(\tau)$
of MSW invariants\footnote{MSW invariants are defined as the values of generalized DT invariants
$\bOm(\gamma,z^a)$ where $\gamma$ is supported on a divisor class, and the
\kahler moduli $z^a$ are evaluated at the  large volume attractor point $z_\infty(\gamma)$, see \eqref{defbOmMSW}. The
name `MSW invariants' was coined in \cite{Alexandrov:2012au} in reference to \cite{Maldacena:1997de},
but  the relevance of the large volume attractor chamber for modularity was later identified in
\cite{deBoer:2008fk,Andriyash:2008it,Manschot:2009ia}.
The rational DT invariants $\bOm(\gamma)$ are related to the integer-valued DT invariants $\Omega(\gamma)$ by \eqref{OmtoOmb} below
and are better suited for expressing the constraints of modularity \cite{Manschot:2010xp}.
\label{foo1}}
$\bOmMSW(\gamma):=\bOm(\gamma,z^a_\infty(\gamma))$ supported on a divisor class $[\cD]=p^a\omega_a$
(where $\omega_a$, $a=1,\dots, b_2(\CY)$, is a basis of $\Lambda=H_4(\CY,\IZ)$,
and $\mu\in\Lambda^*/\Lambda$ keeps track of the residual flux  after spectral flow,
see \eqref{decq} below). By imposing the existence of an isometric action of S-duality
on the instanton corrected moduli space $\cM_V$, we showed that $h_{p,\mu}(\tau)$ must
transform as a vector-valued {\it mock} modular form under fractional linear transformations of $\tau$.
More precisely, we derived a specific non-holomorphic modular completion
$\whh_{p,\mu}$, made out of all the holomorphic  functions $h_{p_i,\mu_i}$ with
$\sum_{i=1}^n p_i^a=p^a$ for any number $n$ of effective divisor classes $[\cD_i]=p_i^a\omega_a$ (see \eqref{exp-whh} below),
such that $\whh_{p,\mu}$ transforms as an (ordinary, but non-holomorphic) vector-valued modular form with
a specific multiplier system. When the divisor  class $[\cD]$ is irreducible, such that $n=1$, $\whh_{p,\mu}$ coincides
with  $h_{p,\mu}$, which is therefore an ordinary weakly holomorphic modular form, as anticipated in
\cite{Maldacena:1997de}. When the divisor  class $[\cD]$ decomposes into a sum of at most $n=2$
irreducible divisors, then the modular completion involves a sum of products
$ \widehat \Psi_{p_1,p_2,\mu_1,\mu_2} \, h_{p_1,\mu_1}\, h_{p_2,\mu_2}$ with $p=p_1+p_2$, where
$\Psi_{p_1,p_2,\mu_1,\mu_2}$ is the indefinite theta series constructed in \cite{Manschot:2009ia},
whose kernel can be written as an Eichler integral\footnote{
Recall that the Eichler integral
$\Phi(\tau)= \int_{-\bar\tau}^{\I\infty} \frac{\overline{F(-\bar u,\bar\tau)}\de u}{[-\I(u+ \tau)]^{2-w}}$
of  an analytic modular form $F(\tau,\bar\tau)$ of weight $(w,\bar w)$
transforms with modular weight $(2-w+\bar w,0)$ under $\tau\to (a\tau+b)/(c\tau+d)$,
up to a non-homogeneous term proportional to the period integral
$\int_{d/c}^{\I\infty} \frac{\overline{F(-\bar u,\bar\tau)}\de u}{[-\I(u+ \tau)]^{2-w}}$ (see e.g. \cite[(A.15)]{Alexandrov:2012au})}
of a Siegel-Narain theta series of signature $(1,b_2(\CY)-1)$ \cite{Alexandrov:2016tnf}.
The holomorphic generating function  $h_{p,\mu}$ is therefore an example of a mixed mock modular form \cite{Zwegers-thesis,MR2605321}.

In general \cite{Alexandrov:2018lgp}, the difference $\whh_{p,\mu}-h_{p,\mu}$ is a sum of products
$\widehat \Psi_{p_1,\dots p_n,\mu_1,\dots \mu_n}
\prod_{i=1}^n h_{p_i,\mu_i}$ with $p=p_1+\dots+p_n$,
where $\widehat \Psi_{p_1,\dots p_n,\mu_1,\dots \mu_n} $ is an indefinite theta series
whose kernel can be written as an $(n-1)$-times iterated Eichler integral of a Siegel-Narain
theta series of signature $(1,b_2(\CY)-1)$ \cite{Alexandrov:2016enp,bringmann2017higher, Manschot:2017xcr}.
As a result, the holomorphic generating function $h_{p,\mu}$ must transform non-homogeneously under
$SL(2,\IZ)$ like a vector-valued mock modular form of depth $n-1$ \cite{ZagierZwegersunpublished}.
Correspondingly, the modular completion $\whh_{p,\mu}$ satisfies a holomorphic anomaly equation,
sourced by a combination of the $\whh_{p_i,\mu_i}$'s with $\sum_i p_i^a=p^a$.
Importantly, the modular anomaly affects the growth of the Fourier coefficients
of $h_{p,\mu}$, hence would have consequences for a detailed comparison with the macroscopic entropy.

The construction of the modular completion $\whh_{p,\mu}$ in \cite{Alexandrov:2018lgp}
involved an interplay between different expansions, most notably
\begin{itemize}
\item
the multi-instanton expansion of the `instanton generating potential'\footnote{This
potential is closely related to the `contact potential' on $\cM_V$, or to the \kahler potential on its
twistor space \cite{Alexandrov:2008gh}, and therefore to the $\IR^3$ index studied in \cite{Alexandrov:2014wca}.}
$\cG$, a modular function of weight $(-\frac32,\frac12)$ on the vector multiplet
moduli space $\cM_V$ in $D=3$;
\item
the `attractor flow tree expansion' \cite{Manschot:2010xp, Alexandrov:2018iao} of the DT invariants $\bOm(\gamma,z^a)$ in terms of
$\bOmMSW(\gamma_i)$'s, obtained by iterating the
 primitive wall-crossing formula \cite{Denef:2007vg}
\be
\label{primwc}
\Delta \bOm(\gamma_L+\gamma_R) = (-1)^{\langle \gamma_L,\gamma_R\rangle+1} \,
\langle \gamma_L,\gamma_R\rangle\, \bOm(\gamma_L) \, \bOm(\gamma_R).
\ee
\end{itemize}
After summing over all possible types of trees, partitions etc.,
one eventually arrives at indefinite theta series with kernels expressed in terms of sums of multiple derivatives of generalized error functions,
where the sums and derivatives appear due to the presence of the Dirac product $\langle \gamma_L,\gamma_R\rangle$ in \eqref{primwc}.

One of the goals of the present work is to extend and streamline this construction by introducing a
refinement parameter $y$, drawing inspiration from earlier works on wall-crossing
and on the attractor flow tree formula \cite{Manschot:2010qz,Alexandrov:2018iao}.
At the most basic level, this involves replacing the factor $\langle \gamma_L,\gamma_R\rangle$ in \eqref{primwc}  by
the character
\be
\label{replacegLR}
\langle \gamma_L,\gamma_R\rangle \ \to \ \frac{y^{\langle \gamma_L,\gamma_R\rangle}-y^{-\langle \gamma_L,\gamma_R\rangle}}{y-1/y}\, ,
\ee
of the $SU(2)$ representation of spin $j=\frac12|\langle \gamma_L,\gamma_R\rangle|-1$.
The explicit powers of $y^{\pm \langle \gamma_L,\gamma_R\rangle}$ can be absorbed
in the parameters of the theta series, leading to simpler kernels expressed directly in terms
of generalized error functions, as opposed to sums of multiple derivatives thereof.
The price to pay is that
inverse powers of $y-1/y$ appear in front of individual contributions, and the
unrefined limit $y\to 1$ can only be taken after summing over all contributions.

While the introduction of the parameter $y$ allows to considerably simplify the construction of the
modular completion of the generating functions $h_{p,\mu}(\tau)$ of {\it unrefined} MSW invariants, it can also be used to obtain a natural
non-holomorphic completion  $\whhr_{p,\mu}(\tau,w)$ for the generating functions $\hr_{p,\mu}(\tau,w)$ of {\it refined}
invariants $\bOmMSW(\gamma,y)$ with $y=e^{2\pi\I w}$, in situations where a refined version of
$\bOmMSW(\gamma)$ can be defined. In such a situation, this construction
also suggests a natural generalization
$\cGr(w)$ of the instanton generating potential $\cG$, such that $\cGr(w)$ transforms with modular weight $(-\frac12,\frac12)$
whenever the completions $\whhr_{p,\mu}(\tau,w)$ transform as vector-valued Jacobi modular forms
of weight $(-\frac12 b_2(\CY),0)$ and suitable index $m(p)$.
 Remarkably, for a suitable choice of the index (more precisely, whenever $m(p)-\frac16 p^3$ is linear in $p^a$),
 $\cGr(w)$  has a simple representation \eqref{nonpert-cGr}
in terms of solution to  a certain integral equation \eqref{inteqH-star},
analogous to the TBA-like equations for Darboux coordinates on the twistor space of
$\cM_V$ \cite{Gaiotto:2008cd,Alexandrov:2008gh,Alexandrov:2009zh,Alexandrov:2010pp}. This equation involves a non-commutative star product
of functions on $\cM_V$ similar to the one which has appeared in the context of line operators \cite{Ito:2011ea,Hayashi:2019rpw}.
It would be interesting to understand the relation (if any) to other  refined versions of TBA-like equations that have already appeared in the
literature \cite{Gaiotto:2010be,Cecotti:2014wea}, and understand the  physical significance of $\cGr$.
Meanwhile, our result confirms the general expectation \cite{Nekrasov:2009rc,Gaiotto:2010be,Bullimore:2015lsa, Beem:2016cbd}
that the refinement (or equivalently, introducing an $\Omega$-background along two directions in $\IR^3$) effectively quantizes
the moduli space $\cM_V$ as well as its twistor space, and strongly suggests that this refinement preserves the action of S-duality,
which seems to be consistent with the analysis in \cite{Gaiotto:2019wcc}.

Unfortunately, for type II strings on a {\it compact}  CY threefold $\CY$, the refined
BPS invariants defined by $\Omega(\gamma,y)=\Tr_{\cH_\gamma} (-y)^{2J_3}$, where $J_3$ generates rotations around a fixed axis in $\IR^3$,
are emphatically not protected
by supersymmetry, in particular they may well depend both on the K\"ahler and complex structure of $\CY$. On the mathematical side,
the moduli spaces of semi-stable objects do not carry
any $\IC^\times$ action which would allow to refine the DT invariants in a deformation-invariant
way (although it may still be possible to define motivic invariants \cite{ks}).
Therefore the status of the refined construction in this setup is not clear.
In contrast, in the `local Calabi-Yau' case where  $\CY={\rm Tot}(K_S)$ is the total space
of the canonical
bundle over a smooth complex surface $S$ (leading to $\cN=2$ rigid supersymmetry in $D=4$),
an additional $SU(2)$ R-symmetry arises allowing to define
refined indices (or `protected spin character')
$\Omega(\gamma,y)=\Tr_{\cH_\gamma} (-1)^{2J_3} y^{2(J_3+I_3)}$
\cite{Gaiotto:2010be}, corresponding to the $\chi_{y^2}$-genus of the moduli space of stable objects.
In this case, the refined DT invariants counting D4-branes
wrapped $N$ times around the surface $S$ \cite{Gholampour:2013jfa} are expected to coincide
\cite{Minahan:1998vr,Alim:2010cf,gholampour2017localized} with
the refined Vafa-Witten (VW) invariants with gauge group $U(N)$
on $S$ \cite{Vafa:1994tf,gottsche2017virtual,thomas2018equivariant,Toda:2019opw}.
Moreover, for smooth projective surfaces with $b_1(S)=0$ and $b_2^+(S)>1$,
the generating function of refined VW invariants is conjectured to have precise modular properties~\cite{Vafa:1994tf, Gottsche:2017vxs, gottsche2018refined}.
In contrast, for $b_2^+(S)=1$, the generating function is expected
to be mock modular, as already evident for VW theory on $\IP^2$ in the rank 2 case \cite{Vafa:1994tf}.

For any Fano surface and arbitrary $N$, the refined generating functions
$\hrVW_{N,\mu}$ can be determined using a sequence of blow-ups and wall-crossing
transitions  \cite{Gottsche:1990, Yoshioka:1994, Yoshioka:1995,
 Manschot:2010nc} in terms of generalized
Appell-Lerch sums \cite{Manschot:2014cca}, which indicates that they should transform as
mock Jacobi forms of depth $N-1$. By viewing the generating functions of VW invariants
$\hrVW_{N,\mu}$ as
special instances of generating functions $\hr_{p,\mu}$  for local Calabi-Yau threefolds,
our formalism in principle predicts the precise modular completion
$\whhrVW_{N,\mu}$ for arbitrary rank $N$ and any smooth Fano surface with
$b_2^+(S)=1$ and $b_1(S)=0$.
We verify that this prescription precisely reproduces\footnote{As explained
in footnote \ref{foot-sign}, the agreement apparently requires a minus sign in the relation
between DT and  VW invariants. We do not yet understand the origin of this sign flip.}
the known modular completions for $S=\IP^2$, $N=2$ and $N=3$ obtained by different methods
in \cite{Bringmann:2010sd,Manschot:2017xcr}.
Importantly, the structure of the modular completion is
universal and does not require prior knowledge
of the VW invariants themselves. The knowledge of the
detailed modular properties could in principle be leveraged
to determine $\hrVW_{N,\mu}$ from the knowledge of its polar terms,
e.g. using a Rademacher sum \cite{Bringmann:2010sd, Bringmann:2018cov}.

More generally, the refined VW invariants are expected to be  well-defined
and deformation-invariant for any almost complex four-manifold
$S$ (this is because the twisting preserves 4 topological charges in this case \cite{Dijkgraaf:1997ce}).
We conjecture that for such manifolds, the generating function of
refined VW invariants (completed with non-holomorphic terms when $b_2^+(S)=1$)
transforms as a vector-valued Jacobi form, whose weight $w_S$ and index
$m_S(N)$ are given by
\be
\label{genindex0}
w_S=\frac14(\sigma(S)-\chi(S)),
\qquad
m_S(N)=-\frac{1}{6}(2N^3+N)\,\chi(S)-\frac{1}{2}N^3\sigma(S) ,
\ee
where $\chi(S)$ and $\sigma(S)$ are the Euler number and signature of $S$.

While mathematically attractive, our conjecture apparently presents a conundrum: as shown
in  \cite{Alexandrov:2018lgp},  the general
holomorphic anomaly equation satisfied  by the completion $\whh_{p,\mu}$ is sourced by
products  of generating functions $\whh_{p_i,\mu}$ for any decomposition $p=\sum_{i=1}^n p_i$,
whereas according to the conjecture in \cite{Minahan:1998vr}, the holomorphic anomaly equation for
the partition function $Z_N^{\rm VW}$ of rank $N$ VW invariants should only be sourced
by sums of products $Z_{N_1}^{\rm VW} \, Z_{N_2}^{\rm VW}$ with $N_1+N_2=N$.
As a crucial test of this conjecture, we show that in the special case where all magnetic charges
$p_i$ are collinear, which is pertinent for local CY threefolds, all contributions
to the holomorphic anomaly equation satisfied  by  $\whhr_{p_i,\mu}$ (and by its
unrefined counterpart $\whh_{p_i,\mu}$)
vanish whenever\footnote{This truncation to quadratic order is reminiscent
of the holomorphic anomaly equations for topological string amplitudes \cite{Bershadsky:1993ta}.
In fact, the two are related by T-duality for elliptically fibered threefolds \cite{Klemm:2012sx}.}
$n>2$. Assuming the validity of our conjecture, we establish the
general holomorphic anomaly equation for $Z_N^{\rm VW}$, and find
precise agreement with \cite{Minahan:1998vr}, which up until now had
only been tested for $\frac{1}{2}K3$, and for rational surfaces for $N\leq 3$.
Furthermore, we derive a refined version \eqref{holanomZpref} of this holomorphic anomaly equation, which involves
the same non-commutative star product which appeared in the TBA-like equations \eqref{inteqH-star}.
It would be interesting to relate it to the holomorphic anomaly equation for the Taylor coefficients
of the refined topological string amplitude near $w=0$ proposed in \cite[Eq.(8.16)]{Huang:2013yta}.

\medskip

The outline of this work is as follows. In \S\ref{sec-review} we review the definition of the generating functions of MSW invariants
and the construction in \cite{Alexandrov:2018lgp} of their modular completion.
In \S\ref{sec-ref} we simplify and generalize this construction to include the refinement parameter.
In \S\ref{sec-potential} we construct a `refined  instanton generating potential' $\cGr$,
which is modular invariant whenever the non-holomorphic completions
of the generating functions of refined MSW invariants are themselves modular.
In \S\ref{subsec-TBA}
we provide a non-perturbative
definition of this object, in terms of solutions to non-commutative TBA-like equations.
In \S\ref{sec_localCY}, we apply our construction to local Calabi-Yau threefolds, and conjecture
the precise form of the modular completion of the generating function of  VW invariants
for any rank. We check this prediction against results for $U(2)$ and $U(3)$ VW
theory on $\IP^2$. In \S\ref{sec-holan} we establish the holomorphic anomaly equation for
completed generating functions of refined invariants, and show that this equation truncates whenever
all magnetic charges lie in a one-dimensional lattice. Moreover, we establish a refined
version of the truncated holomorphic anomaly equations, based on the same non-commutative
structure used in \S\ref{subsec-TBA}, and show that it reduces to the conjecture  \cite{Minahan:1998vr}
in the unrefined case. \S\ref{sec-concl} is devoted to the discussion of our results.
Finally, a few appendices contain various useful information and details of our proofs.

\section{DT invariants, MSW invariants and modular completions}
\label{sec-review}

In this section  we review general  properties of  generalized DT invariants $\bOm(\gamma,z^a)$
supported on an effective divisor class $[\cD]=p^a\omega_a$
where $\omega_a$ denotes an integer basis of irreducible divisors in $ \Lambda=H_4(\CY,\IZ)$,
dual to the basis $\omega^a$ of curve classes in $\Lambda^*=H_2(\CY,\IZ)$.
In such cases the charge vector $\gamma\in H_{\rm even}(\CY,\IQ)$ has the form $\gamma=(0,p^a,q_a,q_0)$
where the entries corresponds to the D6, D4, D2 and D0-brane charges. Since the D6-brane
charge vanishes, the Dirac product on the charge lattice is given by $\langle \gamma_1,\gamma_2 \rangle = p_2^a q_{1,a} - p_1^a q_{2,a}$
and depends only on the  {\it reduced} charge vector $\gama=(p^a,q_a) \in H_4(\CY,\IQ)\oplus H_2(\CY,\IQ)$.
The charges satisfy
the following quantization conditions \cite{Alexandrov:2010ca}:
\be
\label{fractionalshiftsD5}
p^a\in\IZ ,
\qquad
q_a \in \IZ  + \frac12 \,(p^2)_a  ,
\qquad
q_0\in \IZ-\frac{1}{24}\,c_{2,a} p^a.
\ee
Here $c_{2,a}$ are components of the second Chern class of $\CY$ and we introduced convenient notations
$(kp)_a=\kappa_{abc}k^b p^c$ and $(lkp)=\kappa_{abc}l^a k^b p^c$ where $\kappa_{abc}$ are the intersection numbers on $H_4(\CY,\IZ)$.
The condition that $\cD$ is an effective divisor in $\CY$ and
belongs to the K\"ahler cone means that
\be
\label{khcone}
p^3> 0,
\qquad
(r p^2)> 0,
\qquad
k_a p^a > 0,
\ee
for all  divisor classes $r^a \omega_a \in H_4(\CY,\IZ)$ with $r^a>0$, and
curve classes $k_a \gamma^a \in H_2^+(\CY,\IZ)$ with $k_a>0$.
The charge $p^a$ induces a quadratic form $\kappa_{ab}=\kappa_{abc} p^c$ on $\Lambda\otimes \IR$ of signature $(1,b_2(\CY)-1)$.
This quadratic form allows to embed $\Lambda$ into $\Lambda^*$, but the map $\epsilon^a \mapsto \kappa_{ab} \epsilon^b$
is in general not surjective, the quotient $\Lambda^*/\Lambda$ being a finite group of order $|\det\kappa_{ab}|$.

For primitive charge vector $\gamma$, the rational
DT invariants
\be
\label{OmtoOmb}
\bOm(\gamma,z^a):=
\sum_{m|\gamma} \frac{1}{m^2} \Omega(\gamma/m,z^a)
\ee
coincide with the generalized DT invariants
$\Omega(\gamma,z^a)$, which are defined schematically
as the Euler number\footnote{More precisely, $(-1)^{d_{\IC}}$ times the Euler number, where
$d_{\IC}$ is the complex dimension of $\cM_{\gamma,z^a}$.}
of the moduli space $\cM_{\gamma,z^a}$ of semi-stable coherent sheaves on $\cD$, for the stability condition determined by
the complexified \kahler moduli $z^a=b^a+\I t^a \in \cM_K(\CY)$.
When the vector $\gamma$ is not primitive, the moduli space $\cM_{\gamma,z^a}$ is singular
and its Euler number can be determined using intersection cohomology
in favourable cases \cite{Meinhardt:2017, Manschot:2016gsx}.
Due to the Bogomolov bound, DT invariants are known to vanish for
$\hat q_0 \geq \hat q_0^{\rm max}=\frac{1}{24}\chi(\cD)=\tfrac{1}{24}(p^3+c_{2,a}p^a)$ where
\be
\label{defqhat}
\hat q_0 \equiv
q_0 -\frac12\, \kappa^{ab} q_a q_b
\ee
and $\kappa^{ab}$ is the inverse of the quadratic form $\kappa_{ab}$.
The generating series
\be
\label{defchimu}
h^{\rm DT}_{p,q}(\tau,z^a) = \sum_{\hat q_0 \leq \hat q_0^{\rm max}}
\bOm(\gamma,z^a)\,\expe{-\hat q_0 \tau },
\ee
where  $\expe{x}:=e^{2\pi\I x}$,
defines  a holomorphic function on the Poincar\'e upper half-plane $\tau\in \mathbb{H}$,
locally constant as a function of the \kahler moduli $z^a\in \cM_K$. However, due to wall-crossing phenomena,
this function is not expected to possess any simple modular properties.

Instead, to eliminate the dependence on the moduli $z^a$, it is natural to consider generating functions
of MSW invariants, defined as  generalized DT invariants evaluated at their respective large volume attractor point,
\be
\bOmMSW(\gamma)=\bOm(\gamma,z^a_\infty(\gamma)),
\qquad
z^a_\infty(\gamma)= \lim_{\lambda\to +\infty}\(-q^a+\I\lambda  p^a\).
\label{defbOmMSW}
\ee
They are invariant under spectral flow transformations acting on the D2 and D0 charges via
\be
\label{flow}
q_a \mapsto q_a - \kappa_{abc}p^b\epsilon^c,
\qquad
q_0 \mapsto q_0 - \epsilon^a q_a + \frac12\, (p\epsilon \epsilon)
\ee
with $\epsilon^a\in\Lambda$, which leave the charge $\hat q_0$ \eqref{defqhat} invariant.
As a result, they only depend on $\hat q_{0}$, $p^a$ and $\mu_a\in \Lambda^*/\Lambda$,
the residue class of $q_a$ modulo shifts \eqref{flow}, defined by
\be
\label{decq}
q_a=\mu_a+  \kappa_{ab}\eps^b+\hf\, \kappa_{abc} p^b p^c.
\ee
Thus, we can write $\bOmMSW(\gamma)=\bOm_{p,\mu}( \hat q_0)$ and the generating series
\be
h_{p,\mu}(\tau) =\sum_{\hat q_0 \leq \hat q_0^{\rm max}}
\bOm_{p,\mu}(\hat q_0)\,\expe{-\hat q_0 \tau }
\label{defhDT}
\ee
defines a vector-valued holomorphic function of $\tau$, with components labelled by elements of the finite group $\Lambda^*/\Lambda$.

\medskip

In \cite{Alexandrov:2012au, Alexandrov:2016tnf,Alexandrov:2017qhn,Alexandrov:2018lgp},
by postulating the existence of an isometric action of S-duality
on  the vector multiplet moduli space $\cM_V$ in $D=3$, it was shown that the generating functions
$h_{p,\mu}$ must satisfy precise modular properties, which depend on the reducibility
of the divisor class $[\cD]$. If $[\cD]$ is an {\it irreducible} divisor class, then $h_{p,\mu}$
should be a vector-valued modular form of weight $-\frac12\, b_2(\CY)-1$,
transforming as \eqref{ST}
as anticipated in \cite{Maldacena:1997de,Gaiotto:2006wm,deBoer:2006vg,Denef:2007vg}.
If instead $[\cD]$ is a sum of $n\geq 2$ {\it irreducible} effective divisors $[\cD_i]$,
then $h_{p,\mu}$  is a mock modular form of depth $n-1$, which means
that it admits a non-holomorphic modular completion $\whh_{p,\mu}$,
which can be expressed in terms of $n-1$ iterated Eichler integrals involving the generating functions
$h_{p_i,\mu_i}$ of the constituents. Specifically, $\widehat h_{p,\mu}$ should transform
as\footnote{In \cite{Alexandrov:2016tnf}, based on earlier work \cite{Alexandrov:2010ca,Alexandrov:2012au} we found
that $\whh_{p,\mu}$ must transform
with multiplier system $M_\eta^{c_2\cdot p} \times \overline M_{\theta}$, where $M_\eta$
is the multiplier system of the Dedekind eta function, and $M_{\theta}$ is that of the Siegel-Narain
theta series for the lattice $\Lambda$, given in \cite[Eq.(2.4)]{Alexandrov:2016enp} with $n=b_2(\CY)$, $\lambda=-1$.
Eq. \eqref{ST} follows by combining these two observations, and ensures that
the partition function and the instanton generating potential defined below in \eqref{defZp} and \eqref{treeFh-flh},
respectively, transform as modular forms of weight $(-\frac32,\frac12)$ and trivial multiplier system.}
\be
\label{ST}
\begin{split}
\whh_{p,\mu}(-1/\tau)=&\, -\frac{(-\I\tau)^{-\frac{b_2(\CY)}{2}-1}}{\sqrt{|\Lambda^*/\Lambda|}}\,
\expe{-\tfrac{1}{4}\, p^3 -\tfrac18\, c_{2,a}p^a }
\sum_{\nu \in \Lambda^*/\Lambda}
\expe{-\mu \cdot \nu}\whh_{p,\nu}(\tau),
\\
\whh_{p,\mu}(\tau+1)=&\, \expe{\tfrac{1}{24}\, c_{2,a}p^a +
\tfrac{1}{2}(\mu+\tfrac12 p)^2} \whh_{p,\mu}(\tau).
\end{split}
\ee

\lfig{An example of Schr\"oder tree contributing to $R_8$. Near each vertex we showed the corresponding factor
using the shorthand notation $\gamma_{i+j}=\gamma_i+\gamma_j$.}
{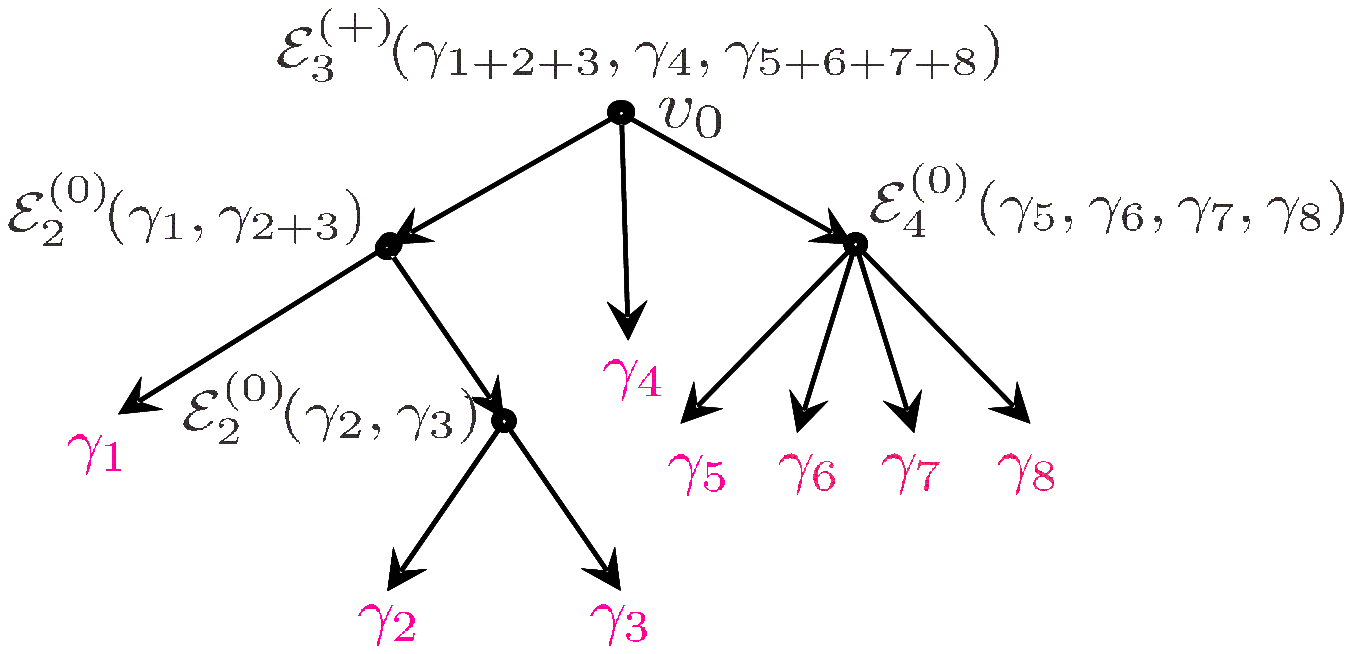}{9.75cm}{fig-Rtree}{-1.2cm}

The explicit form of the completion  for a generic reducible divisor was found to be \cite[Eq.(5.1)]{Alexandrov:2018lgp}
\be
\whh_{p,\mu}(\tau)= \sum_{n=1}^\infty
\sum_{\sum_{i=1}^n \gama_i=\gama}
R_n(\{\gama_i\},\tau_2)
\, e^{\pi\I \tau Q_n(\{\gama_i\})}
\prod_{i=1}^n h_{p_i,\mu_i}(\tau).
\label{exp-whh}
\ee
Here $Q_n$ denotes the quadratic form, originating  from the quadratic term in the definition \eqref{defqhat} of
the invariant charge $\hat q_0$,
\be
Q_n(\{\gama_i\})= \kappa^{ab}q_a q_b-\sum_{i=1}^n\kappa_i^{ab}q_{i,a} q_{i,b} \, ,
\label{defQlr}
\ee
where $\kappa_i^{ab}$ is the inverse of $\kappa_{i,ab}=\kappa_{abc}p_i^c$.
The sum in \eqref{exp-whh} runs over all ordered decompositions of $\gama$ (for any fixed choice of
$q_a$ consistent with the residue class $\mu_a$, for example $q_a=\mu_a+\frac12 (p^2)_a$)
into a sum of reduced charge vectors $\gama_i=(p_i^a,q_{i,a}) \in H_4^+(\CY) \oplus H_2(\CY)$,
with first components satisfying \eqref{khcone} and
with arbitrary residue classes $\mu_i$.
It produces a theta series of signature $(n-1) (1,b_2(\CY)-1)$, whose kernel is determined by the coefficient $R_n$.
This coefficient is in turn given by
\be
R_n(\{\gama_i\},\tau_2)= \Sym\left\{\sum_{T\in\IT_n^{\rm S}}(-1)^{n_T-1} \cEp_{v_0}\prod_{v\in V_T\setminus{\{v_0\}}}\cEf_{v}\right\},
\label{solRn}
\ee
where $\Sym$ denotes symmetrization (with weight $1/n!$) with respect to the charges $\gama_i$.
Here the sum goes over so-called Schr\"oder trees with $n$ leaves (see Figure \ref{fig-Rtree}), i.e. rooted planar
trees such that all vertices $v\in V_T$ (the set of vertices of $T$ excluding the leaves) have $k_v\geq 2$ children,
$n_T$ is the number of elements in $V_T$, and $v_0$ labels the root vertex.
The vertices of $T$ are labelled by charges so that the leaves carry charges $\gama_i$, whereas the charges assigned to other vertices
are given recursively by
the sum of charges of their children, $\gama_v\in\sum_{v'\in\Ch(v)}\gama_{v'}$.
Finally, the functions $\cEf_v$ and $\cEp_v$ are determined from a set of functions
$\cE_n(\{\gama_i\},\tau_2)$ depending on $n$ charges as follows.
Namely, each function has a {canonical} decomposition
\be
\cE_n(\{\gama_i\},\tau_2)=\cEf_n(\{\gama_i\})+\cEp_n(\{\gama_i\},\tau_2),
\label{twocEs}
\ee
where the first term $\cEf_n$ does not depend on $\tau_2$,
whereas the second term $\cEp_n$ is exponentially suppressed as $\tau_2\to\infty$ keeping
the charges $\gama_i$ fixed. Then, given a Schr\"oder tree $T$,
we set $\cE_{v}\equiv \cE_{k_v}(\{\gama_{v'}\})$ (and similarly for $\cEf_{v}, \cEp_{v}$)
where $v'\in \Ch(v)$ runs over the $k_v$ children of the vertex $v$. The functions $\cE_n$
are defined as suitable combinations of derivatives of the generalized error functions introduced in
\cite{Alexandrov:2016enp,Nazaroglu:2016lmr}. Since the concrete form of $\cE_n$ will not be important
in this paper, we relegate their definition to appendix \ref{ap-E}.

\medskip

In the following we will also need the expansion  of the generating function of DT invariants \eqref{defchimu}
in terms of the completion \eqref{exp-whh}, given by \cite[Eq.(5.2)]{Alexandrov:2018lgp}
\be
\begin{split}
h^{\rm DT}_{p,q}(\tau,z^a)=&\,
\sum_{n=1}^\infty \sum_{\sum_{i=1}^n \gama_i=\gama}
\whg_n(\{\gama_i,c_i\}, \tau_2) \,e^{\pi\I \tau Q_n(\{\gama_i\})}
\prod_{i=1}^n \whh_{p_i,\mu_i}(\tau).
\end{split}
\label{multihd-full}
\ee
Here, the sum runs over the same range as in \eqref{exp-whh}, but
 the coefficients are expressed through $\cE_n(\{\gama_i\},\tau_2)$ and another set of functions $\gf_{n}(\{\gama_i,c_i\})$,
 also specified in appendix \ref{ap-E},
as follows:
\be
\whg_n(\{\gama_i,c_i\}, \tau_2) = \Sym\left\{\sum_{T\in\IT_n^{\rm S}}(-1)^{n_T-1}
\(\gf_{v_0}-\cE_{v_0}\)\prod_{v\in V_T\setminus{\{v_0\}}}\cE_{v}\right\},
\label{soliterg}
\ee
where $\gf_{v}= \gf_{k_v}(\{\gama_{v'},c_{v'}\})$ similarly to the definition of  $\cE_v$ in terms of $\cE_n$ above.
The stability parameters $c_i$ are defined in terms of the reduced charges $\gama_i$ and moduli $z^a$ by the standard formula
\be
c_i(z^a) =  \Im\bigl[ Z_{\gamma_i}\bZ_\gamma(z^a)\bigr] ,
\label{fiparam}
\ee
where $Z_\gamma(z^a)$ is the central charge function {and $\gamma=\sum_{i=1}^n \gamma_i$.
Note that $c_i(z)$ is independent of the D0-brane charge in the large volume limit, and thus depends only on $\gama_i$.}
The relation \eqref{multihd-full} allows to recover the DT invariants from MSW ones: it is suffices to drop all
$\tau_2$-dependent terms on the right-hand side, which is {tantamount} to replacing $\cE_n$ by $\cEf_n$.
In this way one obtains
\be
\label{Omsumtree}
\bOm(\gamma,z^a) =
\sum_{\sum_{i=1}^n \gamma_i=\gamma}
\gtr(\{\gamma_i,c_i\})\,
\prod_{i=1}^n \bOmMSW(\gamma_i),
\ee
where the coefficients
are given by
\be
\begin{split}
\gtr(\{\gamma_i,c_i\})
=&\, \Sym\left\{\sum_{T\in\IT_n^{\rm S}}(-1)^{n_T-1} \(\gf_{v_0}-\cEf_{v_0}\)\prod_{v\in V_T\setminus{\{v_0\}}}\cEf_{v}\right\}.
\end{split}
\label{soliterg-tr}
\ee
The formula \eqref{Omsumtree} is a special case of the `flow tree formula' introduced in \cite{Alexandrov:2018iao},
and gives an alternative way of computing the coefficients
$\gtr$, which {were} called `tree indices' in \cite{Alexandrov:2018iao}.

Finally, it is also important to note that the function $\cEf_n$ is a special case of $\gf_n$, evaluated
at the large volume attractor point,  namely
\be
\cEf_n(\{\gama_i\})=\gf_n(\{\gama_i,\beta_{ni}\}),
\label{rel-gE}
\ee
where $\beta_{k\ell}=\sum_{i=1}^k \gamma_{i\ell}$ and
$\gamma_{ij}=\langle\gama_i,\gama_j\rangle$. This relation shows that, as explained in appendix \ref{ap-E},
the knowledge of the functions $\gf_n$ uniquely determines both the coefficients $R_n$ in the
modular completion \eqref{exp-whh}, and the tree index, since the exponentially decaying contribution $\cEp_n$ is
determined from $\cEf_n$ by the requirement that the sum $\cE_n$
(after a suitable rescaling) is a smooth solution of Vign\'eras equation \eqref{Vigdif},
which {ensures that the theta series with kernel $\cE_n$ is modular}.

\section{Refined construction}
\label{sec-ref}

In this section, we refine the previous construction, by introducing a deformation parameter $y$
which physically may be thought {of} as a fugacity conjugate to the physical angular momentum $J_3$
carried by BPS states in $D=4$ dimensions (and therefore also by BPS instantons in $D=3$).

\subsection{Refined BPS indices}

A natural way to refine generalized DT invariants is to replace the Euler number of the moduli space
$\cM_{\gamma,z^a}$ of semi-stable coherent sheaves with charge vector $\gamma$
by the Poincar\'e polynomial (more precisely, Laurent polynomial)\footnote{Note the following
properties of the Poincar\'e-Laurent polynomial:
$P(\cM,-1)=(-1)^{d_\IC(\cM)} \chi(\cM):=\tilde\chi(\cM)$;
if $\cM$ is compact, then $P(\cM,y)=P(\cM,1/y)$ by Poincar\'e duality. If $\cM$ is compact and \kahler,
then $P(\cM,y)$ is the character for the $SU(2)$-Lefschetz action on $H^*(\cM)$.
When the Dolbeault cohomology of $\cM$ is supported in degree $(p,p)$ only, then
$P(\cM,-y)$ coincides with the $\chi_{y^2}$-genus $\chi_{y^2}(\cM)=\sum_{p,q} (-1)^{p+q-d} y^{2p-d} h_{p,q}(\cM)$.
\label{fooP}}
\be
\label{defOmref}
\Omega(\gamma,z^a,y) =  P(\cM_{\gamma,z^a},-y),
\qquad
P(\cM,y):=\sum_{p=0}^{2d_{\IC}(\cM)} y^{p-d_{\IC}(\cM)}\, b_p(\cM),
\ee
where $d_\IC(\cM)$ is the complex dimension of $\cM$.
As in the unrefined case \eqref{OmtoOmb}, it is advantageous to define
the rational invariant
\be
\label{defbOm}
\bOm(\gamma,z^a,y) = \sum_{m|\gamma} \frac{y-1/y}{m(y^m-1/y^m)}\, \Omega(\gamma/m, z^a,y^m) ,
\ee
where $m$ runs over all positive integers such that $\gamma/m$ is in the charge lattice. If
all the moduli spaces $\cM_{\gamma/m,z^a}$ are compact (or at least compactifiable), then $\bOm(\gamma,z^a,y)$
is a rational function of $y$, invariant under $y\to 1/y$. The advantage of \eqref{defbOm} compared
to \eqref{defOmref} is that the wall-crossing formula for these invariants, and hence the
flow tree formula,
takes a much simpler form \cite{Manschot:2010xp,Alexandrov:2018iao}.

\subsection{Refined modular completions}

In \cite[Section 5.4.2]{Alexandrov:2018lgp}
a natural refinement of the function $\gf_n$, called $\gref_n$, was introduced,
satisfying various consistency conditions.
In particular, using \eqref{soliterg-tr} and \eqref{rel-gE},
it reproduces the refined tree index defined in \cite{Alexandrov:2018iao},
and it reduces  to $\gf_n$ in the limit $y\to 1$.
Explicitly \cite[Eqs.(5.49-51)]{Alexandrov:2018lgp},
\be
\gref_n(\{\gama_i,c_i\},y)
= \frac{(-1)^{\sum_{i<j} \gamma_{ij} }}{(y-y^{-1})^{n-1}}
\,\Sym\Bigl\{
\Fref_n(\{\gama_i,c_i\})\,y^{\sum_{i<j} \gamma_{ij}}\Bigr\},
\label{whgF}
\ee
where $\Fref_n$ is defined by a sum over ordered partitions of $n$,
\be
\Fref_n(\{\gama_i,c_i\})=\frac{(-1)^{n-1}}{2^{n-1}}\sum_{n_1+\cdots +n_m= n\atop n_k\ge1}
\prod_{k=1}^m  b_{n_k}\prod_{k=1}^{m-1}\sgn(\cs_{j_k}).
\label{defFref}
\ee
The variables $\cs_j=\sum_{i=1}^j c_i,$ are linear combinations of stability parameters defined in \eqref{fiparam},
$j_k=n_1+\cdots + n_k$, while the $b_k$'s are  the Taylor coefficients of $\tanh t=\sum_{k\geq 1} b_k \, t^k$.

\medskip

Upon using the representation \eqref{whgF}, the dependence on the refinement parameter $y$ is mainly due to multiplication
by a power of $(y-y^{-1})$ and insertion of the factor $y^{\sum_{i<j} \gamma_{ij}}$ into the sum over D2-brane charges.
The former can be absorbed into a redefinition of the generating functions \eqref{defchimu}, \eqref{defhDT},
\bea
\label{defchimur}
\hrDT_{p,q}(\tau,z^a,w) &=& \sum_{\hat q_0 \leq \hat q_0^{\rm max}}
\frac{\bOm(\gamma,z^a,y)}{y-y^{-1}}\,\expe{-\hat q_0 \tau },
\\
\hr_{p,\mu}(\tau,w) &=& \sum_{\hat q_0 \leq \hat q_0^{\rm max}}
\frac{\bOm_{p,\mu}(\hat q_0,y)}{y-y^{-1}}\,\expe{-\hat q_0 \tau },
\label{defhDTr}
\eea
where we recall that $y=\expe{w}$.
As for  the factor $y^{\sum_{i<j} \gamma_{ij}}$, we show in appendix \ref{ap-theta} that it
can be absorbed into a shift of  the elliptic parameter
of the indefinite theta series implementing the sum over D2-brane charges.
As a result, if one starts from a modular theta series, it continues to transform as a (vector-valued) modular form after the refinement,
provided the argument of its kernel, depending on the  imaginary part of the elliptic parameter,
involves the same shift (see \eqref{thetaref}) and
$w=\alpha-\tau \beta$ transforms as an elliptic variable, $w\to w/(c\tau+d)$.

\medskip

The above observations suggest a simple way to include the refinement while preserving modularity
of all relevant theta series, and in turn, a simple  procedure for constructing a modular
completion of the generating function \eqref{defhDTr} of the {\it refined} MSW invariants:
\begin{enumerate}
\item
Define $\cErf_n$  from $\gref_n$  by a formula analogous to \eqref{rel-gE};
\item
Complete $\cErf_n$ into a smooth solution of Vign\'eras equation \eqref{Vigdif} as explained in appendix \ref{ap-E},
page \pageref{procfun} in the unrefined case and let $\cEr_n$ be that solution evaluated at
the argument $\bfx$ shifted by the refinement (cf. \eqref{argx})
\be
x_i^a=\sqrt{2\tau_2}(\kappa_i^{ab} q_{i,b}+ b^a+\beta\ptt_i^a),
\label{argxref}
\ee
where $b^a=\Re z^a$ is the Kalb-Ramond field and $\ptt$ is introduced in \eqref{vecptt};
\item
Define $\cErp_n$ to be the difference between $\cEr_n$ and $\cErf_n$, as in \eqref{twocEs}.
Importantly, these functions are still exponentially suppressed at large $\tau_2$ provided one keeps fixed $w$ and $\bw$:
since $\beta=\frac{\bw-w}{2\I\tau_2}$, the additional shift appearing in \eqref{argxref} disappears and
$\cEr_n$ reduce to $\cErf_n$, similarly to the unrefined case\footnote{It might seem more natural
to do the shift induced by the refinement also in $\cErf_n$ and $\gref_n$.
However, it would make the limit $y\to 1$ discontinuous and
lead to a $\beta$-dependence in the position of walls of marginal stability in the moduli space.
Furthermore, we will see that the definitions given here allow to get a nice integral equation for refined Darboux coordinates
on the twistor space of $\cM_V$ and agree with known results in VW theory.};
\item
Obtain the modular completion $\whh_{p,\mu}$ by equations analogous to \eqref{exp-whh}, \eqref{solRn}.
\end{enumerate}
Although this prescription is not derived from any known S-duality  constraint on the moduli space
$\cM_V$, it is a mathematically natural generalization of the construction in \cite{Alexandrov:2018lgp}. Moreover, in
\S\ref{sec-potential} we shall see that it ensures that
a certain potential $\cGr$ constructed out of the refined DT invariants transforms with modular
weight $(-\frac12,\frac12)$, as expected from the refined elliptic genus of the superconformal field theory on the wrapped 5-brane,
after factoring out the contribution $\tau_2^{1/2}$ from the translational zero-mode.

\medskip

Following the steps  above, one finds that all contributions corresponding to partitions
in \eqref{defFref} other than the trivial partition $n_k=1$ for all $k$  {remarkably}
cancel after summing over Schr\"oder trees.\footnote{This is similar to the cancellation
of contributions
of marked trees noted in the remark in the end of appendix \ref{ap-E}.
The proof of this cancellation is completely analogous to the proof of proposition 10 in \cite{Alexandrov:2018lgp}
and we omit it here.}
Therefore, we can further simplify the construction and take as a starting point
the function
\be
\gfr_n(\{\gama_i,c_i\},y)=
\frac{(-y)^{\sum_{i<j} \gamma_{ij} }}{2^{n-1}}\prod_{k=1}^{n-1}\sgn(-\cs_{k}).
\label{rel-gF}
\ee
Although the function $(y-y^{-1})^{1-n}\gfr_n$ does not have a smooth $y\to 1$ limit
due to omitted symmetrization and contributions of non-trivial partitions,
one can check that it leads to the same completion as the more complicated version based on
\eqref{defFref}. Starting from \eqref{rel-gF},
using \eqref{rel-gE} and promoting products of sign functions to smooth solutions of Vign\'eras equation \eqref{Vigdif},
we arrive at the definitions
\bea
\cErf_n(\{\gama_i\},y)&=&
\frac{(-y)^{\sum_{i<j} \gamma_{ij} }}{2^{n-1}}\prod_{k=1}^{n-1}\sgn(\Gamma_{k}),
\label{rel-gEf}
\\
\cEr_n(\{\gama_i\},y)&=& \frac{(-y)^{\sum_{i<j} \gamma_{ij} }}{2^{n-1}}\Phi^E_{n-1}(\{ \bfv_{\ell}\};\bfx),
\label{Erefsim}
\eea
where (the vectors $\bfv_{ij}$ are defined explicitly in \eqref{defvij})
\be
\Gamma_{k}=\sum_{i=1}^k\sum_{j=k+1}^n \gamma_{ij},
\qquad
\bfv_k=\sum_{i=1}^k\sum_{j=k+1}^n\bfv_{ij},
\label{defbfvk}
\ee
while $\Phi^E_{n-1}$ are (boosted) generalized error functions described in appendix \ref{ap-generr}
whose argument $\bfx$ has components \eqref{argxref}.

In terms of these functions, we claim that
the modular completion of the refined generating function \eqref{defhDTr}
 is given by
\be
\whhr_{p,\mu}(\tau,w)= \sum_{n=1}^\infty
\sum_{\sum_{i=1}^n \gama_i=\gama}
\Rr_n(\{\gama_i\},\tau_2,y)
\, e^{\pi\I \tau Q_n(\{\gama_i\})}
\prod_{i=1}^n \hr_{p_i,\mu_i}(\tau,w),
\label{exp-whhr}
\ee
where $\Rr_n$ is defined by a formula similar to \eqref{solRn},
\be
\Rr_n(\{\gama_i\},\tau_2,y)= \Sym\left\{\sum_{T\in\IT_n^{\rm S}}(-1)^{n_T-1} \cErp_{v_0}\prod_{v\in V_T\setminus{\{v_0\}}}\cErf_{v}\right\}.
\label{solRnr}
\ee
More specifically, we propose that the completion defined in \eqref{exp-whhr} transforms as a vector-valued Jacobi modular form
of weight $(-\frac12 b_2,0)$ and suitable index $m(p)$, namely
\bea
\label{STref}
\whhr_{p,\mu}\left(-\frac{1}{\tau},\frac{w}{\tau}\right)
&=&-\I\, \frac{(-\I\tau)^{-\frac{b_2}{2}}}{\sqrt{|\Lambda^*/\Lambda|}}\,
\expe{-\tfrac{1}{4}\, p^3 -\tfrac18\, c_{2,a}p^a + m(p)\, \frac{w^2}{\tau}  }
\sum_{\nu \in \Lambda^*/\Lambda}
\expe{-\mu \cdot \nu}\whhr_{p,\nu}(\tau,w),
\nn\\
\whhr_{p,\mu}(\tau+1,w)&=&\expe{ \tfrac{1}{24}\, c_{2,a} p^a +
\tfrac{1}{2}(\mu+\tfrac12 p)^2} \, \whhr_{p,\mu}(\tau,w),
\\
\whhr_{p,\mu}(\tau,w+k\tau+\ell)&=&
\expe{- m(p)\, \left( k^2\tau + 2 k w}\right)\whhr_{p,\mu}(\tau,w).
\nn
\eea
Note that these transformations imply \eqref{ST} provided
\be
\label{hrlim}
\whh_{p,\mu}(\tau) = \lim_{y\to 1} (y-1/y) \,
\whhr_{p,\mu}(\tau,w),
\ee
consistently with the relation between the holomorphic generating functions \eqref{defhDT} and \eqref{defhDTr}.
In \S\ref{sec-potential} below, we shall give support to this proposal, and find strong indications that the index $m(p)$ should be equal
to $-\frac16 p^3$ up to a linear term in $p^a$ which we do not know yet how to fix in the compact
case. For non-compact CY threefolds, we shall show that the modular completion \eqref{exp-whhr}
reproduces known results for VW invariants on $\IP^2$, closely related to Donaldson-Thomas invariants on local $\IP^2$, and make
a specific proposal for the index in \eqref{mp2} below.

\subsubsection*{A subtle point: the case of vanishing Dirac products}

For particular sets of charges $\gama_i$, the Dirac products $\Gamma_k$ \eqref{defbfvk} may vanish
so that, to find the modular completion $\whhr_{p,\mu}$, one needs to evaluate the function \eqref{rel-gEf} on its loci of discontinuity.
This problem does not arise in the unrefined case where the sign functions always come together with the factors of $\gamma_{ij}$
in front (see e.g. \eqref{defDf-gen})  ensuring continuity (but not smoothness) on this locus.
Naively, one could apply the standard prescription $\sign(0)=0$ as in \eqref{defsign}, leading to $\cErf_n(\{\gama_i\},y)=0$
for any set $\{\gama_i\}$ for which at least one $\Gamma_k$ vanishes.
However,  the functions $\Rr_n$ so defined would not have\footnote{This problem can be traced back to
the sign identities like
$$
(\sgn(x_1)+\sgn(x_2))\,\sgn(x_1+x_2)=1+\sgn(x_1)\,\sgn(x_2),
$$
on which the construction of completion in \cite{Alexandrov:2018lgp} heavily relies.
Whereas this identity holds if either $x_1$ or $x_2$ vanishes upon using $\sign(0)=0$,
it is violated when two of them vanish simultaneously.}
 a zero of order $n-1$ at $y=1$,
which is necessary to cancel the higher order poles in \eqref{exp-whhr} and produce a simple pole
consistent with \eqref{hrlim}.
Thus the above naive extension of the definition of $\cErf_n$ to the discontinuity loci is incorrect and must be  modified.

To solve this problem, we observe that the naive prescription for $\cErf_n$
spoils the exponential fall-off of $\cErp_n$ at large $\tau_2$
because the generalized error functions with even rank do not vanish at $\bfx=0$.
This suggests a simple remedy: {\it define}
\be
\cErf_n(\{\gama_i\},y):=\under{\lim}{\tau_2\to\infty}\cEr_n(\{\gama_i\},y),
\label{redef-cErf}
\ee
such that the exponential suppression is automatic.
For configurations where all $\Gamma_k\ne 0$, this definition is automatically consistent with \eqref{rel-gEf}.
The issue arises only when some of them vanish.
In particular, applying \eqref{redef-cErf} to the configurations where all charges $\gama_i$ are equal
(which contribute when $(p,\mu)$ is not primitive) so that all $\Gamma_k=0$,
one finds that the product of $m$ sign functions with vanishing arguments
must be replaced by
\be
e_{m}=\Phi^E_{m}(\{ \bfv_{\ell}\};0).
\label{en-PhiEn}
\ee
The function on the right-hand side coincides with the generalized error function
$E_{m}(\cM^{(0)};\vu)$ evaluated at $\vu=0$ and $\cM^{(0)}$ the matrix of parameters computed in \eqref{matM0}, with $n=m+1$.
For $m$ odd, it vanishes because the generalized error function is odd in $\vu$,
but for $m$ even it gives a non-vanishing constant $e_{m}$.

Unfortunately,  it appears to be non-trivial to compute $e_m$ directly using \eqref{en-PhiEn} for generic $m$.
To determine these coefficients, we resort to another consistency condition --- namely,
by requiring smoothness of the completion in the limit $y\to 1$.
To this end, note that \eqref{hrlim} relies
on the regularity of the function \eqref{whgF} at $y=1$, which in turn is ensured by the identity \cite{Alexandrov:2018lgp}
\be
\Sym\Fref_n(\{c_i\})=0
\label{eqB}
\ee
satisfied by the function \eqref{defFref}. The latter is also constructed from products of  sign functions and
exhibits the same ambiguity as $\cErf_n$.
Therefore, a natural guess is that the correct value $e_m$ must be such
that the identity \eqref{eqB} continues to hold even at these loci.
In particular, choosing the locus where all arguments vanish, we see that
the coefficients $e_m$ should satisfy
\be
f_n:= \sum_{n_1+\cdots +n_m= n\atop n_k\ge1}
e_{m-1}\prod_{k=1}^m  b_{n_k}=0,
\qquad n>1.
\ee
To resolve this constraint, we set $f_1=e_0=1$ and construct the generating function
\be
\sum_{n=1}^\infty f_n \, t^n = \sum_{m=1}^{\infty} e_{m-1} \left( \sum_{i=1}^{\infty} b_i t^i \right)^m =
\sum_{m=1}^\infty  e_{m-1} \, (\tanh t )^m
\ee
Clearly, all $f_n$ with $n>1$ vanish and this generating function reduces to $t$,
if  $e_{m-1}$ is the $m$-th Taylor coefficient of ${\rm arctanh}$, namely
\be
e_m=\left\{ \begin{array}{cc}
0 & \mbox{ if }m \mbox{ is odd},
\\
\frac{1}{m+1}& \mbox{ if }m \mbox{ is even}.
\end{array}\right.
\label{valek}
\ee
Comparing with \eqref{en-PhiEn}, this leads us to conjecture  that
the following non-trivial identity must hold,
\be
\label{EnM0}
E_{m}(\cM^{(0)};0)=\frac{1}{m+1},
\ee
where $\cM^{(0)}$ is given in \eqref{matM0}. For $m=2$, this equation holds true thanks to the identity \eqref{E2zero}
with $\alpha=1/\sqrt{3}$, and we have checked numerically that it holds also for $m=4,6$, giving strong
evidence for this conjecture.

Applying the definition \eqref{redef-cErf}, we conclude that  \eqref{rel-gEf}  must be modified as follows,
\be
\cErf_n(\{\gama_i\},y)=
\frac{e_{|\cI|}}{2^{n-1}}\, (-y)^{\sum_{i<j} \gamma_{ij} }\prod_{k\in \Zv_{n-1}\setminus \cI}\sgn(\Gamma_{k}),
\quad \mbox{where }\cI\subseteq \Zv_{n-1}:\
\left\{ \begin{array}{ll}
\Gamma_k=0 & \mbox{ for } k\in \cI,
\\
\Gamma_k\ne 0 & \mbox{ for } k\notin \cI,
\end{array}\right.
\label{rel-gEf-zero}
\ee
where $\Zv_{n}=\{1,\dots,n\}$ and $|\cI|$ is the cardinality of the set.
With this prescription,  the coefficients $(y-y^{-1})^{1-n}\Rr_n$ now
reduce to $R_n$ in the limit $y\to 1$, while $\hr_{p,\mu}$ and
$\whhr_{p,\mu}$ multiplied by $y-y^{-1}$ reduce to $h_{p,\mu}$ and $\whh_{p,\mu}$, respectively.
In particular, whereas each term on the {r.h.s. of \eqref{exp-whhr}} has a pole of order
$n$ as $y\to 1$, the l.h.s. has only a simple pole, and all higher order poles cancel after summing over Schr\"oder trees.
Furthermore, the exponential fall-off of $\cErp_n$ at large $\tau_2$ ensures the same property for
$\Rr_n$ and hence in this limit $\whhr_{p,\mu}$ reduces to $\hr_{p,\mu}$, as appropriate for the modular completion of a mock modular form.

\subsection{Refined instanton generating potential}
\label{sec-potential}

The starting point of the derivation of the modular completion \eqref{exp-whh} in \cite{Alexandrov:2018lgp} was the existence
of the instanton generating potential $\cG$ transforming as a modular form of weight $(-\frac32,\frac12)$
and the knowledge of its {\it non-perturbative} definition
\be
\cG= \sum_{\gamma\in\Gamma_+}\bOm(\gamma)\int_{\ell_{\gammap}} \de z\, H_{\gamma}(z)
 -\hf\,\sum_{\gamma_1,\gamma_2\in\Gamma_+}  \bOm(\gamma_1)\,  \bOm(\gamma_2)\,
\int_{\ell_{\gamma_1}}\de z_1\int_{\ell_{\gamma_2}} \de z_2
K_{\gamma_1\gamma_2}(z_1,z_2)\,H_{\gammap_1}(z_1)H_{\gammap_2}(z_2)
\label{defcF2}
\ee
in terms of solutions of a system of TBA-like equations
\be
H_\gamma(z)=\Hcl_\gamma(z) \, \exp\[\sum_{\gamma'\in\Gamma_+}
 \bOm(\gamma) \int_{\ell_{\gamma'}}\de z'\, K_{\gamma\gamma'}(z,z')\,H_{\gamma'}(z')\].
\label{expcX}
\ee
Here\footnote{We deviate from the notation in \cite{Alexandrov:2018lgp},
where $\bOm(\gamma)$
was included in the definition of $H_\gamma(z)$, which led to shorter formulae at the cost of breaking
the multiplicativity property $H_{\gamma_1} H_{\gamma_2} =
\frac{(-1)^{\langle \gamma_1,\gamma_2\rangle}}{2\pi^2}\,  H_{\gamma_1+\gamma_2}$.}
 $H_\gamma(z)= \frac{\sigma_\gamma}{(2\pi)^2}\cX_\gamma(z)$
is proportional to the Fourier mode $\cX_\gamma(z)$ playing the role of the holomorphic Darboux coordinate on the twistor space over
the \qk moduli space $\cM_V$, whereas $\Hcl_\gamma(z)$ is the same function
with $\cX_\gamma$ replaced by its classical limit
\be
\cXcl_\gamma(z) =e^{-S^{\rm cl}_p}\,
\expe{- \hat q_0\tau- \frac{\tau}{2}\,(q+b)^2+c^a (q_a +\haf (pb)_a)+\I \tau_2 \((pt^2)z^2+2\I z t^a (q_a+(pb)_a)\)}.
\label{Xtheta}
\ee
Here $S^{\rm cl}_p=\pi\tau_2(pt^2) - 2\pi \I  p^a \tc_a$
is the modular invariant leading part of the Euclidean D3-brane action in the large volume limit $t^a\to\infty$,
and $b^a$, $c^a$, $\tc_a$ are periods of the Neveu-Schwarz $B$-field
and Ramond-Ramond fields along bases of 2 and 4-cycles.\footnote{In our conventions,
$b^a$ is also the real part of the complexified K\"ahler moduli, $z^a=b^a+\I t^a$.}
Besides, in the above equations, $\sigma_\gamma$ is a quadratic refinement of the DSZ pairing,
$K_{\gamma\gamma'}(z,z')=2\pi\((tp_1p_2)+\frac{\I\langle\gamma_1,\gamma_2\rangle}{z-z'}\)$ is a meromorphic function
with a simple pole at $z=z'$, and $\ell_\gamma$ is a suitable contour known as BPS ray \cite[\S3.2]{Alexandrov:2018lgp}.
Upon solving \eqref{expcX} iteratively, plugging into \eqref{defcF2}, expanding in powers of DT invariants
and performing a theta series decomposition expressing them in terms of MSW invariants,
it was shown that $\cG$ takes the following form
\be
\cG=\sum_{n=1}^\infty\frac{2^{-\frac{n}{2}}}{\pi\sqrt{2\tau_2}}\[\prod_{i=1}^{n}
\sum_{p_i,\mu_i}\sigma_{p_i}\whh_{p_i,\mu_i}\]
e^{-S^{\rm cl}_p}V_\bfp\,\vartheta_{\bfp,\bfmu}\bigl(\whPhi^{{\rm tot}}_{n},n-2\bigr),
\label{treeFh-flh}
\ee
where
$V_\bfp$ is the factor \eqref{Jacfac} cancelling the  non-vanishing index of the theta series,
$\sigma_p$ is a phase factor related to the quadratic refinement through
$\sigma_\gamma =\expe{\hf\, p^a q_a}\sigma_p$,
and $\vartheta_{\bfp,\bfmu}\bigl(\whPhi^{{\rm tot}}_{n},n-2\bigr)$ is the theta series
\eqref{Vignerasth} with the kernel $\whPhi^{{\rm tot}}_{n}$  \cite[Eq.(5.39)]{Alexandrov:2018lgp} solving Vign\'eras equation
\eqref{Vigdif} with parameter $\lambda=n-2$.
Since such theta series is a vector valued modular form of weight $(n(1+\tfrac12 b_2(\CY))-2,0)$,
the modularity of $\whh_{p,\mu}$ follows from the modularity of $\cG$. In fact,
each term in the sum over $n$ and $p_i$ in \eqref{treeFh-flh} transforms separately as
a modular form, starting with the first term
\be
\label{Gexp}
\cG = \frac{1}{4\pi^2 \sqrt{2\tau_2 (p t^2)}} \sum_p \whcZ_p + \cdots\, ,
\ee
where
\be
\label{defZp}
\whcZ_{p} = \sigma_p\, e^{-S^{\rm cl}_p} \,V_p
\sum_{\mu\in\Lambda^*/\Lambda}
\whh_{p,\mu}\, \thetasn_{p,\mu}
\ee
and $ \thetasn_{p,\mu}$ is the Siegel-Narain theta series
\be
\label{defthSN}
\thetasn_{p,\mu}(\tau,v)  =    \sum_{q\in\Lambda+ \mu+\hf p^2}
(-1)^{p^aq_a} \, \expe{ -\frac{\tau}{2} \,q^2 +\I \tau_2 (q+b)_+^2 + q_a v^a},
\ee
where $v^a=c^a-\tau b^a$ and the subscript $+$ denotes the projection
$q_+=q_a t^a/{\sqrt{(pt^2)}}$ along the K\"ahler form.
The latter is known to transform as a vector-valued Jacobi form of weight
$(\frac{b_2(\CY)-1}{2},\frac12)$,
so that \eqref{defZp} transforms as a multivariate Jacobi form of weight $(-\frac32,\frac12)$,
and may be identified with the modified elliptic genus of the superconformal field theory obtained by
wrapping the M5-brane along the divisor $\cD$ \cite{Maldacena:1997de,Alexandrov:2016tnf}.
In this context, the non-holomorphic terms in $\whh_{p,\mu}$ are expected to come
from a spectral asymmetry in the continuous spectrum of the conformal field theory, similarly to
the analysis in \cite{Troost:2010ud}.
The explicit form of the $n=2$ contribution can be found in \cite{Alexandrov:2016tnf}.

\medskip

In this subsection we generalize this construction to the refined case, reversing the logic of this derivation.
Namely, we start with \eqref{treeFh-flh} and replace $\whh_{p_i,\mu_i}$ and
$V_\bfp\,\vartheta_{\bfp,\bfmu}\bigl(\whPhi^{{\rm tot}}_{n},n-2\bigr)$ by their
natural refinements, such that the resulting function $\cGr(w)$ remains modular. In the next subsection,
we shall elucidate its non-perturbative origin in terms of a suitable system of TBA-like equations,
which may encode a refinement of the vector multiplet moduli space $\cMr$.

\medskip

The kernel $\whPhi^{{\rm tot}}_{n}$ appearing in \eqref{treeFh-flh} is constructed from the functions $\tcEPhi_n$ and $\cEPhi_n$,
which are introduced in appendix \ref{ap-E} and uniquely determined by $\gf_n$. Therefore,
it is natural to expect that its refined counterpart $\whPhi^{{\rm ref}}_n$ {should be}
given by the same formula, but with $\tcEPhi_n$ and $\cEPhi_n$
replaced by the functions determined by $\gfr_n$ along the same procedure detailed on page \pageref{procfun}.
Thus, one obtains
\be
\whPhi^{{\rm ref}}_n= \intPhi_1\sum_{T\in\IT_n^{\rm S}}(-1)^{n_T-1}
\(\tPhi_{v_0}-\Phi_{v_0}\)\prod_{v\in V_T\setminus{\{v_0\}}}\Phi_{v},
\label{kertotsolr}
\ee
where, upon denoting by $x^a=\kappa^{ab} \sum_i \kappa_{bcd} \, p_i^c x^d_i$ the sum of the entries in
$\bfx$,
\be
\intPhi_1(x)=\frac{e^{-\frac{\pi(pxt)^2}{(pt^2)}}}{2\pi\sqrt{2\tau_2(pt^2)}}, \quad
\label{defPhi1-main}
\ee
whereas $\Phi_{v}$ and $\tPhi_{v}$ are defined in the usual way from
\be
\Phi_n(\bfx)=\frac{1}{2^{n-1}}\Phi^E_{n-1}(\{ \bfv_{\ell}\};\bfx),
\qquad
\tPhi_n(\bfx)=\frac{1}{2^{n-1}}\Phi^E_{n-1}(\{ \bfu_{\ell}\};\bfx).
\label{defPhiref}
\ee
Here $\bfv_{\ell}$ are the vectors introduced in \eqref{defbfvk} and $\bfu_{\ell}$ are defined similarly
from the moduli-dependent vectors $\bfu_{ij}$ \eqref{defvij},
\be
\bfu_k=\sum_{i=1}^k\sum_{j=k+1}^n\bfu_{ij}.
\label{shiftS}
\ee
The kernel $\whPhi^{{\rm ref}}_n$ {satisfies} Vign\'eras equation \eqref{Vigdif} with parameter $\lambda=-1$ and
therefore {leads to a modular theta series}
of weight $\tfrac12 n b_2(\CY)-1$.
Thus, we arrive at the following refined version of the instanton generating potential\footnote{Comparing to \eqref{treeFh-flh},
we multiplied each term by $2^{\frac{n-1}{2}}$. This factor takes into account the change in the $\lambda$-parameter
of the theta series and the mismatch between the prefactors in \eqref{Vignerasth} and \eqref{rescEnPhi}. }
\be
\cGr(w)=\frac{1}{2\pi\sqrt{\tau_2}}\sum_{n=1}^\infty\[\prod_{i=1}^{n}
\sum_{p_i,\mu_i}\sigma_{p_i}\, y^{-\beta m(p_i)}\,\whhr_{p_i,\mu_i}\]
e^{-S^{\rm cl}_p}\hV_\bfp\,\vthref_{\bfp,\bfmu}\bigl(\whPhi^{\rm ref}_{n},-1\bigr),
\label{treeFh-flhr}
\ee
where $\vthref_{\bfp,\bfmu}$ is the theta series \eqref{thetaref},
\be
\hV_\bfp=\expe{\hf\, \hbfv\cdot\hbfb}=\expe{\haf\,\beta w\,\ptt^2} V_\bfp\, ,
\label{hJacfac}
\ee
and $m(p)$ is a function of $p^a$ which will be fixed momentarily.
Assuming that the non-holomorphic completion $\whhr_{p,\mu}$ is a vector-valued
Jacobi form of weight $- \frac12 b_2(\CY)$ and index $m(p)$, the function
\eqref{treeFh-flhr} will then be a Jacobi form of weight $(-\hf,\hf)$ and index 0,
whose residue at $y=1$ will coincide with the unrefined instanton generating potential $\cG$.
Indeed, the factor $y^{-\beta m(p_i)}$  cancels the phase factor $\expe{ m(p_i)\, \frac{c\, w^2}{c\tau+d}}$ appearing
in the modular transformation of  $\whhr_{p_i,\mu_i}$, while  the factor  $\hV_\bfp$ cancels
a similar factor $\expe{  -\hf\, \frac{c\, \hbfv^2}{c\tau+d}}$  appearing in the transformation of the theta series.

Evaluating the norm of the vector $\ptt$ \eqref{vecptt}
\be
\ptt^2=\sum_{i=1}^{n-1}\sum_{j=1}^{i}\sum_{k=1}^{i+1}(p_jp_{i+1}p_k),
\label{normp}
\ee
we see that all $y$-dependent factors in \eqref{treeFh-flhr} can be naturally combined into
a single factor $y^{-\beta m(p)}$, which depends only on the total charge $p^a=\sum_i p_i^a$,
provided $m(p)$ is chosen as
\be
\label{indexconjDT}
m(p)=-\frac16 \,p^3+\rho_a p^a,
\ee
where $\rho_a$ is an arbitrary constant vector.\footnote{A natural candidate for $\rho_a$ is a multiple of the second
Chern class $c_{2,a}$, but in the context of local surfaces we shall find an additional contribution in \eqref{mp2} which is not of this form.}
Furthermore, for this choice of $m(p)$ (and only then) it is possible,
starting from the theta series decomposition \eqref{treeFh-flhr},
to perform the manipulations done in \cite{Alexandrov:2018lgp} in reverse.
Some steps of this procedure are reported in appendix \ref{ap-G}.
As a result, one arrives at the following expansion in terms of iterated integrals
\be
\cGr=
\frac{1}{2\pi\I}
\sum_{n=1}^\infty
\[\prod_{i=1}^{n} \sum_{\gamma_i\in \Gamma_+}\frac{\I\sigma_{\gamma_i}}{2\pi}\,
\frac{\bOm(\gamma_i,y)\, y^{-\beta m(p_i)}}{y-y^{-1}}
\int_{\ell_{\gamma_i}}\de z_i\, \cXcl_{\gamma_i}(z_i) \]
\frac{y^{\frac{\beta}{2}\,\ptt^2+\sum_{i<j} \gamma_{ij}} (y\by)^{-\I\sum_{i<j} (p_ip_jt)(z_i-z_j)}}{\prod_{i=1}^{n-1}(z_i-z_{i+1})},
\label{treeFh}
\ee
where $\cXcl_{\gamma}$ is the same classical Fourier mode as in \eqref{Xtheta}.

\subsection{Star product and TBA-like equations}
\label{subsec-TBA}

{It is natural to expect that a {\it non-perturbative} representation
similar to \eqref{defcF2}--\eqref{expcX}
should exist for $\cGr$ as well.}
Technically we need  to {devise} an integral equation whose {iterative}
 solution generates the expansion \eqref{treeFh}.
The main obstacle on this way are the $y$-dependent factors in the numerator.
On the other hand, typically switching on a $y$-dependent deformation corresponds to making some of the structures non-commutative.
For instance, in \cite{Gaiotto:2010be,Cecotti:2010fi,Cecotti:2014wea} the Darboux coordinates $\cX_\gamma$ become
non-commutative operators, whereas in \cite{Ito:2011ea,Hayashi:2019rpw} a non-commutative star product arises on the moduli space.

Inspired by these constructions, we now suggest a simple way to generate the expansion \eqref{treeFh} which in a sense combines
the above mentioned non-commutative structures.
Let us define the following {\it modular invariant} star product
\be
f \star g =  f \exp\[ \frac{1}{2\pi\I}\( \overleftarrow{\fp}_{\!\!a}\overrightarrow{\p}_{\!\tc_a}-
\overleftarrow{\p}_{\!\tc_a} \overrightarrow{\fp}_{\!\! a}\) \] g,
\label{starproduct-alt}
\ee
where
\be
\fp_a=\alpha\p_{c^a}+\beta\p_{b^a}=w\p_{v^a}+\bw\p_{\bv^a}
\ee
with $v^a$ as in \eqref{defthSN}. The modular invariance follows from the fact that $w$ and $v^a$ transform as
elliptic variables, whereas $\p_{\tc_a}$ is modular invariant.
The crucial observation is that the classical {Fourier modes} \eqref{Xtheta} satisfy
\be
\cXcl_{\gamma_1}(z_1)\star\cXcl_{\gamma_2}(z_2)=
y^{\gamma_{12}+\frac{\beta}{2}\, (p_1p_2 (p_1+p_2))}
(y\by)^{-\I(p_1p_2t)(z_1-z_2)}\cXcl_{\gamma_1}(z_1)\cXcl_{\gamma_2}(z_2).
\label{starXX}
\ee
One immediately recognizes the same $y$-dependent factors which appear in the last ratio in \eqref{treeFh} for the case $n=2$.
Similarly, taking into account \eqref{normp},
it is straightforward to check that the star product of $n$ classical Darboux coordinates
reproduces the $y$-dependent factors in this ratio for generic $n$.

Next, we introduce {\it refined} Fourier modes $\cXr_\gamma$ which are defined by an integral equation involving
the above introduced star product,
\be
\Hr_{\gamma}(z)=\Hclr_{\gamma}(z)\star\[1+2\pi\I\sum_{\gamma'} \bOm(\gamma,y)\,
\int_{\ell_{\gamma'}}\frac{\de z'}{z-z'}\Hr_{\gamma'}(z')\],
\label{inteqH-star}
\ee
where
\be
\Hr_\gamma(z)= \frac{1}{(2\pi)^2}\,\frac{\sigma_\gamma y^{-\beta m(p)}}{y-y^{-1}}\, \cXr_\gamma(z)
\label{prepHnewr}
\ee
and $\Hclr_{\gamma}$ is defined as in \eqref{prepHnewr} with the refined Fourier mode replaced
by its classical counterpart $\cXcl_\gamma$.
Finally, it is easy to see that the expansion \eqref{treeFh} is generated by solving \eqref{inteqH-star} {iteratively}
and substituting this solution into the following simple non-perturbative definition of
the refined instanton generating potential,
\be
\cGr=\sum_{\gamma\in \Gamma_+}  \bOm(\gamma,y) \int_{\ell_{\gamma}}\de z\,\Hr_{\gamma}(z).
\label{nonpert-cGr}
\ee
Amazingly, this non-perturbative definition appears even simpler than in the unrefined case where \eqref{defcF2}
involves a  combination of a single and double integrals along BPS rays.

Note that the integral equation \eqref{inteqH-star} is quite different from the equations
describing the refined case which can be found in the literature \cite{Gaiotto:2010be,Cecotti:2010fi,Cecotti:2014wea}.
The main difference is that usually the integral appears exponentiated as in \eqref{expcX}, whereas here it enters linearly.
However, the price to pay for this simplification, and the reason for the above discrepancy, is that the refined Fourier modes
$\cXr_\gamma$ do not have a smooth limit $y\to 1$.
In particular, it is a non-trivial problem to extract from them the unrefined modes $\cX_\gamma$.
Also their relation to the quantum Darboux coordinates appearing in \cite{Gaiotto:2010be,Cecotti:2010fi,Cecotti:2014wea}
is not clear and will be investigated elsewhere.
Nevertheless, the above construction strongly suggests that the refinement
effectively quantizes the moduli space $\cM_V$ along with its twistor space, while preserving the isometric action of S-duality.
It would be interesting to understand better the physical significance
of this quantization and of its  invariance under S-duality.

\section{Local Calabi-Yau and Vafa-Witten invariants
\label{sec_localCY}}

Although the construction of the previous section is very natural mathematically and passes several consistency checks
including the existence of a refined instanton generating potential, its significance is shadowed by the fact that
for a compact Calabi-Yau threefold, there appears to be no natural,
deformation invariant notion of refined  BPS indices.
In such situation the refinement appears just as a device for simplifying the construction
of the modular completion of generating functions of `ordinary' BPS indices.
In contrast, for toric (hence non-compact) Calabi-Yau threefolds,
BPS indices can be refined in the presence of an $SU(2)$ R-symmetry.
When $\CY$ is the total space ${\rm Tot}(K_S)$ of the canonical bundle over a projective surface $S$,
the BPS indices of $N$ D-branes supported on the divisor $[S]$ are expected
to be equal to the Vafa-Witten invariants of $S$ for gauge group $U(N)$
\cite{Minahan:1998vr,Alim:2010cf,Gholampour:2013jfa,gholampour2017localized},
both at the unrefined and refined levels.\footnote{If
  $h^{1,0}(S)=h^{2,0}(S)=0$, the VW invariants agree with local DT
  invariants, while refined VW invariants agree with the K-theoretic DT
  invariants defined using a $\mathbb{C}^\times$ action \cite{thomas2018equivariant}. However if $h^{2,0}>0$, the
  relation between definitions for refined DT invariants and refined
  VW invariants remains unclear. In fact, the
  numerical DT invariants vanish for $h^{2,0}>0$, while this is not
  the case for VW invariants. We thank Richard Thomas for correspondence on this issue.}
Indeed, the latter determine the partition function of (topologically twisted) $\cN=4$
super-Yang-Mills (SYM) theory, which describes the world-volume dynamics of
$N$ D4-branes wrapped on $S$.
When $b_2^+(S)=1$, these invariants depend on a choice of polarization of $S$
(except when $b_2(S)=1$), reflecting the moduli dependence
of the BPS indices. S-duality of $\cN=4$ SYM implies that  the generating functions
of VW invariants should be  modular, though only after including suitable non-holomorphic contributions
from reducible connections  \cite{Vafa:1994tf}.
In this section, we shall propose a general prescription for determining
the non-holomorphic completion of the generating function of {\it refined} VW invariants
for any rank $N$, by taking the decompactification limit of the completion for the compact case.
For $N=2,3$, we shall see that this prescription precisely reproduces earlier results for $\IP^2$ from \cite{Vafa:1994tf,Manschot:2017xcr},
giving strong evidence that it may be valid for any $N$ and for a large class of surfaces with $b_2^+(S)=1$.

\subsection{Vafa-Witten invariants}

Recall that Vafa-Witten theory arises as one of three possible topological twists of $\cN=4$ super-Yang-Mills theory,
which is available on an arbitrary smooth compact Riemannian 4-manifold $S$ \cite{Vafa:1994tf}.
Motivated by applications to local Calabi-Yau geometries,
we shall assume that $S$ is a connected
almost Fano  surface with $b_1(S)=0$ and $b_2^+(S)=1$; since $b_2^+(S)=1+2 h_{2,0}(S)$
for any complex surface, this implies that $h_{2,0}(S)=0$.
We further assume that the gauge group is $U(N)$, and that $S$ is
equipped with a polarization $J\in H^2(S,\IR)$.
In this case, the path integral localizes on solutions of hermitian Yang-Mills equations  for the
field strength $F$ \cite{Vafa:1994tf,Dijkgraaf:1997ce}.
Solutions are classified by $c_1(F)=-\mu \in \Lambda_S:= H^2(S,\IZ)$ and
$\int_S c_2(F)= n\in \IZ$, and span a moduli space $\cM_{N,\mu,n,J}$ of expected complex dimension
\be
\label{dimM}
d_\IC (\cM_{N,\mu,n,J}) = 2N^2\Delta(F) - (N^2-1) \chi(\cO_S),
\qquad
\Delta(F):= \frac{1}{N} \left( n - \frac{N-1}{2N} \mu^2 \right),
\ee
isomorphic to the moduli space of Gieseker\footnote{Gieseker stability is a finer notion than
slope-stability; the latter depends only on the rank $N$ and first Chern class $\mu$, whereas
the former also depend on $n$.}  semi-stable torsion-free coherent
sheaves with respect to the polarization $J$. Here, $\Delta(F)$ is known as the Bogomolov discriminant,
$\mu^2:=\int_S \mu^2$ is the intersection pairing on $\Lambda_S$, and
$\chi(\cO_S)=1-h_{0,1}(S)+h_{0,2}(S)$ is the holomorphic Euler characteristic, equal to 1 for
the surfaces under consideration.

The moduli space $\cM_{N,\mu,n,J}$ is invariant upon tensoring $F$ with a line bundle $\cL$,
under which  the first Chern class shifts as $\mu\to \mu-N c_1(\cL)$ while  $\Delta(F)$ and $N$ stay invariant.
The parameter $\mu$ (known as the 't Hooft flux) can therefore be restricted to
$\Lambda_S/ N \Lambda_S$, which is isomorphic to $\IZ_N^{b_2(S)}$ since the lattice $\Lambda_S$ is unimodular.
As a result, the partition function of twisted $\cN=4$ Yang-Mills theory
has a theta series decomposition
\be
\label{defZVW}
Z_{N,J}^{\rm VW}(\tau,v) = \sum_{\mu\in \Lambda_S/ N \Lambda_S}
\hVW_{N,\mu,J}(\tau)\, \vartheta^{\rm VW}_{N,\mu,J}(\tau,v),
\ee
where\footnote{The exponential prefactor is necessary to match our conventions for
theta series in Appendix \ref{ap-thetamod}. Without this prefactor, the theta series would
possess both holomorphic and anti-holomorphic index.
When $b_2(S)=1$, the series \eqref{defthetaVW} becomes antiholomorphic (up to the
same exponential prefactor), and $Z_{N,J}^{\rm VW}$ is the complex conjugate of
a  skew-holomorphic Jacobi form \cite{Cheng:2017dlj}.}
\be
\label{defthetaVW}
\vartheta^{\rm VW}_{N,\mu,J}(\tau,v) = e^{-2\pi \frac{(\Im v_+)^2}{\tau_2}}\!\!\!\!\!\!
\sum_{k\in H^2(S,N\IZ) + \mu + \frac{N}{2} K_S}\!\!
(-1)^{K_S\cdot k}\,
\expe{-\frac{\tau\, k_-^2}{2N} -  \frac{\bar\tau\, k_+^2}{2N}  +
v\cdot k_- + \bar v\cdot k_+}
\ee
with $k_+ = \frac{(k\cdot J) J}{J\cdot J}$ and $k_-=k-k_+$. Furthermore, $K_S=-c_1(S)$ is the canonical class of $S$
and $v$ is a chemical potential conjugate to (minus) the first Chern class. Ignoring contributions from
boundaries of moduli space, the functions $\hVW_{N,\mu,J}(\tau)$ are holomorphic generating functions of invariants of
$\cM_{N,\mu,n,J}$ for fixed rank $N$ and  first Chern class $\mu$,
\be
\label{defhVW}
\hVW_{N,\mu,J}(\tau) = \sum_{n\geq 0} c_{N,\mu,n,J} \,
\q^{n- \tfrac{N-1}{2N}\mu^2 - \tfrac{N \chi(S)}{24} },
\ee
where $\q=\expe{\tau}$ and $\tau=\frac{\theta}{2\pi} + \frac{4\pi\I}{g^2_{\rm YM}}$ is the complexified gauge coupling.
When $(N,\mu,n)$ are coprime, the coefficients are equal to the
Euler number $\chi(\cM_{N,\mu,n,J})$.
When $(N,\mu,n)$ is not
primitive, the moduli space $\cM_{N,\mu,n,J}$ is singular, and it is expected that the
coefficients  $c_{N,\mu,n,J}$ (known as rational VW  invariants)
are given by \cite[(6.1)]{Manschot:2017xcr}
\be
c_{\gamma,J} =
\sum_{m|\gamma} (-1)^{\dim_\IC(\cM_{\gamma/m,J})-\dim_\IC(\cM_{\gamma,J})}\frac{\chi(\cM_{\gamma/m,J})}{m^2}\, ,
\label{defCNnJ}
\ee
where $\gamma=(N,\mu,-n+\frac12\, \mu^2)$ and  $\chi(\cM_{\gamma,J})$ is defined using
intersection cohomology \cite{Manschot:2016gsx}. To our knowledge this prescription has not yet been
derived from the path integral, but it is needed for modularity, as will become clear below.

Remarkably, the Hodge numbers of $\cM_{N,\mu,n,J}$ turn to be independent of the \kahler class $J$,  when $b_2^+(S)>1$. However,
this is no longer the case when $b_2^+(S)=1$ and $b_2(S)>1$.
In the following we shall restrict to the particular choice $J=-K_S$,
which corresponds to the attractor chamber for generalized DT invariants.

Electric-magnetic duality suggests that \eqref{defZVW} is a multivariate Jacobi form of weight
$(-\frac32,\frac12)$ with a suitable multiplier system. Since \eqref{defthetaVW} transforms as a
vector-valued modular form of weight $(\hf\,(b_2(S)-1),\hf)$,  this implies that \eqref{defhVW}
transforms as vector-valued modular form of weight $-\frac{1}{2}\, b_2(S)-1$,
which for a simply connected surface coincides with the weight $-\frac12\chi(S)$ predicted in \cite{Vafa:1994tf}.
This expectation is borne out for $N=1$, since in this case
\be
\label{h10anyS}
\hVW_{1,0}(\tau) = \eta(\tau)^{-\chi(S)}
\ee
for any manifold \cite{Gottsche:1990}, where the Dedekind eta function $\eta(\tau)=\q^{1/24}\prod_{n>0}(1-\q^n)$
is modular with weight $1/2$. It is also supported
by explicit computations for $S=K3$ at arbitrary rank \cite{Vafa:1994tf}, provided the
multicover formula \eqref{defCNnJ} is used for non-primitive Chern vectors.
However, for $N \geq 2$ and  $b_2^+(S)=1$, modularity can only be expected after including non-holomorphic
contributions from boundaries of moduli space where  the connections become reducible \cite{Vafa:1994tf}.
For $S=\mathbb{B}_9$ the del Pezzo surface of degree 0 (also known as $\frac12 K3$,
see \S\ref{sec_geodata} for a brief summary of the geometry of del Pezzo and Hirzebruch surfaces),
generating functions $\hVW_{N,\mu,J}(\tau)$ were computed for any rank in a particular chamber
in \cite{Minahan:1998vr}\footnote{The generating functions in other chambers have been studied in \cite[Section 6]{Klemm:2012sx},
and involve genuine mock modular forms.},
and turned out to be weak quasi-modular forms, whose
modular completion is obtained straightforwardly by replacing the quasi-modular Eisenstein
series $E_2(\tau)$ by $\widehat{E_2}(\tau)=E_2(\tau)-\frac{3}{\Im\tau}$. For $S=\IP^2$,
the non-holomorphic contribution for $N=2$ is more involved and was obtained long ago in \cite{Vafa:1994tf}
by borrowing results from the mathematics literature  \cite{Zagier:1975,hirzebruch1976intersection}.
For $N=3$, it was obtained only recently \cite{Manschot:2017xcr}, using the method developed in \cite{Alexandrov:2016enp}.

\subsection{Refined Vafa-Witten invariants}

Refined VW invariants are defined by replacing the signed Euler number $\tilde\chi(\cM)$ in \eqref{defCNnJ}
by the $\chi_{y^2}$-genus defined in footnote \ref{fooP}.
Since
the cohomology of $\cM$ is expected to be supported only in Dolbeault degree $(p,p)$
for the surfaces under consideration (as shown for $S=\IP^2$ in \cite{ellingsrud1993towards})
one may use the Poincar\'e-Laurent polynomial   $P(\cM,y)$
defined in \eqref{defOmref}. Thus,
after including contributions from boundary of moduli space, the
partition function of  twisted $\cN=4$ Yang-Mills theory with a fugacity $y=\expe{w}$ conjugate to
the R-symmetry current is expected to be given by
\be
\label{defZVWref}
\widehat{Z}_{N,J}^{\rm VW,ref}(\tau,v,w) = \sum_{\mu\in \Lambda_S/N\Lambda_S}
\whhrVW_{N,\mu,J}(\tau,w)\, \vartheta^{\rm VW}_{N,\mu,J}(\tau,v),
\ee
where $\whhrVW_{N,\mu}(\tau,w)$ is the modular completion  of the
holomorphic generating function of refined rational VW invariants
\be
\label{defhVWref}
\hrVW_{N,\mu,J}(\tau,w) =
\sum_{n\geq 0}
\frac{ c_{N,\mu,n,J}^{\rm ref}(y)}{y-y^{-1}}\,
\q^{n -\tfrac{N-1}{2N} \mu^2 - \tfrac{N \chi(S)}{24}},
\ee
given by the standard multicover prescription \cite{joyce2008configurations}
\be
c_{\gamma,J}^{\rm ref}(y) =  \sum_{m|\gamma} \frac{(-1)^{m-1}(y-1/y) }{m(y^m-y^{-m})}\, P(\cM_{\gamma/m,J}, - (-y)^{m})\, ,
\label{defcref}
\ee
where the Poincar\'e polynomial was defined in \eqref{defOmref}.
Note that in contrast to the definition \eqref{defbOm} of refined Donaldson-Thomas invariants, and
in line with conventions in the literature,  the refined Vafa-Witten invariants for primitive charge
vector are defined as the Poincar\'e polynomial $P(\cM_{\gamma},\cdot )$
evaluated at $+y$, so that the coefficients are positive integers.
For the surfaces of interest the cohomology is supported in even degree
(in fact, in Dolbeault degree $(p,p)$) so that
\be
c_{\gamma,J}^{\rm ref}(-y)=(-1)^{d_\IC(\cM_{\gamma,J})}c_{\gamma,J}^{\rm ref}(y) .
\label{signcref}
\ee
The unrefined invariants  \eqref{defCNnJ} therefore arise in the limit $y\to 1$, whereas
the limit $y\to -1$ produces the same invariants up to an overall sign $(-1)^{\dim_\IC \cM_{\gamma,J}}$.
Since the parity of the dimension given in \eqref{dimM} is independent of the
second Chern class $n$, the relation \eqref{signcref} implies that the generating functions \eqref{defhVWref} satisfy
\be
\hrVW_{N,\mu,J}\(\tau,w+\hf\)=(-1)^{(N-1)(\mu^2-\chi(\cO_S))+1}\hrVW_{N,\mu,J}(\tau,w).
\label{signfliphVW}
\ee

It is natural to expect that \eqref{defZVWref} should transform as a multivariate
Jacobi form of weight $(-\frac12,\frac12)$ and suitable index. This implies
that $\whhrVW_{N,\mu}$ (omitting
$J$ from the notation) should transform as a vector-valued Jacobi
form of weight $-\frac12b_2(S)$,
\bea
\whhrVW_{N,\mu}\left(-\frac{1}{\tau},\frac{w}{\tau}\right)&=&
\I (-1)^{N-1}\,
\frac{(-\I\tau)^{-\frac{b_2(S)}{2} } }{N^{b_2(S)/2}}\,
\expe{ m_S(N)\, \tfrac{w^2}{\tau}
}\!\!\!
\sum_{\nu\in \Lambda_S/N\Lambda_S}\!\!\!
\expe{-\frac{\mu\cdot \nu}{N}} \whhrVW_{N,\nu}(\tau,w),
\nn\\
\whhrVW_{N,\mu}(\tau+1,w)&=&
\expe{-\tfrac{N-1}{2N}\, \mu^2 -\tfrac{N \chi(S)}{24}} \whhrVW_{N,\mu}(\tau,w),
\label{STrefVW}\\
\whhrVW_{N,\mu}(\tau,w+k\tau+\ell)&=& \expe{- m_S(N) \( k^2\tau + 2 k w\)} \whhrVW_{N,\mu}(\tau,w),
\nn
\eea
where the index $m_S(N)$ remains to be specified.
In support of this expectation,  for $N=1$, the Betti numbers are known for any $S$ \cite{Gottsche:1990}, leading when
$b_1(S)=0$ to a  Jacobi form
of weight $-\frac12 b_2(S)$ and index $-2$,
\be
\label{h10anySref}
\hrVW_{1,0}(\tau,w) = \frac{\I}{\theta_1(\tau,2w)\, \eta(\tau)^{b_2(S)-1}},
\ee
where the Jacobi theta function\footnote{Note that
if $\phi(\tau,w)$ is a Jacobi form of weight $k$ and index $m$ for a subgroup $\Gamma\subset SL(2,\IZ)$, then $\phi(r\tau,sw)$ is
a Jacobi form of weight $k$ and index $m s^2/r$ for a level $r$
congruence subgroup of $\Gamma$.}
$\theta_1(\tau,w)=\I \sum_{r\in\IZ+\frac12}(-1)^{r-\frac12} \q^{r^2/2} y^r$
is a Jacobi form of weight $1/2$ and index $1/2$.
In the limit $w\to 1/2$, this reduces as it should to
\eqref{h10anyS} upon using   $\theta_1(\tau,z+1) = -
\theta_1(\tau,z) \sim 2\pi z\,\eta(\tau)^3$ as $z\to 0$.

\medskip

In order to predict the  index of the generating functions \eqref{defhVWref}  (or rather, their
non-holomorphic completion) under modular transformations,
we shall rely on the known answer for stacky invariants\footnote{Stacky invariants
are polynomial combinations of the rational invariants $c_n^{\rm ref}(y)$ which
transform simply under wall-crossing. Whenever the
Chern vector is primitive and $J$ lies inside the \kahler cone they agree with $c_n^{\rm ref}(y)$.}
of the moduli space of semi-stable sheaves on the Hirzebruch surface $\mathbb{F}_1$
(see \S\ref{sec_geodata} for basic properties of these surfaces).
For arbitrary  rank $N$ and in the chamber where $J$ is aligned along the elliptic fiber class $f$ of  $\mathbb{F}_1$,
the generating function of stacky invariants is given by
(\cite[Conjecture 4.1]{Manschot:2011ym}, proven in \cite{Mozgovoy:2013zqx})
\be
\label{defHN}
H_{N}(\tau,w) = \frac{\I (-1)^{N-1} \eta(\tau)^{2N-3}}
{\theta_1(\tau,2Nw)\, \prod_{k=1}^{N-1} \theta_1(\tau,2kw)^2}
\ee
for $\mu \cdot f=0 \mod N$, or 0 otherwise, which
is a Jacobi form of weight $-1$ and index $-\frac13(4N^3+2N)$.
We do not expect the weight or index to depend on whether $(N,\mu)$ is primitive or not. Therefore,
the generating function of rational VW invariants on $S=\mathbb{F}_1$ should have the same weight and index in any chamber.
Since $\mathbb{F}_1$ is isomorphic to the one-point blow-up of $\IP^2$, one can obtain the generating function of refined
VW invariants on $\IP^2$ by applying the blow up formula \cite{Yoshioka:1996,0961.14022}, which amounts
to  dividing by
\be
\label{defBNk}
B_{N,k}(\tau,w) = \frac{1}{\eta(\tau)^N} \sum_{\sum_{i=1}^N a_i=0 \atop a_i\in\IZ+\frac{k}{N}}
\expe{-\tau\, \sum_{i<j} a_i a_j}\, y^{\sum_{i<j}(a_i-a_j)}.
\ee
This factor arises from bundle moduli associated to the exceptional divisor,
and transforms as a vector-valued Jacobi form of weight $-\frac12$ and index\footnote{This follows by recognizing
the numerator of \eqref{defBNk} as a theta series of the form \eqref{Vignerasth} with $\Phi=1$ for the positive
definite lattice $A_{N-1}$, with $\bfv=w \varrho$ where $\varrho$ is  the  Weyl vector $(\frac{N-1}{2},\frac{N-3}{2}, \dots,
\frac{3-N}{2},\frac{1-N}{2})$ of norm $(N^3-N)/3$.}
$\frac16(N^3-N)$. We conclude that  the generating function \eqref{defhVWref} for $S=\IP^2$
(or rather, its non-holomorphic modular completion)
has weight $-\frac12=-\frac12 b_2(S)$ and index $-\frac{3N^3+N}{2}$.
This agrees with the results of  \cite{Manschot:2017xcr} for $N\le 3$.
Similarly, by blowing up $1\leq k\leq 8$ generic points on $\IP^2$, one obtains that
the generating function \eqref{defhVWref} for the del Pezzo surface  $S=\mathbb{B}_k$ has
weight $-\frac12(k+1)$ and index $-\frac16(9-k)(N^3-N)-2N$ for any $N$.
Since $K_S^2=9-k$ for  $S=\mathbb{B}_k$, this supports the conjecture
that for any  almost Fano surface $S$ with $b_1(S)=0$ and $b_2^+(S)=1$, the generating function \eqref{defhVWref}
should have the weight and index given by
\be
\label{indexconjVW}
w_S = -\frac12\, b_2(S),
\qquad
m_S(N) = -\frac16\, K_S^2 (N^3-N)-2N ,
\ee
In the next subsection, we shall compare the index \eqref{indexconjVW} with the general prediction \eqref{indexconjDT}
and confirm the universal value of the coefficient of the cubic term.

Since refined VW invariants are expected to be well-defined on any almost  complex surface,
it is natural to ask what should be the weight $w_S(N)$ and index $m_S(N)$ of the generating function\footnote{When $b_2^+(S)>1$, both
$\whhrVW_{N,\mu}$ and its unrefined counterpart include additional contributions from
the `monopole branch', analyzed recently in \cite{thomas2018equivariant,Laarakker:2018isn}.}
$\whhrVW_{N,\mu}$ for an almost
complex surface $S$ with arbitrary values of $b_1(S)$ and $b_2^+(S)$.
Using similar arguments as in \cite{Vafa:1994tf}, we expect that the weight and index should be linear
combinations of the Euler number $\chi(S)$ and signature $\sigma(S)$,
which are the only topological invariants which can be expressed as an integral over a local field,
\be
w_S(N)=a_N \chi(S)+b_N \sigma(S),
\qquad
m_S(N)=c_N\chi(S)+d_N\sigma(S).
\ee
The coefficients $a_N,b_N,c_N,d_N$
can be determined by computing the weight and index for two (non-birationally equivalent)
four-manifolds. We can take for these four-manifolds, for example,
$\mathbb{F}_n$ and $K3$, which gives the formula \eqref{genindex0} stated in the introduction,
\be
\label{genindex}
w_S(N)=\frac14(\sigma(S)-\chi(S)),
\qquad
m_S(N)=-\frac{1}{6}(2N^3+N)\,\chi(S)-\frac{1}{2}N^3\sigma(S).
\ee
These expressions are in agreement with \eqref{indexconjVW} for a Fano surface due to relations
$\chi(\mathcal{O}_S)=\frac14\,(\chi(S)+\sigma(S))=1$ and $K_S^2=2\chi(S)+3\sigma(S)=12-\chi(S)$.
As an independent check, we verified that \eqref{genindex} also
matches with the explicit generating functions for ruled surfaces with the base given by
a genus $g$ Riemann surface \cite{Manschot:2011ym}.
In this case the weights of the unrefined and refined
partition functions \eqref{defZVW} and \eqref{defZVWref} become
$\(b_1(S)-\frac12 b_2^+(S)-1, \hf\, b_2^+(S)\)$ and
$\(\hf(b_1(S)-1),\hf\,b_2^+(S)\)$, respectively, consistently with the fact
that $\widehat{Z}_{N,J}^{\rm VW}$ coincides with the limit of $(y-1/y)^{\chi(\cO_S)} \widehat{Z}_{N,J}^{\rm VW,ref}$ as $y\to -1$.
Note that the weight in \eqref{genindex} is in agreement with the weight $w_S=-5\chi(\cO_S)+K_S^2/2$ for an algebraic surface $S$
proposed in \cite[Conjecture 1.10]{gottsche2018refined}, while the index $m_S(N)$ only agrees
with their proposal for $N=1,2,3$ thanks to the numerical coincidence
$\frac{N^3-N}{6} = (N-1)^2$ for these values.

\subsection{Local limit of elliptically fibered CY threefolds \label{sec_locallim}}
\label{subsec-local}

In order to extract the modular completion of VW invariants from that of generalized DT invariants
for compact CY threefold, we consider a smooth elliptic fibration $\CY\to S$ with a single section over a compact, smooth
almost Fano base $S$ with $b_1(S)=0$ and $b_2^+(S)=1$.
This implies that $\chi(\cO_S)=1$, such that the divisor $S$ is rigid inside $\CY$.
This restricts $S$ to be either a Hirzebruch surface $\mathbb{F}_k$ with $0\leq k\leq 2$
or a del Pezzo surface $\mathbb{B}_k$ with $0\leq k \leq 8$ (if we were to allow orbifold singularities on the base,
then there are more possibilities corresponding to reflexive two-dimensional polytopes).
The \kahler moduli space of $\CY$ includes \kahler deformations of the base and of the fiber, hence has
dimension $h_{1,1}(\CY)=b_2(S)+1$.
The \kahler cone  of $\CY$ is generated by $(\omega_e,\omega_\alpha)$
where $\alpha=1,\dots, h^{1,1}(S)$, while the Mori cone is generated by the dual divisors
$(D_e,D_\alpha)$. These divisors satisfy
\be
\label{Dinter}
D_e^3 = C_{\alpha\beta} c_1^\alpha c_1^\beta,
\quad
D_e^2 \cap D_\alpha = C_{\alpha\beta} c_1^\beta,
\quad
D_e \cap D_\alpha \cap D_\beta = C_{\alpha\beta},
\quad
D_\alpha \cap D_\beta \cap D_\beta =0,
\ee
where $C_{\alpha\beta}$  is the intersection matrix on $S$ and $c_1^\alpha$ are
the components of the first Chern class, given by
\be
C_{\alpha\beta}=\int_S \omega_\alpha \wedge \omega_\beta,
\qquad
c_1(S) = c_1^\alpha \, \omega_\alpha.
\ee
The divisor
\be
\label{deffrakB}
[S] = D_e - c_1^\alpha D_\alpha
\ee
can be identified with the base $S$ of the elliptic fibration, or equivalently
with its unique section, since it
satisfies
\be
[S] ^3 = D_e^3 ,
\qquad [S] ^2 \cap  D_\alpha =
[S]  \cap  D_\alpha \cap D_\beta = 0.
\ee
The Euler number of $\CY$ is
\be
\chi(\CY) = 2 ( h_{1,1}(\CY) - h_{2,1}(\CY) ) =
-60 \int_S c_1^2 = -60 \, C_{\alpha\beta} c_1^{\alpha} c_1^{\beta},
\ee
while the components of the second Chern class $c_2(T\CY)$ are
\be
\label{c2loc}
c_{2,e} = \int_{S} ( 11 c_1^2 + c_2),
\qquad
c_{2,\alpha} = 12 \, C_{\alpha\beta} c_1^{\beta},
\ee
which ensures that $[S]^3 + c_2(T\CY) \cap [S]$ coincides with the
Euler number $\chi(S) = \int_S c_2(S)$ of the base. Using \cite[Eq.(3.3)]{Maldacena:1997de} and
$\chi(\cO_S)=\frac{1}{12}  \int_{S} ( c_2+ c_1^2)=1$, we see  that the dimension of the space
of deformations of $S$ inside $\CY$  vanishes,
\be
\label{nbdef}
\frac13\, S^3 + \frac16\, c_2(T\CY)\cdot S -2 = \frac16 \int_{S} ( c_2(S) + c_1(S)^2  ) - 2
=2 \(\chi(\cO_S) - 1 \) = 0,
\ee
in agreement with the fact that the divisor $S$ is rigid.
The local limit is obtained by taking $t_e\to\infty$,
so that the CY threefold degenerates
to the total space ${\rm Tot}(K_S)$ of the canonical bundle over $S$.\footnote{In the case of
elliptic fibrations with $r>1$ sections one finds \cite{Klemm:2012sx}
that the intersection numbers in \eqref{Dinter} are rescaled by a factor of $r$,
$c_{2,\alpha}$ is still given by the second equation in \eqref{c2loc}, while the formulae for
$\chi(\CY)$ and for $c_{2,e}$ are modified in such a way that $[S]^3 + c_2 \cap [S]$ coincides with $r\chi(S)$.
The number of complex deformations \eqref{nbdef} then
becomes $2r-2$, so the divisor $S$ is not rigid unless $r=1$.}

For $N$ D4-branes wrapping $S$, the magnetic charge is given by $p^a = N p_0^a$ where
$p_0^a = (1, -c_1^{\alpha})$ are the components of the divisor $[S]$ in \eqref{deffrakB}, while
the electric charges are given by \cite{Diaconescu:2007bf},
\cite[Eq.(3.5)]{Alexandrov:2012au}
\be
\begin{split}
q_\alpha =&\, - \int_S\, i^*(\omega_\alpha) \(c_1(F) + \frac{N}{2}\, c_1(S) \),
\qquad
q_e=0,
\\
q_0 =&\, \frac{1}{2N} \left( c_1(F) + \frac{N}{2}\, c_1(S) \right)^2 + \frac{N}{24}\, \chi(S) - N\, \Delta(F),
\end{split}
\label{geom-q}
\ee
where $\Delta(F)$ is the Bogomolov discriminant defined in \eqref{dimM}
and $i^*(\omega_\alpha)$ is the pull back of the Poincar\'e dual of the divisor $\omega_\alpha$ under the inclusion $S \to \CY$.
The consistency of these formulae with the  quantization conditions \eqref{fractionalshiftsD5}
can be checked by using Wu's formula \cite{Vafa:1994tf}
\be
c_1(F) \(c_1(F)+c_1(S)\)\in 2 \IZ,
\label{Wu}
\ee
the relation $c_1^2(S)=p_0^3 $, and the fact that the function
\be
L_0(p)=\frac{1}{6}\, p^3 + \frac{1}{12}\, c_{2,a}p^a
\label{delL0}
\ee
is integer-valued since it coincides with $\sum_{i=0}^4 (-1)^i \dim H^i(\CY,\CL)$ where $\CL$
is the line bundle associated to the divisor $\cD$ defined by the vector $p^a$ \cite{Maldacena:1997de}.
When $\cD$ is very ample, it is equal to the number of complex deformations of $\cD$ inside $\CY$.
In our case, $[S]$ is not ample, but due to \eqref{nbdef} one finds
\be
L_0(p)=N\chi(\cO_S)+\frac{N(N^2-1)}{6}\,K_S^2,
\label{valL0}
\ee
which is indeed integer. In particular, for $N=1$ one has $L_0(p_0)=\chi(\cO_S)=1$.

Using \eqref{geom-q}, the Dirac product between two charge vectors
$\gamma_i=( 0, N_i p_0^a , q_{i,a}, q_0)$ (or equivalently, the reduced charge vectors $\gama_i=( N_ip_0^a , q_{i,a})$)
can be written as
\be
\label{gam12}
\langle\gamma_1,\gamma_2\rangle=\,
c_1^{\alpha} (N_1 q_{2,\alpha} -N_2 q_{1,\alpha}) = K_S \cdot ( N_1\, c_1(F_2) - N_2\, c_1(F_1) )
\ee
in agreement with \cite[Eq.(3.12)]{Diaconescu:2007bf}.
Furthermore, the invariant D0-brane charge \eqref{defqhat} becomes
\be
\label{defqhatloc}
\hat q_0 = q_0 - \frac{1}{2N}\,C^{\alpha\beta}q_\alpha q_\beta =
 -\int_S \left[ c_2(F) - \frac{N-1}{2N}\, c_1^2(F) \right] + \frac{N}{24} \,\chi(S) .
\ee
Substituting this into the generating function \eqref{defhDT}, one finds that the power of $\q=\expe{\tau}$ precisely matches
the one in the VW generating function \eqref{defhVW}.

Next, we compute the induced quadratic form on the lattice $\Lambda$,
\be
\kappa_{ab} = \kappa_{abc} p^c = N\(\begin{array}{cc}
0 & 0
\\
0 & C_{\alpha\beta}
\end{array}\).
\ee
It is degenerate along the fiber direction, corresponding to the fact that the divisor $[S]$ is not ample.
The standard decomposition of the electric charge \eqref{decq}
no longer holds in this case, since the spectral flow leaves the charge $q_e$
along the fiber invariant. Fortunately, $q_e$ automatically vanishes
for a D4-brane wrapped on $S$, while other D4-brane configurations carrying $q_e\neq 0$
necessarily wrap the elliptic fiber and
are therefore exponentially suppressed in the limit $t_e\to \infty$.
The remaining part can be decomposed as
\be
\label{quant-q}
q_\alpha =  \mu_\alpha  + N \, C_{\alpha\beta} \epsilon^{\beta}  - \frac{N^2}{2} C_{\alpha\beta} c_1^\beta,
\ee
where $\epsilon^\alpha$ runs over $\Lambda_S=H^2(S,\IZ)$. Since the lattice $\Lambda_S$
equipped with the intersection form $C_{\alpha\beta}$ is unimodular, $\mu_\alpha$ runs
over  $\Lambda_S/ (N \Lambda_S)$. The sum over the spectral flow parameter $\epsilon^\alpha$
leads to the Siegel-Narain theta series \eqref{defthSN} which now takes the form
\be
\label{thetaloc}
\thetasn_{Np_0,\mu}(\tau,v) =\sum_{q\in N \Lambda_S+ \mu+\frac{N^2}{2}K_S}
 (-1)^{N K_S \cdot q}\, \expe{ -\frac{\tau}{2N}\, q^2 + \frac{\I \tau_2}{N} \, (q+N b)_+^2 + q_a v^a}.
\ee
This differs from the theta series \eqref{defthetaVW} by the choice of the characteristic vector,
$N^2K_S$ and $NK_S$ respectively, which leads to the following  relation (see footnote \ref{foochar})
\be
\thetasn_{Np_0,\mu}
=\, (-1)^{(N-1)\( \mu^2 +\frac{N}{2}\, K_S^2\)}\,\vartheta^{\rm VW}_{N,\mu+\frac{1}{2}\,N(N-1) K_S},
\label{thetarel}
\ee
where we used Wu's formula \eqref{Wu} to replace $K_S\cdot \mu$ by $\mu^2$.

The relation \eqref{thetarel} shows that the translation from CY to VW conventions
involves a shift of  the residual flux.
Besides, comparing the definitions of the refined invariants \eqref{defbOm} and \eqref{defcref},
one observes that the corresponding refinement parameters are related by a sign flip.
Therefore, we expect that the generating functions of MSW and VW invariants should be related by
\be
\hr_{Np_0,\mu}(\tau,w)  = \hrVW_{N,\mu+\frac12N(N-1) K_S }(\tau,w+\haf)\, ,
\label{relhh}
\ee
possibly up to an overall $\mu$-independent sign.
As a consistency check of this relation, one can verify with help of the property \eqref{signfliphVW} that
the modular transformations \eqref{STrefVW} agree with those in \eqref{STref}.

Most importantly, the relation \eqref{relhh} implies that the modular completions of the
respective generating functions should
satisfy the same identity. In appendix \ref{ap-vwp2}, we show
that this is indeed the case for $\IP^2$ at rank $N\leq 3$, with the relative sign chosen as in \eqref{relhh}.\footnote{Actually,
assuming that both DT and VW invariants are given by the
$\chi_{y^2}$-genus of the respective moduli spaces, one finds
%the sign flip of the refinement parameter implies
 that the relation \eqref{relhh}
between the generating functions \eqref{defhDTr} and \eqref{defhVWref} should
have an extra minus sign. %, due to the prefactor $(y-1/y)^{-1}$.
However, this would lead to an additional factor $(-1)^{N-1}$ in the relation \eqref{gtogp}
and correspondingly to sign discrepancies in the modular completions.
Thus, the comparison of the two completions suggests that
%this sign should not be included so that
 the DT invariants should be identified with {\it minus} the VW invariants.
We do not yet understand yet the reason for this overall sign
flip, nonetheless, we observe that the VW invariants for $\IP^2$ descend via blow-down
from the VW invariants of
$\mathbb{F}_1$ in a particular chamber which
contains {\em more} states than nearby chambers \cite{Manschot:2010nc, Manschot:2011dj},
unlike the usual  attractor chamber in supergravity.
 This observation might be relevant for this sign issue. \label{foot-sign}}
This remarkable agreement gives strong evidence that our prescription for constructing the modular completion
of generating functions of refined VW invariants should
hold for any rank $N\geq 2$ and for any smooth almost Fano surface with $b_1(S)=0$, $b_2^+(S)=1$.
In appendix \ref{ap-rank4} we make our prescription explicit for $S=\IP^2$, $N=4$.

Finally, it is suggestive to rewrite the (conjectural) formula \eqref{indexconjVW} for the
index of the generating function of refined VW invariants in terms of the  local Calabi-Yau geometry.
Using \eqref{valL0}, it is straightforward to see that the index \eqref{indexconjVW} evaluates to
\be
\label{mp2}
m_S(N) =  -L_0(p)-N,
\ee
where $L_0$ is the function defined in \eqref{delL0}, which is known to be integer.
This is consistent with \eqref{indexconjDT} upon setting $\rho_a=-\frac1{12}c_{2,a}$,
up to an integer shift by $N$, which can be viewed as the largest integer such that $p^a/N$ is
a primitive vector. This shift is consistent with the linearity of $m(p)+\frac16 p^3$
and, upon restricting to the BPS states with $p^a=Np_0^a$,
does not spoil the derivation of \eqref{treeFh} and hence of the integral representation
for the instanton generating potential. It is conceivable that this shift is related to the existence
of a degenerate direction in the quadratic form $\kappa_{ab}$.   It would be interesting to confirm the formula \eqref{mp2}
for a wider class of non-compact Calabi-Yau geometries, possibly with more singular quadratic form.

\section{Holomorphic anomaly for refined Vafa-Witten invariants}
\label{sec-holan}

In the previous section, we gave strong support that the modular completion of the
generating function of (refined) VW invariants should be given by the general prescription \eqref{exp-whhr}
for a wide class of complex surfaces $S$.
Despite being quite attractive, this conjecture apparently raises a puzzle: as shown
in \cite[Prop. 9]{Alexandrov:2018lgp} in the unrefined case, the modular completion $\whh_{p,\mu}$
generically satisfies a holomorphic anomaly equation of the form
\be
\p_{\bar\tau}\whh_{p,\mu}(\tau)= \sum_{n=2}^\infty
\sum_{\sum_{i=1}^n \gama_i=\gama}
\cJ_n(\{\gama_i\},\tau_2)
\, e^{\pi\I \tau Q_n(\{\gama_i\})}
\prod_{i=1}^n \whh_{p_i,\mu_i}(\tau),
\label{exp-derwh0}
\ee
where the coefficients $\cJ_n(\{\gama_i\},\tau_2)$ are generically non-vanishing for all $n$'s.
In contrast, for the case $S=\hf K3$ studied in  \cite{Minahan:1998vr},
it was found that the non-holomorphic completion satisfies an equation of the form
\be
\label{naivehol}
\p_{\bar\tau}\widehat{Z}^{\rm VW}_{N}= \frac{\rm cte}{\tau_2} \sum_{N=N_1+N_2}  N_1 N_2\,
\widehat{Z}^{\rm VW}_{N_1} \widehat{Z}^{\rm VW}_{N_2},
\ee
with only products of two partition functions appearing on the r.h.s. In this section,
we shall resolve this puzzle and establish a precise version of the holomorphic
equation \eqref{naivehol}, as well as its refined version, assuming
that the modular completion of $\hrVW_{p,\mu}$ is indeed given by the
prescription \eqref{exp-whhr} for all $N>1$.

\subsection{Refined holomorphic anomaly for compact CY threefold}

Let us first return to the compact case, and establish the  holomorphic anomaly
equation satisfied by the modular completion $\whhr_{p,\mu}(\tau)$ of the generating
function of refined DT invariants, ignoring the fact that these invariants may not be well-defined.
Since the anomaly coefficients $\cJ_n$ \eqref{exp-derwh0} given explicitly in \cite[Eq.(5.36)]{Alexandrov:2018lgp}
are expressed solely in terms of the functions $\cE_n$, it is clear that the anomaly
can be directly translated to $\whhr_{p,\mu}$ by replacing $\cE_n$ by $\cEr_n$.
As a result, one obtains
\be
\p_{\bar\tau}\whhr_{p,\mu}(\tau,w)= \sum_{n=2}^\infty
\sum_{\sum_{i=1}^n \gama_i=\gama}
\cJr_n(\{\gama_i\},\tau_2,y)
\, e^{\pi\I \tau Q_n(\{\gama_i\})}
\prod_{i=1}^n \whhr_{p_i,\mu_i}(\tau,w),
\label{exp-derwh}
\ee
where
\be
\cJr_n(\{\gama_i\},\tau_2,y)= \frac{\I}{2}\,\Sym\left\{\sum_{T\in\IT_n^{\rm S}}(-1)^{n_T-1}
\p_{\tau_2}\cEr_{v_0}\prod_{v\in V_T\setminus{\{v_0\}}}\cEr_{v}\right\}.
\label{solJn}
\ee
Hence, for a reducible divisor class $[\cD]=p^a \omega_a$,
the antiholomorphic derivative of the modular completion $\whhr_{p,\mu}$ is given by a
sum of monomials in $\whhr_{p_i,\mu_i}$ for all splittings $p^a=\sum_{i=1}^n p_i^a$, multiplied
by an indefinite theta series of signature $(n-1) (1,b_2-1)$ with the kernel given by $\cJr_n$.
Using the techniques in \cite{Alexandrov:2016enp}, one may express this theta series as an
iterated Eichler integral of an ordinary Gaussian theta series, justifying the name `mock
{Jacobi}
form of depth $n-1$' for the holomorphic generating functions $\hr_{p,\mu}$.

The new feature here compared to the unrefined case
is that the $\tau_2$-dependence in $\cEr_n$ originates not only from the overall factor in the argument $\bfx$
of generalized error functions (see \eqref{argxref}), but also due to the $\beta$-dependent shift
in this argument since the $\tau_2$ derivative in \eqref{solJn} {is evaluated by keeping $y$
and $\bar y$ constant,}
while $\beta=\log(y\by)/(4\pi\tau_2)$. As a result, taking into account that
the derivative lowers the {order}
of generalized error functions \cite{Nazaroglu:2016lmr}, the $\tau_2$-derivative gives
\bea
\p_{\tau_2}\cEr_n(\{\gamma_i\},y)&=&\frac{(-y)^{\sum_{i<j} \gamma_{ij} }}{2^{n}\tau_2}\,
\bfx' \cdot\p_\bfx\Phi^E_{n-1}(\{ \bfv_{\ell}\};\bfx)
\nn\\
&=&\frac{(-y)^{\sum_{i<j} \gamma_{ij} }}{2^{n-1}\tau_2}\sum_{k=1}^{n-1}\frac{(\bfv_k,\bfx')}{|\bfv_k|}\,
e^{-\pi\,\frac{(\bfv_k,\bfx)^2}{\bfv_k^2} }\Phi^E_{n-2}(\{ \bfv_{\ell\perp k}\}_{\ell\ne k};\bfx),
\label{derivPhiE}
\eea
where
\be
\bfx'=\bfx-2\sqrt{2\tau_2}\,\beta \ptt=\sqrt{2\tau_2}(\bfq+\bfb-\beta\ptt)
\ee
and $\bfv_{\ell\perp k}$ denotes the projection of $\bfv_\ell$ on the hyperplane orthogonal to $\bfv_k$.

\subsection{Collinear charges}

Let us now assume that the only possible allowed splittings of the magnetic charge $p^a$
are of the form $p^a = \sum_i N_i p_0^a$ where $p_0^a$ is a fixed vector such that $p^a=N p^a_0$,
and $N_i$ run over all possible partitions of $N$. This case is relevant in particular for compact
CY threefolds with $b_2=1$ such as the quintic, or for local Calabi-Yau threefolds, where $p_0^a$
is the class of the divisor $S$. An important consequence of this assumption is that DT invariants
$\bOm(\gamma,z^a)$ with $p^a=N p^a_0$ become independent of the $b$-field, which cancels from
the relative central charge in the large volume limit \cite[Eq.(2.11)]{Alexandrov:2016tnf}
\be
\label{ImZZb}
\begin{split}
\Im(Z_{\gamma_1}\bZ_{\gamma_2}) =&  -\frac12\, \Bigl[
(p_2 t^2)\, (q_{a,1}+(p_1 b)_a) t^a - (p_1 t^2)\, (q_{a,2}+(p_2 b)_a) t^a \Bigr]
\\
=& -\frac12 (p_0 t^2) \(N_2 q_{1,a} - N_1 q_{2,a} \) t^a.
\end{split}
\ee
This is consistent with the fact that VW invariants depend only on the \kahler class
$t^a$ and not on $b^a$.
Another remarkable consequence of the above restriction is that it ensures the vanishing of all terms of degree
higher than two in the holomorphic anomaly equation \eqref{exp-derwh}. Namely, in appendix \ref{ap-theorem}
we prove the following

\begin{theorem}\label{theor}
For $n>2$, $\cJr_n=0$.
\end{theorem}

Due to this theorem, the holomorphic anomaly for the completion of the generating function of refined MSW invariants
contains only one term quadratic in $\whhr_{N_i,\mu_i}$. Substituting the explicit expression for $\cJr_2$ \eqref{valJr2},
one therefore obtains
\be
\p_{\bar\tau}\whhr_{N,\mu}(\tau,w)=
\sum_{N_1+N_2=N\atop q_{1,a}+q_{2,a}=q_a}
\frac{\I(-y)^{\gamma_{12}}}{2\sqrt{2\tau_2}}\, \frac{\gamma_{12}-\beta p_0^3 NN_1N_2}{\sqrt{p_0^3 NN_1N_2}}\,
e^{-2\pi\tau_2\,\frac{(\gamma_{12}+\beta p_0^3 NN_1N_2)^2}{p_0^3 NN_1N_2} }
\, e^{\pi\I \tau Q_2(\gama_i,\gama_2)}
\whhr_{N_1,\mu_1}\whhr_{N_2,\mu_2}.
\label{exp-derwh-col}
\ee

This result can be used to get the holomorphic anomaly in the unrefined case which must emerge in the limit $y\to 1$.
However, one should be careful evaluating this limit since the result depends on the direction along which $y$ approaches to 1.
The correct way to proceed is to take the holomorphic limit,
i.e. set $\by=1$ while sending $y$ to 1.
This is because the holomorphic limit commutes with
any modular operation such as construction of the completion or evaluation of the shadow.
Indeed, since $w=\frac{1}{2\pi\I}\,\log y$ is a modular variable, Taylor coefficients of an expansion
in $w$ of a modular function are also modular.
In particular, the holomorphic limit of \eqref{exp-derwh-col} gives
\be
\p_{\bar\tau}\whh_{N,\mu}(\tau)=
\sum_{N_1+N_2=N\atop q_{1,a}+q_{2,a}=q_a}
\frac{(-1)^{\gamma_{12}}}{8\pi\I(2\tau_2)^{3/2}}\,\sqrt{p_0^3 NN_1N_2}\,
e^{-\frac{2\pi\tau_2 \gamma_{12}^2}{p_0^3 NN_1N_2} }
\, e^{\pi\I \tau Q_2(\gama_i,\gama_2)}
\whh_{N_1,\mu_1}\whh_{N_2,\mu_2}.
\label{exp-derwh-lim}
\ee
This result indeed coincides with the direct evaluation of the holomorphic anomaly in the unrefined case:
the right-hand side of \eqref{exp-derwh-lim} is the explicit form of the quadratic term in \eqref{exp-derwh0},
whereas the vanishing of higher order terms for collinear D4-brane charges can also be shown with some effort.
However, this is more difficult than in the refined case due to the more complicated form of the functions $\cE_n$
(cf. \eqref{rescEnPhi}, \eqref{rescEn} and \eqref{Erefsim}), which demonstrates once more the simplicity and power of
the refined construction proposed here.

\subsection{Partition function: unrefined case}
\label{subsec-pfunref}

We now examine the holomorphic anomaly of the partition function
\eqref{defZp} in the collinear case $p^a=N p_0^a$, restricting to the attractor point
$t^a=\lambda p_0^a$ with $\lambda\gg 1$. It will be convenient to introduce
\be
\cXsn_{p,q}=e^{-S^{\rm cl}_p-2\pi\tau_2(q+b)_+^2} \expe{- \tfrac{\tau}{2}\,q^2+q_a v^a}
\label{SNseries}
\ee
and denote $\whcZ_N=V_p^{-1}\whcZ_p$ where $\whcZ_p$ was defined in \eqref{defZp}, so that
\be
\whcZ_N = \sum_{q}\sigma_\gamma\,\whh_{p,\mu}\cXsn_{p,q}.
\label{pfZBN-compl}
\ee
The non-holomorphic dependence
on $\tau$ comes both from  the completion $\whh_{p,\mu}$,
governed by  \eqref{exp-derwh-lim}, and from the explicit dependence on $\tau_2$
 in \eqref{SNseries}. To obtain the  holomorphic anomaly equation, it suffices
to multiply \eqref{exp-derwh-lim}, by $\sigma_\gamma\cXsn_{p,q}$ and
convert the product of $\whh_{N_i,\mu_i}$ appearing on the right-hand side into a
product of two partition functions.
To this end, we use the following identity valid for $\gama_1+\gama_2=\gama$
\be
\cXsn_{p,q}=e^{-\pi\I \tau Q_2(\gama_1,\gama_2)}\, e^{2\pi\tau_2\, \frac{\(2\Im(Z_{\gamma_1}\bZ_{\gamma_2})\)^2}{(pt^2)(p_1t^2)(p_2t^2)}}
\cXsn_{p_1,q_1}\cXsn_{p_2,q_2}.
\label{decompSN}
\ee
Upon specifying \eqref{ImZZb} to the attractor point  $t^a=\lambda p^a_0$,
the argument of the exponential in \eqref{decompSN}  simplifies to
\be
\frac{\(2\Im(Z_{\gamma_1}\bZ_{\gamma_2})\)^2}{(pt^2)(p_1t^2)(p_2t^2)}=\frac{\gamma_{12}^2}{(pp_1p_2)}\, .
\ee
As a result, all exponential factors cancel and one remains with
\be
\overline{\mathcal{D}}\whcZ_N=\frac{\sqrt{2\tau_2}}{32\pi\I}\sum_{N_1+N_2=N}\sqrt{p_0^3 NN_1N_2}\,\whcZ_{N_1}\whcZ_{N_2},
\label{holanom-ZN}
\ee
where
\be
\overline{\mathcal{D}}=\tau_2^2\left( \partial_{\bar \tau}
-\frac{\I}{4\pi}\, (\partial_{c_+}+2\pi\I b_+)^2 \right)
\ee
is designed to commute with the Siegel-Narain theta series and, acting on the completion of a mock modular form,
decreases its holomorphic weight by 2.
This agrees with the fact that both sides of \eqref{holanom-ZN} are modular forms of weight $(-\frac32,\hf)$.
Finally, rescaling the partition function as
\be
\whcZ'_N=\frac{\sqrt{ p_0^3}}{\sqrt{N}}\,\whcZ_N,
\ee
one can further simplify the anomaly equation, which becomes
\be
\label{holanomZp}
\overline{\mathcal{D}}\whcZ'_N=\frac{\sqrt{ 2\tau_2}}{32\pi\I}\sum_{N_1+N_2=N} N_1N_2\, \whcZ'_{N_1}\whcZ'_{N_2},
\ee
giving a precise and general version of the conjecture \eqref{naivehol} from \cite{Minahan:1998vr}.

\subsection{Refined partition function}
\label{subsec-pfref}

Given the elegant form of the holomorphic anomaly equation for  the unrefined partition
function \eqref{defZp}, it is natural to ask whether its refined
counterpart
\be
\label{defZrN}
\whcZr_N= \sigma_p \, e^{-S^{\rm cl}_p} \,
\sum_{\mu\in\Lambda^*/\Lambda}\,
\whhr_{p,\mu}\, \thetasn_{p,\mu}
=\sum_{q_a}\sigma_\gamma\,\whhr_{p,\mu}\cXsn_{p,q},
\ee
satisfies a similar equation. Indeed, \eqref{defZrN}
 arises as the first term in  the expansion \eqref{treeFh-flhr} of the refined
instanton generating potential $\cGr$ and coincides up to a trivial factor with the partition function \eqref{defZVWref}
in the local CY case $\CY={\rm Tot}(K_S)$,\footnote{This follows from applying
\eqref{relhh}, \eqref{signfliphVW} and \eqref{thetarel}. Importantly, the two $\mu$-dependent sign factors coming from
the sign flip of the refinement parameter in the generating function and from the change
of the characteristic vector in the theta series cancel each other.
Note also that in contrast to \eqref{relhh}, this relation does not involve any shift of $w$.}
\be
\whcZr_N(w)=\sigma_p \, e^{-S^{\rm cl}_p} \, (-1)^{N +\frac{N(N-1)}{2}\, K_S^2}\, \widehat{Z}_{N}^{\rm VW,ref}(w).
\ee
The main difficulty is again to absorb all exponential factors appearing in the anomaly equation \eqref{exp-derwh-col}
for the refined generating function.
As we now demonstrate, this goal can be achieved using the same star product which has already appeared
in the discussion of the instanton generating potential in section \ref{sec-potential}.

The key observation is that the exponentials $\cXsn_{p,q}$ defined in \eqref{SNseries}
satisfy (for collinear charges and for attractor values of the K\"ahler moduli)
\be
\cXsn_{p_2,q_2}\star\cXsn_{p_1,q_1}=\by^{-\gamma_{12}}e^{-2\pi\tau_2\beta^2 (pp_1p_2)}\cXsn_{p_2,q_2}\cXsn_{p_1,q_1}
\ee
with respect to the star product defined in \eqref{starproduct-alt}.
To reproduce not only the exponential factor, but also the prefactor in \eqref{exp-derwh-col}, we actually need
a bit more complicated property, which reads as
\be
\begin{split}
&
\frac{1}{4\pi^2}\,\Bigl(\p_{\bv^a} \cXsn_{p_2,q_2}\star\p_{\tc_a}\cXsn_{p_1,q_1}- \p_{\tc_a}\cXsn_{p_2,q_2}\star\p_{\bv^a}\cXsn_{p_1,q_1}\Bigr)
\\
&\qquad
=\(\gamma_{12}-\beta(pp_1p_2)\)\by^{-\gamma_{12}}e^{-2\pi\tau_2\beta^2 (pp_1p_2)}\cXsn_{p_2,q_2}\cXsn_{p_1,q_1}.
\end{split}
\ee
Combining it with \eqref{decompSN}, one finds
\be
\begin{split}
\(\gamma_{12}-\beta(pp_1p_2)\)\cXsn_{p,q}=&\,
\frac{\by^{\gamma_{12}}}{4\pi^2}\,
e^{-\pi\I \tau Q_2(\gama_i,\gama_2)}\, e^{2\pi\tau_2\( \frac{\gamma_{12}^2}{(pp_1p_2)}+\beta^2 (pp_1p_2)\)}
\\
&\,
\times
\Bigl(\p_{\bv^a} \cXsn_{p_2,q_2}\star\p_{\tc_a}\cXsn_{p_1,q_1}- \p_{\tc_a}\cXsn_{p_2,q_2}\star\p_{\bv^a}\cXsn_{p_1,q_1}\Bigr).
\end{split}
\label{decompSN-star}
\ee
The factors on the right-hand side of this identity precisely cancel those appearing in the anomaly equation \eqref{exp-derwh-col}
so that one obtains the following holomorphic anomaly equation for the refined partition function
\be
\overline{\mathcal{D}}\whcZr_N=\frac{\I(2\tau_2)^{3/2}}{32\pi^2}\sum_{N_1+N_2=N}
\frac{1}{\sqrt{p_0^3 NN_1N_2}}\,
\Bigl(\p_{\bv^a}\whcZr_{N_2}\star\p_{\tc_a}\whcZr_{N_1}- \p_{\tc_a}\whcZr_{N_2}\star\p_{\bv^a}\whcZr_{N_1}\Bigr).
\ee
Since the star product is modular invariant, both sides of this equation are now modular forms of weight $(-1,\hf)$.
Rescaling the partition function as
\be
\whcZrp_N=\frac{\whcZr_N}{\sqrt{Np_0^3}}\, ,
\ee
the anomaly equation can be further simplified and becomes
\be
\label{holanomZpref}
\overline{\mathcal{D}}\whcZrp_N=\frac{\I(2\tau_2)^{3/2}}{32\pi^2 p_0^3 N}\sum_{N_1+N_2=N}
\Bigl(\p_{\bv^a}\whcZrp_{N_2}\star\p_{\tc_a}\whcZrp_{N_1}- \p_{\tc_a}\whcZrp_{N_2}\star\p_{\bv^a}\whcZrp_{N_1}\Bigr).
\ee
It is straightforward to check that in the (holomorphic) limit $y\to 1$ this equation reduces to \eqref{holanomZp}.
Furthermore, it is expected that the higher terms in the Taylor expansion around $y=1$ satisfy the
holomorphic anomaly equation for the refined topological string proposed in \cite[Eq.(8.16)]{Huang:2013yta}.
It would be very interesting to verify whether this is indeed the case.

\section{Conclusions and future directions}
\label{sec-concl}

In this paper we constructed a natural non-holomorphic completion $\whhr_{p,\mu}(\tau,w)$ \eqref{exp-whhr}
of the generating function of refined BPS invariants on a CY threefold $\CY$, evaluated in
the large volume attractor chamber, and proposed that this completion must be a vector valued Jacobi
form of specific weight and index. Since refined BPS invariants are mathematically well-defined only when $\CY$ admits a $\IC^\times$ action
(which implies that it is non compact), this proposal applies in such cases, which include toric Calabi-Yau threefolds and in particular
local CY geometries of the form ${\rm Tot}(K_S)$, although it is conceivable that it holds even
when the refined invariants are not protected.

In the limit $w\to 0$, this completion reduces to the one constructed in our earlier work \cite{Alexandrov:2018lgp},
allowing to sidestep many of the complications in this construction, which are now relegated to the final step
of  extracting the residue at the pole at $w=0$. In this sense, the proposed completion is natural.
However, unlike in the unrefined case where modularity followed from general S-duality constraints
on the vector multiplet moduli space $\cM_V$, an analogue of this constraint involving refined invariants
is not known at present,  and the modularity of $\whhr_{p,\mu}(\tau,w)$ should be viewed as conjectural.

Strong support for this conjecture comes from the local CY case $\CY={\rm Tot}(K_S)$,
where BPS indices of D-branes supported on the base $S$ are expected to be equal to
VW invariants of $S$. In this case, the  partition function \eqref{defZVWref}
built out of the generating functions is expected to be modular, as a consequence of
S-duality of $\cN=4$ Yang-Mills theory. Unfortunately, it is not known at present how to
derive the non-holomorphic completion from a gauge theory
computation,\footnote{Similar holomorphic anomalies are known for
  topological versions of $\cN=2$, $N_f=4$ and $\cN=2^*$ theory,
and can be understood from the physical path integral in these cases \cite{MooreStringMath2018,ManschotMooretoappear:2019}.}
and the known completions for $U(2)$ and $U(3)$ VW partition functions on $\IP^2$ rely on expressing
the holomorphic generating functions as Appell-Lerch sums and applying
mathematical recipes to find their modular completions.
Remarkably, we have shown that our natural construction reproduces
exactly these modular completions (up to the sign flip discussed in footnote \ref{foot-sign}),
without prior knowledge of the invariants.  Moreover, we have
found that at {\it arbitrary} rank $N$, the holomorphic anomaly equation reproduces the earlier proposal of \cite{Minahan:1998vr}
in the unrefined limit, and have found its generalization \eqref{holanomZpref} for refined invariants.
Clearly, it would be very useful to test these results at higher rank or on other surfaces.

 If the completed generating functions $\whhr_{p,\mu}(\tau,w)$ are indeed modular, then the construction
of \S\ref{sec-potential} produces a natural function on  $\cM_V \times \IC^\times$ (where the second
factor keeps track of the dependence on the chemical potential $y$) which transforms as a Jacobi
form of weight $(-\frac12,\frac12)$ and index 0, generalizing the `instanton generating potential'
from \cite{Alexandrov:2018lgp}. Moreover, this function has an extremely simple representation
\eqref{nonpert-cGr} in terms of solutions to the non-commutative integral equations \eqref{inteqH-star}.
These results suggest that the refinement induces a quantization of the moduli space $\cM_V$
along with its twistor space, such that the modularity of
generating functions of refined invariants might follow from requiring a consistent action of
S-duality on these deformed spaces.
It is also worth noting that the non-commutative integral equations \eqref{inteqH-star}
are reminiscent of \cite{Cecotti:2014wea}, where a non-commutative deformation of the TBA equations from \cite{Gaiotto:2008cd}
involving three parameters $\theta$, $\epsilon_1$ and $\epsilon_2$ was proposed. Comparing
the non-commutativity relation \eqref{starXX} with \cite{Cecotti:2014wea}
suggests that one should identify our parameter $w=\alpha-\tau\beta$ as
\be
\alpha \sim \theta,
\qquad
\beta\sim \epsilon_1.
\ee
We hope to report on a detailed comparison of the two constructions  elsewhere.

Finally, returning to our main motivation, our results open the way to a detailed
understanding of degeneracies of BPS black holes $\cN=2$ in  string theory.   Having characterized
the precise modular properties of generating functions of BPS indices, one can in principle determine
them from the knowledge of their polar coefficients, which could be computed by generalizing
ideas in \cite{Gaiotto:2005rp,Gaiotto:2007cd}. Corrections to the Bekenstein-Hawking area law
could in principle be computed by applying the circle method for the
completed partition functions \cite{Bringmann:2010sd, Bringmann:2018cov}.
It would be very interesting to understand the physical origin
of the coefficient $R_n$ in the non-holomorphic completion, which presumably arises
from a spectral asymmetry in the continuum of the superconformal field theory describing
wrapped five-branes, or alternatively  from boundaries of the moduli space
of anti-selfdual configurations in the gauge theory description.

\section*{Acknowledgements}
We thank Lothar G\"ottsche, Anton Mellit, Gregory Moore, Richard Thomas and Yan Soibelman for discussions or
correspondence, and Sibasish Banerjee for collaboration on the related works \cite{Alexandrov:2016tnf,Alexandrov:2016enp,Alexandrov:2017qhn}.
S.A. and J.M. are grateful to the Institut Henri Poincar\'e for financial support
under the Research in Paris program ``Quantum black holes and mock modular forms",
during which this project was initiated. The research of J.M. is
supported by IRC Laureate Award 15175 ``Modularity in Quantum Field
Theory and Gravity''.

\appendix

\section{Theta series and modularity}
\label{ap-thetamod}

\subsection{Theta series and refinement}
\label{ap-theta}

In this work we consider theta series of the following type
\be
\label{Vignerasth}
\vartheta_{\bfp,\bfmu}(\Phi,\lambda;\tau, \bfv)=\tau_2^{-\lambda/2}
\!\!\!\!
\sum_{{\bfq}\in \Lat+\bfmu+\hf\bfp}\!\!
(-1)^{\bfq\cdot\bfp}\,\Phi(\sqrt{2\tau_2}(\bfq+\bfb))\, \expe{- \tfrac{\tau}{2}\,\bfq^2+\bfq\cdot \bfv},
\ee
where $\bfv=\bfc-\tau\bfb$. Here $\Lat$ is a $d$-dimensional lattice equipped with a bilinear form
$(\bfx,\bfy)\equiv \bfx\cdot\bfy$, where $\bfx,\bfy\in \Lat \otimes \IR$, such that its associated quadratic form
has signature $(n,d-n)$ and is integer valued, i.e. $\bfq^2\equiv \bfq\cdot \bfq\in\IZ$ for $\bfq\in\Lat$.
Furthermore, $\bfp$ is a characteristic vector\footnote{A characteristic vector is
an element  $\bfp\in \Lat$ such that $\bfq\cdot(\bfq+ \bfp)\in 2\IZ$, $\forall \,\bfq \in \Lat$.
For distinct choices of characteristic vectors, the theta series are related by
$\vartheta_{\bfp,\bfmu} = (-1)^{(\bfmu+\frac12\bfp)\cdot(\bfp-\bfp')}
\vartheta_{\bfp',\bfmu+\frac12(\bfp-\bfp')}$. \label{foochar}},
$\bfmu\in\Lat^*/\Lat$ a glue vector, and $\lambda$ an arbitrary integer.
Provided the kernel $\Phi(\bfx)$ satisfies the following differential equation
\be
\label{Vigdif}
\Vop_\lambda\cdot \Phi(\bfx)=0,
\qquad
\Vop_\lambda= \partial_{\bfx}^2   + 2\pi \( \bfx\cdot \pa_{\bfx}  - \lambda\) ,
\ee
which we call Vign\'eras equation,
and subject to suitable decay conditions on $\Phi(\bfx) e^{\pi \bfx^2}$, then
under $SL(2,\IZ)$ transformations
\be\label{SL2Z}
\tau \mapsto \frac{a \tau +b}{c \tau + d} \, ,
\qquad
\begin{pmatrix} \bfc \\ \bfb \end{pmatrix} \mapsto
\begin{pmatrix} a & b \\ c & d  \end{pmatrix}
\begin{pmatrix} \bfc \\ \bfb \end{pmatrix} ,
\qquad \bfv \mapsto \frac{\bfv}{c\tau+d}
\ee
the theta series transforms
as a vector-valued Jacobi form of weight $(\lambda+d/2,0)$ and index $-1/2$ \cite{Vigneras:1977}.
Namely,
\bea
\vartheta_{\bfp,\bfmu}\(\Phi,\lambda; -\frac{1}{\tau}, \frac{\bfv}{\tau}\)
&=&\frac{(-\I\tau)^{\lambda+\frac{n}{2}}}{\sqrt{|\Lat^*/\Lat|}}\,
\expe{\frac14 \,\bfp^2-\frac{\bfv^2}{2\tau} } \sum_{\bfnu\in\Lat^*/\Lat}
\expe{\bfmu\cdot\bfnu}
\vartheta_{\bfp,\bfnu}(\Phi,\lambda;\tau, \bfv),
\nn\\
\vartheta_{\bfp,\bfmu}(\Phi,\lambda;\tau+1, \bfv)
&=&\expe{-\tfrac12 \,(\bfmu+\tfrac12 \bfp)^2}
\vartheta_{\bfp,\bfmu}(\Phi,\lambda;\tau, \bfv).
\label{eq:thetatransforms}
\eea
Note that \eqref{Vignerasth} differs from  the definition used in
\cite{Alexandrov:2012au,Alexandrov:2016tnf,Alexandrov:2017qhn,Alexandrov:2018lgp} by the factor
\be
V_\bfp=\expe{\hf\, \bfv\cdot\bfb}.
\label{Jacfac}
\ee
It changes the index of the theta series so that
$V_\bfp\vartheta_{\bfp,\bfmu}$ transforms as a standard modular form with vanishing index.

We are interested in the case where $\Lat=\oplus_{i=1}^n \Lambda_i$ and $\Lambda_i$ are charge lattices associated to divisors $\cD_i$.
Thus, the charges appearing in the description of the theta series \eqref{Vignerasth} are
of the type $\bfq=(q_1^a,\dots,q_n^a)$, whereas the vectors $\bfb$ and $\bfc$ are taken with $i$-independent components, namely,
$\bfb_i^a=b^a$, $\bfc_i^a=c^a$ for $i=1,\dots, n$, where $b^a$ and $c^a$ are the integrals of the
Neveu-Schwarz and Ramond-Ramond fields on the basis of two-cycles  $\omega^a\in H_2(\CY,\IZ)$.
The lattices $\Lambda_i$ carry the bilinear forms $\kappa_{i,ab}=\kappa_{abc}p_i^c$ which are all of signature $(1,b_2(\CY)-1)$.
This induces a natural bilinear form on $\Lat$
\be
\bfx\cdot\bfy=\sum_{i=1}^n (p_ix_iy_i)
\label{biform}
\ee
of signature $(n,nb_2-n)$.

Let us now study the effect of inserting a factor $y^{\sum_{i<j} \gamma_{ij}}$ into the
summand of the theta series.
Defining $y=\expe{w}$ and the $nb_2$-dimensional vector $\ptt$ with components
\be
\ptt_i^a=-\sum_{j<i}p_j^a+\sum_{j>i}p_j^a,
\label{vecptt}
\ee
this factor can be rewritten as $\expe{w \bfq\cdot\ptt}$. Hence, it
can be absorbed into a redefinition of the Jacobi variable,
$\hbfv=\bfv+w\ptt$,
\be
y^{\sum_{i<j} \gamma_{ij}}\,\expe{\bfq\cdot \bfv}
=
\expe{\bfq\cdot \hbfv} .
\ee
This observation suggests that under $SL(2,\IZ)$ transformations the parameter $w$ must transform
like $\bfv$, i.e.
\be
w\to\frac{w}{c\tau+d}\, .
\label{trans-w}
\ee
Equivalently,  upon decomposing $w$ as $w=\alpha-\tau\beta$,
the real parameters $\alpha$ and $\beta$ must transform as a doublet
\be
\begin{pmatrix} \alpha \\ \beta \end{pmatrix} \mapsto
\begin{pmatrix} a & b \\ c & d  \end{pmatrix}
\begin{pmatrix} \alpha \\ \beta \end{pmatrix},
\label{transab}
\ee
and similarly for the vectors
\be
\hbfb=\bfb+\beta\ptt,
\qquad
\hbfc=\bfc+\alpha\ptt.
\label{yshift}
\ee
Indeed, it follows from the theorem stated at the beginning of this section that
provided the kernel $\Phi(\bfx)$ satisfies Vign\'eras equation \eqref{Vigdif} and suitable
decay conditions,
then the following {\it refined} theta series
\bea
&&
\vthref_{\bfp,\bfmu}(\Phi,\lambda;\tau, \bfv,y)=
\vth_{\bfp,\bfmu}(\Phi,\lambda;\tau, \hbfv)
\label{thetaref}\\
&=&\tau_2^{-\lambda/2}
\!\!\!\!
\sum_{{\bfq}\in \Lat+\bfmu+\hf\bfp}\!\!
(-1)^{\bfq\cdot\bfp}\,\Phi\(\sqrt{2\tau_2}\(\bfq+\bfb+\frac{\Re(\log y )}{2\pi\tau_2}\,\ptt\)\)
\,y^{\sum_{i<j} \gamma_{ij}}\, \expe{- \tfrac{\tau}{2}\,\bfq^2+\bfq\cdot \bfv}
\nn
\eea
transforms as a vector-valued Jacobi form of weight $(\lambda+nb_2(\CY)/2,0)$ and  index
$-1/2$.
Note that modularity requires the $y$-dependent shift in the argument of the kernel which
leads to various important consequences.

\subsection{Generalized error functions}
\label{ap-generr}

In this appendix we recall the definition and some properties of the generalized error functions
introduced in \cite{Alexandrov:2016enp,Nazaroglu:2016lmr}
and revisited from a more conceptual viewpoint in \cite{kudla2016theta}, and prove a new identity which is used
in the study of the instanton generating potential.

First, we define the generalized (complementary) error functions
\bea
E_n(\cM;\vu)&=& \int_{\IR^n} \de \vu' \, e^{-\pi(\vu-\vu')^{\rm tr}(\vu-\vu')} \sign(\cM^{\rm tr} \vu'),
\label{generr-E}
\\
M_n(\cM;\vu)&=&\(\frac{\I}{\pi}\)^n |\det\cM|^{-1} \int_{\IR^n-\I \vu}\de^n z\,
\frac{e^{-\pi \vz^{\rm tr} \vz -2\pi \I \vz^{\rm tr} \vu}}{\prod(\cM^{-1}\vz)},
\label{generr-M}
\eea
where $\vz=(z_1,\dots,z_n)$ and $\vu=(u_1,\dots,u_n)$ are $n$-dimensional vectors, $\cM$ is $n\times n$ matrix of parameters,
and we used the shorthand notations $\prod \vz=\prod_{i=1}^n z_i$ and $\sign(\vu)=\prod_{i=1}^n \sign(u_i)$.
The detailed properties of these functions can be found in \cite{Nazaroglu:2016lmr}.
Here we just note that the information carried by $\cM$ is highly redundant. For instance, the generalized error functions
are invariant (up to sign) under rescaling of its columns. As a result, for $n=1$ the dependence on $\cM$ drops out, whereas
for $n=2$ (respectively $n=3$) they can always be expressed in terms of functions
parametrized only by one (respectively 3) parameters, e.g.
\be
\label{defE23}
E_2(\alpha;u_1,u_2):=  E_2\left(
\( \begin{array}{cc} 1 & 0 \\ \alpha & 1 \end{array}\); \vu\right),
\quad\
E_3(\alpha,\beta,\gamma;u_1,u_2,u_3):=  E_3\left(\(\begin{array}{ccc} 1 & 0  & 0 \\ \alpha & 1 & 0  \\
\beta & \gamma & 1 \end{array}\); \vu\right),
\ee
and similarly for $M_{2}$ and $M_3$. This parametrization will be used to express the explicit results
for modular completions in appendix \ref{ap-vwp2}. In the case of vanishing arguments one has \cite[Eq.(3.23)]{Alexandrov:2016tnf}
\be
E_2(\alpha;0,0)=\frac{2}{\pi}\, {\rm Arctan}(\alpha).
\label{E2zero}
\ee

Next, we define the boosted versions of the generalized error functions.
To write them down, let us consider $d\times n$ matrix $\cV$ which can be viewed as a collection of $n$ vectors,
$\cV=(\bfv_1,\dots,\bfv_n)$. We assume that these vectors span a positive definite subspace,
i.e. $\cV^{\rm tr}\cdot\cV$ is a positive definite matrix.
Let $\cB$ be $n\times d$ matrix whose rows define an orthonormal basis for this subspace.
Then we define the {\it boosted} generalized error functions
\be
\Phi_n^E(\cV;\bfx)=E_n(\cB\cdot \cV;\cB\cdot \bfx),
\qquad
\Phi_n^M(\cV;\bfx)=M_n(\cB\cdot \cV;\cB\cdot \bfx).
\label{generrPhiME}
\ee
Both types of generalized error functions can be shown to be independent of $\cB$ and
solve Vign\'eras equation \eqref{Vigdif}  {with $\lambda=0$}.
However, whereas $\Phi_n^E(\{\bfv_i\};\bfx)$ are smooth functions of $\bfx\in \IR^{n b_2(\CY)}$, asymptotic to
$\prod_{i=1}^n \sign (\bfv_i,\bfx)$, the complementary functions $\Phi_n^M(\{\bfv_i\};\bfx)$
are smooth only away from the real-codimension-1  loci  on $\IR^{n b_2}$, and are
exponentially suppressed for $|\bfx|\to\infty$.
The $\Phi_n^E$'s  provide the kernels for
modular completions of indefinite theta series, while the $\Phi_n^M$'s  provide the non-holomorphic
terms that must be added to reach that modular completions.

Note that the generalized error functions can be lifted to solutions of Vign\'eras equation with $\lambda\in\Nint$
by using the differential operator
\be
\cD(\bfv)=\bfv\cdot\(\bfx+\frac{1}{2\pi}\,\p_\bfx\)
\label{defcDif}
\ee
acting on the functions on $\IR^{n b_2(\CY)}$.
Its main feature is that it maps solutions of Vign\'eras equation with parameter $\lambda$ to another solution with $\lambda+1$.

An important fact is that the functions $\Phi_n^E$ can be expressed as linear combinations of products of
$\Phi_m^M$ and $n-m$ sign functions with $0\leq m\leq n$, {and vice-versa} \cite{Nazaroglu:2016lmr}:
\be
\Phi_n^E(\{\bfv_i\};\bfx)=\sum_{\cI\subseteq \Zv_n}\Phi_{|\cI|}^M(\{\bfv_i\}_{i\in\cI};\bfx)
\prod_{j\in \Zv_n\setminus \cI}\sign(\bfv_{j\perp \cI},\bfx),
\label{expPhiE}
\ee
where the sum goes over all possible subsets (including the empty set) of the set $\Zv_{n}=\{1,\dots,n\}$, $|\cI|$ is the cardinality of $\cI$,
and $\bfv_{j\perp \cI}$ denotes the projection of $\bfv_j$ orthogonal to the subspace spanned by $\{\bfv_i\}_{i\in\cI}$.
However, in the derivation of the integral form of the instanton generating potential we will need a slightly modified version of
this decomposition where $\bfx$ in the argument of the sign functions is shifted by a certain vector.
To state the corresponding result, let us introduce a modification of the boosted complementary generalized error function
replacing the function $M_n$ in its definition by a similar function with an additional shift of the integration contour
\bea
\hPhi_n^M(\cV;\bfx,\bfchi)&=&\hM_n(\cB\cdot \cV;\cB\cdot \bfx,\cB\cdot \bfchi),
\\
\hM_n(\cM;\vu,\vchi)&=&\(\frac{\I}{\pi}\)^n |\det\cM|^{-1} \int_{\IR^n-\I (\vu-\vchi)}\de^n z\,
\frac{e^{-\pi \vz^{\rm tr} \vz -2\pi \I \vz^{\rm tr} \vu}}{\prod(\cM^{-1}\vz)}\, .
\eea
Then we have

\begin{proposition}
\be
\Phi_n^E(\{\bfv_i\};\bfx)=\sum_{\cI\subseteq \Zv_n}\hPhi_{|\cI|}^M(\{\bfv_i\}_{i\in\cI};\bfx,\bfchi)
\prod_{j\in \Zv_n\setminus \cI}\sign(\bfv_{j\perp \cI},\bfx-\bfchi),
\label{expPhiE-mod}
\ee
\label{prop-decomp}
\end{proposition}
\begin{proof}
First, we note that the dependence on $\bfchi$ of the right-hand side of \eqref{expPhiE-mod} is locally constant
because the integration contours can be safely deformed provided they do not cross the poles of the integrands.
Next we note that the smoothness of $\Phi_n^E$ implies that all discontinuities due to signs and contour integrals in \eqref{expPhiE} cancel.
Then the same should be true for the right-hand side of \eqref{expPhiE-mod} as well. Indeed, the shift induced by $\bfchi$ just changes
the position of the discontinuities of both signs and integrals in the same way, and it does not affect the jumps of individual terms
since the integrands are independent of $\bfchi$. It is clear that smoothness in $\bfx$ also implies the smoothness in $\bfchi$.
Combined with the above fact that the dependence on $\bfchi$ is locally constant,
one obtains that the resulting function is actually constant in $\bfchi$ and hence coincides with $\Phi_n^E$.
\end{proof}

\subsection{Matrix of parameters}
\label{ap-matrix}

The main building blocks of the construction proposed in this work are the boosted generalized error functions $\Phi^E_{n-1}(\{ \bfv_{\ell}\};\bfx)$
where the vectors $\bfv_\ell$ are defined in \eqref{defbfvk}.
In this appendix we express them through the original generalized error functions \eqref{generr-E}.

According to \eqref{generrPhiME}, we should find the matrix $\cB$ representing an orthonormal basis in the subspace spanned by $\bfv_\ell$
and evaluate the scalar products $\cB\cdot \cV$ and $\cB\cdot\bfx$. Note that the vectors $\bfv_\ell$ coincide with the vectors $\bfv_e$ \eqref{defue}
computed for the {trivial} unrooted tree
$\bullet\!\mbox{---}\!\bullet\!\mbox{--}\cdots \mbox{--}\!\bullet\!\mbox{---}\!\bullet$,
with vertices labelled by charges $\gamma_1,\dots,\gamma_n$ consecutively.
In \cite[Appendix E]{Alexandrov:2018lgp} it was shown that for any set of vectors defined by an unrooted tree, an orthonormal basis
can be constructed from a rooted ordered binary tree with leaves labelled by the charges,
which is derived in a certain way from the initial unrooted tree.
Namely \cite[Lemma 2]{Alexandrov:2018lgp}:\footnote{That Lemma was proven actually for the moduli dependent vectors $\bfu_\ell$.
However, it is easy to see that the same results apply to $\bfv_\ell$ upon replacement of vectors $\bfu_{ij}$ by $\bfv_{ij}$
and combinations $(p_it^2)(p_jt^2)(p_kt^2)$ by $(p_ip_jp_k)$.}
\be
\begin{split}
\cB=&\, \(\frac{\tbfv_1}{\sqrt{\Delta_1}},\dots, \frac{\tbfv_{n-1}}{\sqrt{\Delta_{n-1}}}\)^{\rm tr},
\\
\cV=&\, \Bigl(\sqrt{\Delta_1}\bfv_1,\dots, \sqrt{\Delta_{n-1}}\bfv_{n-1}\Bigr),
\end{split}
\label{matrix-gener}
\ee
where
\be
\begin{split}
\tbfv_k=&\,\sum_{i\in\cI_{\Lv{v_k}}}\sum_{j\in\cI_{\Rv{v_k}}}\bfv_{ij},
\\
\Delta_k=&\,\tbfv_k^2=(p_{v_k} p_{\Lv{v_k}}p_{\Rv{v_k}}),
\end{split}
\label{dataechoice}
\ee
$\Lv{v}$, $\Rv{v}$ are the two children of the vertex $v$, and $\cI_v$ is the set of leaves which are descendants of $v$.
In our case, one can choose the binary tree to be as in Fig. \ref{fig-basis0}.
Then one finds
\be
\begin{split}
\tbfv_k\cdot \bfv_\ell=&\,\left\{\begin{array}{cc}
\sum_{j=1}^{\ell}(p p_j p_{k+1}),
&\ k\ge \ell,
\\
0, & \ k<\ell,
\end{array}\right.
\\
\Delta_k=&\,\sum_{i=1}^k \sum_{j=1}^{k+1}(p_i p_j p_{k+1}),
\end{split}
\label{vectorstree}
\ee
while the matrix of parameters is lower triangular, given by
\be
\cM_{k\ell}=\left\{\begin{array}{cc}
 \sqrt{\frac{\Delta_\ell}{\Delta_k}} \, \tbfv_k\cdot \bfv_\ell,
&\ k\ge \ell,
\\
0, & \ k<\ell.
\end{array}\right.
\label{matM}
\ee
For the special case of equal charges where $p_i=p_0$ for $i=1,\dots,n$, relevant for the discussion around \eqref{EnM0},
the matrix reduces to
\be
\cM^{(0)}_{k\ell}=\left\{\begin{array}{cc}
n \ell \sqrt{\frac{\ell(\ell+1)}{k(k+1)}}\, p_0^3,
&\ k\ge \ell,
\\
0, & \ k<\ell.
\end{array}\right.
\label{matM0}
\ee

\lfig{The rooted binary tree encoding the orthonormal basis $\cB$ and corresponding to the unrooted tree of trivial topology shown at the bottom.}
{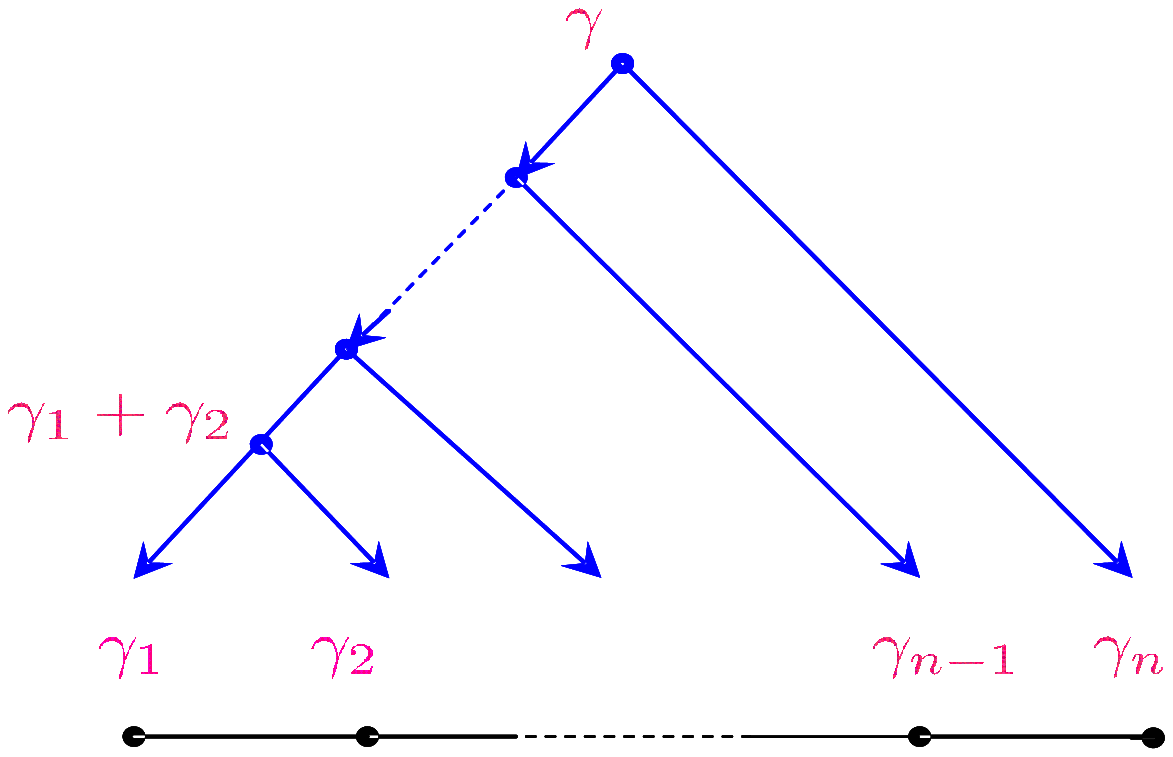}{7.9cm}{fig-basis0}{-1cm}

\section{Relevant functions}
\label{ap-E}

In this section, we provide the definition of various functions determining the completion $\whh_{p,\mu}$
and the theta series decomposition of the instanton generating potential $\cG$.

In fact, all these functions are uniquely determined by one set of functions $\gf_n(\{\gama_i,c_i\})$.
To define {those}, we take $\IT_{n,m}^\ell$ to be the set of {\it marked} unrooted labelled trees with $n$ vertices and $m$ marks assigned to vertices.
Let $m_\ver\in \{0,\dots m\}$ be the number of marks carried by the vertex $\ver$, so that $\sum_\ver m_\ver=m$.
Furthermore, the vertices are decorated by charges from the set $\{\gamma_1,\dots,\gamma_{n+2m}\}$ such that a vertex $\ver$ with
$m_\ver$ marks carries $1+2m_\ver$ charges $\gamma_{\ver,s}$, $s=1,\dots,1+2m_\ver$ and we set $\gamma_\ver=\sum_{s=1}^{1+2m_\ver}\gamma_{\ver,s}$.
Given a tree $\cT\in \IT_{n,m}^\ell$, we denote  the set of its edges by $E_{\cT}$, the set of vertices by $V_{\cT}$, and
the source and target vertex of an edge $e$ by $s(e)$ and $t(e)$, respectively.\footnote{The orientation of edges on a given tree
can be chosen arbitrarily, the final result does not depend on this choice.}
Then $\gf_n$ is given by a sum over marked unrooted labelled trees as follows \cite[Eq.(5.27)]{Alexandrov:2018lgp}
\be
\gf_n(\{\gama_i,c_i\})=\frac{(-1)^{n-1+\sum_{i<j} \gamma_{ij} }}{2^{n-1} n!}
\sum_{m=0}^{[(n-1)/2]}\sum_{\cT\in\, \IT_{n-2m,m}^\ell}
\prod_{\ver\in V_\cT}\tcV_{m_\ver}(\{\gama_{\ver,s}\})
\prod_{e\in E_{\cT}}\gamma_{s(e) t(e)}\,\sign(\cs_e),
\label{defDf-gen}
\ee
where $\cs_k=\sum_{i=1}^k c_i$  and
\be
\tcV_{m}(\{\gama_{s}\})=\sum_{\cT'\in\, \IT_{2m+1}^\ell} a_{\cT'}\prod_{e\in E_{\cT'}}\gamma_{s(e)t(e)}.
\label{deftcV-mw}
\ee
Here for each tree $\cT$ we introduced rational coefficients $a_\cT$ determined recursively by the relation
\be
a_\cT=\frac{1}{n_\cT}\sum_{\ver\in V_\cT} \epsilon_\ver\prod_{s=1}^{n_\ver} a_{\cT_s(\ver)},
\label{res-aT}
\ee
where $n_\cT$ is the number of vertices, $n_\ver$ is the valency of the vertex $\ver$,
$\cT_s(\ver)$ are the trees obtained from $\cT$ by removing this vertex, and $\epsilon_\ver$ is the sign determined
by the choice of orientation of edges,
$
\epsilon_\ver=(-1)^{n_\ver^+}
$
with $n_\ver^+$ being the number of incoming edges at the vertex.
Finally, we use the following definition of the sign function
\be
\sgn(x)=\left\{  \begin{array}{cl}  -1, &\quad x<0, \\ 0, & \quad x=0, \\ 1, &\quad x>0.       \end{array}  \right.
\label{defsign}
\ee

Given $\gf_n$, all other functions can be obtained via the following procedure:
\begin{itemize}
\item\label{procfun}
Setting the stability parameters to the attractor values, $c_i=\beta_{ni}$,
one obtains moduli-independent functions $\cEf_n(\{\gama_i\})$, see \eqref{rel-gE}.

\item
Dividing both $\gf_n$ and $\cEf_n$ by a factor $\frac{(-1)^{\sum_{i<j} \gamma_{ij} }}{(\sqrt{2\tau_2})^{n-1}}$,
one finds that the resulting functions depend on D2-brane charges, $\tau_2$ and the real part $b^a$ of the \kahler moduli
only through the combinations
\be
x_i^a=\sqrt{2\tau_2}(\kappa_i^{ab} q_{i,b}+b^a) ,
\label{argx}
\ee
where $\kappa_i^{ab}$ is the inverse of the quadratic form $\kappa_{i,ab} = \kappa_{abc} p_i^c$.
Therefore, they can be viewed as kernels of theta series of the type considered in appendix \ref{ap-theta}
with $x_i^a$ being the components of $nb_2(\CY)$-dimensional vector $\bfx$.

\item
By adding contributions exponentially suppressed at large $\bfx$, these kernels can be promoted to
smooth solutions of Vign\'eras equation \eqref{Vigdif} {with $\lambda=0$},
which we call $\tcEPhi_n$ and $\cEPhi_n$, respectively.

\item
Finally, restoring the factor $\frac{(-1)^{\sum_{i<j} \gamma_{ij} }}{(\sqrt{2\tau_2})^{n-1}}$, one defines $\cE_n$ by
\be
\cE_n(\{\gama_i\},\tau_2)=
\frac{(-1)^{\sum_{i<j} \gamma_{ij} }}{(\sqrt{2\tau_2})^{n-1}}\, \cEPhi_n(\bfx).
\label{rescEnPhi}
\ee
\end{itemize}

To present the results for $\tcEPhi_n$ and $\cEPhi_n$ following from \eqref{defDf-gen},
we have to define several sets of $nb_2(\CY)$-dimensional vectors.
The two basic sets are given by
\be
\begin{split}
(\bfv_{ij})_k^a=&\, \delta_{ki} p_j^a-\delta_{kj} p_i^a \qquad\qquad\quad\
\mbox{such that} \quad \bfv_{ij}\cdot\bfx=(p_ip_j(x_i-x_j)),
\\
(\bfu_{ij})_k^a=&\, \delta_{ki}(p_jt^2)t^a-\delta_{kj} (p_it^2)t^a \quad
\mbox{such that} \quad
\bfu_{ij}\cdot\bfx=(p_jt^2)(p_ix_it)-(p_it^2)(p_jx_jt),
\end{split}
\label{defvij}
\ee
where $t^a=\Im z^a$ and the bilinear form is defined in \eqref{biform}.
For $\bfx$ as in \eqref{argx}, $\bfv_{ij}\cdot\bfx=\sqrt{2\tau_2}\gamma_{ij}$ and
$\bfu_{ij}\cdot\bfx=-2\sqrt{2\tau_2}\Im\[Z_{\gamma_i}\bZ_{\gamma_j}\]$.
Furthermore, for a tree $\cT\in\IT_{n,m}^\ell$, denote by $\cT_e^s$ and $\cT_e^t$ the two disconnected trees obtained from
$\cT$ by removing the edge $e$.
Then we introduce another two sets of vectors
\be
\bfv_e=\sum_{i\in V_{\cT_e^s}}\sum_{j\in V_{\cT_e^t}}\bfv_{ij},
\qquad
\bfu_e=\sum_{i\in V_{\cT_e^s}}\sum_{j\in V_{\cT_e^t}}\bfu_{ij}.
\label{defue}
\ee

With these definitions, one has
\cite[Eq.(5.32)]{Alexandrov:2018lgp}
\be
\begin{split}
\cEPhi_n(\bfx)=&\,
\frac{1}{2^{n-1} n!}\sum_{m=0}^{[(n-1)/2]}
\sum_{\cT\in\, \IT_{n-2m,m}^\ell}
\[\prod_{\ver\in V_\cT}\cD_{m_\ver}(\{\gama_{\ver,s}\})\]
\\
&\, \times \[\prod_{e\in E_\cT} \cD(\bfv_{s(e) t(e)})\]
\Phi^E_{n-2m-1}(\{ \bfv_e\};\bfx),
\end{split}
\label{rescEn}
\ee
where
\be
\cD_{m}(\{\gama_s\})=\sum_{\cT'\in\, \IT_{2m+1}^\ell} a_{\cT'}\prod_{e\in E_{\cT'}}\cD(\bfv_{s(e)t(v)})
\label{defcDcT}
\ee
is a differential operator constructed from \eqref{defcDif} and $\Phi^E_n$ are (boosted) generalized error functions reviewed
in appendix \ref{ap-generr}.
The functions $\tcEPhi_n(\bfx)$ are given by the same expression with the vectors $\bfv_e$
appearing as parameters in $\Phi^E_{n-2m-1}$ replaced by $\bfu_e$.

\medskip

Note that in \cite{Alexandrov:2018lgp} it was shown that all contributions of trees with a non-zero number of marks {remarkably} cancel
in the sum over Schr\"oder trees like \eqref{solRn} or \eqref{soliterg}.
Therefore, we could omit them from the very beginning arriving at a simpler set of functions $\gf_n$, $\cE_n$, $\tcEPhi_n$ and $\cEPhi_n$.
However, it is the function \eqref{defDf-gen} with contributions of marks included
that is reproduced in the  limit $y\to 1$ of $\gref_n$ defined in \eqref{whgF}.

\section{Refined instanton generating potential}
\label{ap-G}

In this appendix we rewrite the theta series decomposition \eqref{treeFh-flhr} of the refined instanton generating potential
as a sum of iterated integrals of the same type which arise in the unrefined case.
To this end, we {retrace the steps} taken in \cite{Alexandrov:2018lgp}, which allowed to rewrite the unrefined potential
as in \eqref{treeFh-flh}.

First, we rewrite $\cGr$ as an expansion in the holomorphic generating functions $\hr_{p_i,\mu_i}$.
This changes the kernels of the theta series, which now fail to be modular due to the modular anomaly of $\hr_{p,\mu}$.
The result (proven below in \S\ref{ap-proofG})\footnote{A similar statement in the unrefined case was stated as Conjecture 1 in \cite{Alexandrov:2018lgp}.
Here we prove this claim in a more general situation.}
is given by
\be
\cGr=\frac{1}{2\pi\sqrt{\tau_2}}\sum_{n=1}^\infty\[\prod_{i=1}^{n}
\sum_{p_i,\mu_i}\sigma_{p_i}\hr_{p_i,\mu_i}\]
e^{-S^{\rm cl}_p}y^{-\beta m(p)} V_\bfp\,\vthref_{\bfp,\bfmu}\bigl(\Phi^{\rm ref}_{n},-1\bigr),
\label{treeFh-r}
\ee
where
\be
\Phi^{{\rm ref}}_n(\bfx)=\sum_{n_1+\cdots n_m=n \atop n_k\ge 1}  \intPhi_m(\bfx')
\prod_{k=1}^m \trPhi_{n_k}(x_{j_{k-1}+1},\dots,x_{j_k}).
\label{totker-ref}
\ee
The first factor in \eqref{totker-ref},
\be
\intPhi_n(\bfx)=
\frac{\intPhi_1(x)}{2^{n-1}}\,
\hPhi^M_{n-1}(\{ \bfu_\ell\};\bfx,\sqrt{2\tau_2}\beta\ptt),
\label{Phin-final}
\ee
is constructed from the function \eqref{defPhi1-main} and the modified version of the complementary
error functions introduced in appendix \ref{ap-generr}.
In the above formula it appears with the index $m$ equal to the number of {parts in the
partition of $n$}.
Correspondingly, it depends on $\bfx'$ and $\bfp'$ (the later dependence is not indicated explicitly)
which are both $mb_2(\CY)$-dimensional vectors with components
\be
p'^a_k=\sum_{i=j_{k-1}+1}^{j_k}p^a_i,
\qquad
x'^a_k=\kappa'^{ab}_k\sum_{i=j_{k-1}+1}^{j_k} \kappa_{i,bc} x^c_i,
\label{defprimevar}
\ee
where $\kappa'_{k,ab}=\kappa_{abc}p'^c_k$ and $j_k$ are defined below \eqref{defFref}.
The other factors in \eqref{totker-ref} are given by
\be
\begin{split}
\trPhi_{n}(\bfx)
=&\,
(-y)^{-\sum_{i<j} \gamma_{ij} }
\sum_{T\in\IT_n^{\rm S}}(-1)^{n_T-1} \(\gfr_{v_0}-\cErf_{v_0}\)\prod_{v\in V_T\setminus{\{v_0\}}}\cErf_{v},
\end{split}
\label{kerPhiintr}
\ee
where $\bfx$ is related to charges via \eqref{argxref}. Comparing with \eqref{soliterg-tr}, (the symmetrization of) these functions
can be recognized as a rescaled version of the refined tree index relating
the refined DT and MSW invariants. Namely,
\be
\label{Omsumtreeref}
\bOm(\gamma,z^a,y) =
\sum_{\sum_{i=1}^n \gamma_i=\gamma}
\frac{(-y)^{\sum_{i<j} \gamma_{ij} }}{(y-y^{-1})^{n-1}}\,\trPhi_{n}(\bfx)\,
\prod_{i=1}^n \bOmMSW(\gamma_i,y),
\ee
where the symmetrization is ensured by the sum over charges.
Note that the power of $y-y^{-1}$ disappears once this relation is rewritten in terms
of the generating functions \eqref{defchimur} and \eqref{defhDTr}.

Given the relation \eqref{Omsumtreeref}, $\cGr$ can be rewritten as an expansion in the generating functions of refined DT invariants
$\hrDT_{p_i,q_i}$. It is easy to see that such expansion is given by
\be
\cGr=\frac{1}{2\pi}\sum_{n=1}^\infty\[\prod_{i=1}^{n}
\sum_{p_i,q_i}\sigma_{p_i,q_i}\hrDT_{p_i,q_i}  \cXt_{p_i,q_i} \]
y^{-\beta m(p)+\sum_{i<j} \gamma_{ij} }\intPhi_n(\bfx),
\label{expGrDT}
\ee
where $\sigma_{p,q}\equiv\sigma_\gamma =\expe{\hf\, p^a q_a}\sigma_p$ is the quadratic refinement specified for our set of charges and
\be
\cXt_{p,q} = e^{-S^{\rm cl}_p}\,
\expe{- \frac{\tau}{2}\,(q+b)^2+c^a (q_a +\haf (pb)_a)}
\label{Xthetapq}
\ee
is a combination of three contributions evaluated for a single charge:
the classical D3-brane action, the exponential defining the theta series \eqref{Vignerasth},
and the factor $V_p$ \eqref{Jacfac}.

Next, we use the result proven in appendix E of \cite{Alexandrov:2018lgp} which states that
for any unrooted labelled tree $\cT$ one has the following identity
\be
\frac{\intPhi_1(x)}{2^{n-1}}\, \Phi^M_{n-1}(\{ \bfu_e\};\bfx)=
\frac{\I^{n-1}}{(2\pi)^n}\[\prod_{i=1}^n\int_{z_i^\star+\IR}\de z_i \, W_{p_i}(x_i,z_i)\]
\frac{1}{\prod_{e\in E_\cT}\(z_{s(e)}-z_{t(e)}\)}\, ,
\ee
where
\be
W_{p}(x,z)=e^{-2\pi\tau_2 z^2(pt^2)-2\pi\I\sqrt{2\tau_2}\, z\, (pxt)}
\ee
and
$
z_i^\star=- \frac{\I\,(p_ix_i t)}{\sqrt{2\tau_2}(p_it^2)}
$
is the saddle point {governing} the integral over $z_i$.
On the left hand side the data about $\cT$ are encoded in the set of vectors $\bfu_e$ \eqref{defue}.
In our case this set is given by $\bfu_\ell$ defined in \eqref{shiftS}
which can be seen as vectors $\bfu_e$ for the {trivial} unrooted tree
$\bullet\!\mbox{---}\!\bullet\!\mbox{--}\cdots \mbox{--}\!\bullet\!\mbox{---}\!\bullet\,$.
Furthermore, it is easy to see that if one replaces $\Phi^M_{n-1}$ by its modified version $\hPhi^M_{n-1}$
appearing in \eqref{Phin-final},
on the right-hand side one simply changes $z_i^\star$ by
\be
z_{\gamma_i}=- \frac{\I\,(p_i(x_i-\sqrt{2\tau_2}\beta\ptt_i) t)}{\sqrt{2\tau_2}(p_it^2)}
=-\I\,\frac{(q_a+(pb)_a)\,t^a}{(pt^2)}\, .
\ee
Therefore, we conclude that the functions $\intPhi_n$ can be represented in the following integral form
\be
\intPhi_n(\bfx) =
\[\prod_{i=1}^n\int_{z_{\gamma_i}+\IR}\frac{\de z_i}{2\pi} \, W_{p_i}(x_i,z_i)\]
\frac{\I^{n-1}}{\prod_{i=1}^{n-1}(z_i-z_{i+1})}\, .
\label{kerPhinr}
\ee
Besides, due to the shift of the $b$-field produced by the refinement, one has
\be
\prod_{i=1}^n W_{p_i}(\sqrt{2\tau_2}(q_i+\hat b_i),z_i)=\expe{-2\tau_2\beta\sum_{i<j} (p_ip_jt)(z_i-z_j)}
\prod_{i=1}^n W_{p_i}(\sqrt{2\tau_2}(q_i+ b_i),z_i).
\label{shiftW}
\ee
Furthermore, since we work in the large volume limit $t^a\to \infty$, the contours $z_{\gamma_i}+\IR$ can be deformed into the standard BPS rays
$\ell_{\gamma_i}$ \cite{Gaiotto:2008cd,Alexandrov:2008gh}
which in the $z$-plane go along the arcs running from $-1$ to $1$ and passing through $z_{\gamma_i}$.
Thus, substituting \eqref{kerPhinr} and \eqref{shiftW} into \eqref{expGrDT} and taking into account the definition of $\hrDT_{p,q}$,
one obtains the representation \eqref{treeFh}, where we decomposed the $\beta$-dependent power of $y$ as in the original formula \eqref{treeFh-flhr}.

\subsection{Proof of Eq.(C.1)}
\label{ap-proofG}

Substituting the explicit formula for the completion of the generating function \eqref{exp-whhr} into the refined potential \eqref{treeFh-flhr} and
rewriting it as an expansion in powers of $\hr_{p_i,\mu_i}$, it is easy to see that one gets \eqref{treeFh-r} where the kernel $\Phi^{{\rm ref}}_n$
is given by
\be
\Phi^{{\rm ref}}_n(\bfx)=\!\!\sum_{n_1+\cdots n_m=n \atop n_k\ge 1} \!\! \whPhi^{\rm ref}_{m}(\bfx')
\prod_{k=1}^m \[
\sum_{T_k\in\IT_{n_k}^{\rm S}}(-1)^{n_{T_k}-1} \Phip_{v_0}
\!\!\prod_{v\in V_{T_k}\setminus{\{v_0\}}}\!\! \Phif_{v}\],
\label{totker-ref1}
\ee
where we introduced the rescaled versions of $\cErf_{n}$ and $\cErp_{n}$\footnote{The shift of $\bfx$ in the argument of the sign function
compensates the shift in \eqref{argxref} so that the resulting function is $\beta$-independent.}
\be
\Phif_{n}=\frac{1}{2^{n-1}}\prod_{k=1}^{n-1}\sgn(\bfv_k,\bfx-\sqrt{2\tau_2}\,\beta \ptt),
\qquad
\Phip_{n}=\Phi_{n}-\Phif_{n},
\ee
the trees $T_k$ are labelled by the charges $\gama_{j_{k-1}+1},\dots,\gama_{j_{k}}$,
whereas the first factor depends on the sums of charges in each subset $\gama'_k=\gama_{j_{k-1}+1}+\cdots +\gama_{j_{k}}$,
or equivalently on the $mb_2$-dimensional vectors \eqref{defprimevar}.
Taking into account the explicit form of $\whPhi^{\rm ref}_{n}$ \eqref{kertotsolr}, we see that the kernel \eqref{totker-ref1}
can be visualized as a sum over Schr\"oder trees, {the leaves of which themselves
sprouting further Schr\"oder trees}.
Regarding all these trees as parts of one big tree, one arrives at the following representation
\be
\Phi^{{\rm ref}}_n=
\intPhi_1\sum_{T\in\IT_n^{\rm S}}(-1)^{n_{T}-1}\(\tPhi_{v_0}-\Phi_{v_0}\)\sum_{T'\subseteq T}
\prod_{v\in V_{T'}\setminus{\{v_0\}}}\Phi_{v}
\prod_{v\in L_{T'}\setminus L_T} (-\Phip_{v})
\prod_{v\in V_T\setminus (V_{T'}\cup L_{T'})} \Phif_{v},
\label{totker-ref2}
\ee
where the second sum goes over all subtrees $T'$ of $T$ containing its root\footnote{The sum also
includes the contribution of the trivial subtree in which case
the product over vertices should read $ \Phip_{v_0}\prod_{v\in V_T} \Phif_{v}$. \label{foot-root}}
and $L_{T}$ denotes the set of leaves of $T$.

Let us now fix a tree $T$ and a vertex $v$ whose only children are leaves of $T$, i.e. $v$ has height 1.
Then for each subtree $T'$ with $v\in L_{T'}$ one can put into a correspondence another subtree for which the children of $v$ are added to the subtree,
i.e. now $v\in V_{T'}$ and the rest of $T'$ is the same.
The contributions of two such subtrees in \eqref{totker-ref2} differ only by the factor assigned to the vertex $v$:
it is $-\Phip_{v}$ in the first contribution, whereas $\Phi_{v}$ in the second.
Thus, the two contribution recombine giving $\Phif_{v}$ as the weight assigned to the vertex.

After performing such recombination for all vertices of height 1, one moves to the next level and considers $v$ of height 2.
Here again one picks up pairs of subtrees with $v\in L_{T'}$ and $v\in V_{T'}$, respectively.
But now the contribution of the latter is already the one after the recombination done at the first step.
As a result, the two contributions again differ only by the factor assigned to the vertex $v$ and the result of their recombination
is the same as above: the new factor is $\Phif_{v}$.

In this way one covers all vertices of $T$ up to the root. At the root one again compares two contributions:
one of the trivial subtree (see footnote \ref{foot-root}) and another one from all previous recombinations.
Their sum leads to the weight $\tPhi_{v_0}-\Phif_{v_0}$ assigned to the root. As a result, one remains with the following kernel
\be
\Phi^{{\rm ref}}_n=
\intPhi_1\sum_{T\in\IT_n^{\rm S}}(-1)^{n_{T}-1}\(\tPhi_{v_0}-\Phif_{v_0}\)
\prod_{v\in V_{T}\setminus{\{v_0\}}}\Phif_{v}.
\label{totker-ref3}
\ee
{It it worth noting that this formula makes it manifest that the kernel $\Phi^{{\rm ref}}_n$ is smooth
across walls of marginal stability, since all moduli dependence comes from $\tPhi_{v_0}$.}

\medskip

Next, we should take into account the relation \eqref{expPhiE-mod}
between functions $\Phi_n^E$ and $\hPhi_n^M$ with $\bfchi=\sqrt{2\tau_2}\beta\ptt$.
{In our case} $n=k_{v_0}-1$ where $k_{v_0}$ is the number of children of the root vertex, and the vectors $\bfv_i$ coincide with $\bfu_\ell$
defined in \eqref{shiftS}. For such vectors one has

\begin{proposition}
Let $\cI=\{j_1,\dots,j_{m-1}\}$ where $0\equiv j_0<j_1<\cdots <j_{m-1}< j_{m}\equiv n$. For $\ell\in \Zv_{n-1}\setminus \cI$
find $k$ such that $j_{k-1}< \ell<j_k$. Then one has
\be
\bfu_{\ell \perp \cI}=
\frac{(pt^2)}{\sum_{i=j_{k-1}+1}^{j_k} (p_it^2)}\,\bfu^{j_{k-1}+1, j_k}_\ell
\qquad \mbox{where}\qquad
\bfu^{ij}_\ell=\sum_{k=i}^\ell\sum_{k'=\ell+1}^j \bfu_{kk'},
\quad
i\le \ell< j.
\nn
\ee
\label{prop-orthu}
\end{proposition}

\begin{proof}

First, we prove the statement for the case of the set $\cI$ consisting of one element which we formulate as a Lemma.

\begin{lemma}
\be
\bfu_{\ell \perp k}= \left\{
\begin{array}{ll}
\frac{(pt^2)}{\sum_{i=1}^{k} (p_it^2)}\,\bfu^{1k}_\ell, & \ell<k,
\\
\frac{(pt^2)}{\sum_{i=k+1}^n (p_it^2)}\,\bfu^{k+1,n}_\ell,\quad & \ell>k.
\end{array}
\right.
\nn
\ee

\label{lemma-orthu}
\end{lemma}

\begin{proof}
First, let us consider the case $n=3$.
A straightforward calculation gives
\be
\begin{split}&
\bfu_1^2=(pt^2)(p_1t^2)(p_{2+3}t^2),
\qquad
\bfu_2^2=(pt^2)(p_{1+2}t^2)(p_3t^2),
\\ &\qquad\qquad\qquad
(\bfu_1,\bfu_2)=(pt^2)(p_1t^2)(p_3t^2),
\end{split}
\ee
where we introduced the convenient notations $p^a_{i+j}=p^a_i+p^a_j$.
Using these results, one finds
\be
\bfu_{1\perp 2}=\frac{(pt^2)}{(p_{1+2}t^2)}\,\bfu_{12},
\qquad
\bfu_{2\perp 1}=\frac{(pt^2)}{(p_{2+3}t^2)}\,\bfu_{23},
\ee
which agrees with the statement of the Lemma since $\bfu_1^{12}=\bfu_{12}$, $\bfu_2^{23}=\bfu_{23}$.

The case of arbitrary $n$ then reduces to the case $n=3$ by identifying
\be
\tilde \gamma_1=\sum_{i=1}^\ell \gamma_i,
\qquad
\tilde \gamma_2=\sum_{i=\ell+1}^{k} \gamma_i,
\qquad
\tilde \gamma_3=\sum_{i=k+1}^n \gamma_i,
\ee
for $\ell<k$, or
\be
\tilde \gamma_1=\sum_{i=1}^{k} \gamma_i,
\qquad
\tilde \gamma_2=\sum_{i=k+1}^{\ell} \gamma_i,
\qquad
\tilde \gamma_3=\sum_{i=\ell+1}^n \gamma_i,
\ee
in the opposite case.
\end{proof}

If $\cI$ has several elements, let us find $k$ as in the statement of the Proposition.
We note that the projection on the subspace orthogonal to the span of $\{\bfu_i\}_{i\in\cI}$
can be equivalently obtained by first projecting with respect to $\bfu_{j_k}$ and then with respect to $\{\bfu_{i\perp j_k}\}_{i\in\cI\setminus\{j_k\}}$.
According to the Lemma, the latter set is equivalent to $\{\bfu_{j_\ell}^{1j_k}\}_{\ell=1}^{k-1}\cup\{\bfu_{j_\ell}^{j_k+1, n}\}_{\ell=k+1}^{m-1} $,
whereas the first projection gives us
$\bfu_{\ell \perp j_k}=\frac{(pt^2)}{\sum_{i=1}^{j_k} (p_it^2)}\,\bfu^{1j_k}_\ell$.
Since $\bfu^{1j_k}_\ell$ is already orthogonal to any $\bfu_{j_\ell}^{j_k+1, n}$, it remains only to do the orthogonal projection with respect to
$\{\bfu_{j_\ell}^{1j_k}\}_{\ell=1}^{k-1}$.\footnote{If $k=1$, there is already nothing to do. Similarly, if $k=n-1$ one omits the previous step.}
This projection we again split into two steps: with respect to $\bfu_{j_{k-1}}^{1j_k}$ and $\{\bfu_{j_\ell\perp j_{k-1}}^{1j_k}\}_{\ell=1}^{k-2} $.
The projections can be evaluated using the Lemma provided one replaces $n$ by $j_k$ since all components beyond $j_k$ vanish.
As a result, the first projection gives
$\bfu^{1j_k}_{\ell\perp j_{k-1}}=\frac{\sum_{i=1}^{j_k}(p_it^2)}{\sum_{i=j_{k-1}+1}^{j_k} (p_it^2)}\,\bfu^{j_{k-1}+1,j_k}_{\ell} $,
whereas the remaining set of vectors is equivalent to the span of $\{\bfu_{j_\ell}^{1j_{k-1}}\}_{\ell=1}^{k-2} $.
All these vectors are orthogonal to $\bfu^{j_{k-1}+1,j_k}_{\ell}$ and therefore there is no need to do any further projection.
Combining the prefactors, one recovers the statement of the Proposition.
\end{proof}

\lfig{An example of the effect of the partition on the tree for the case where the root has 6 children and the partition is $6=2+1+3$.}
{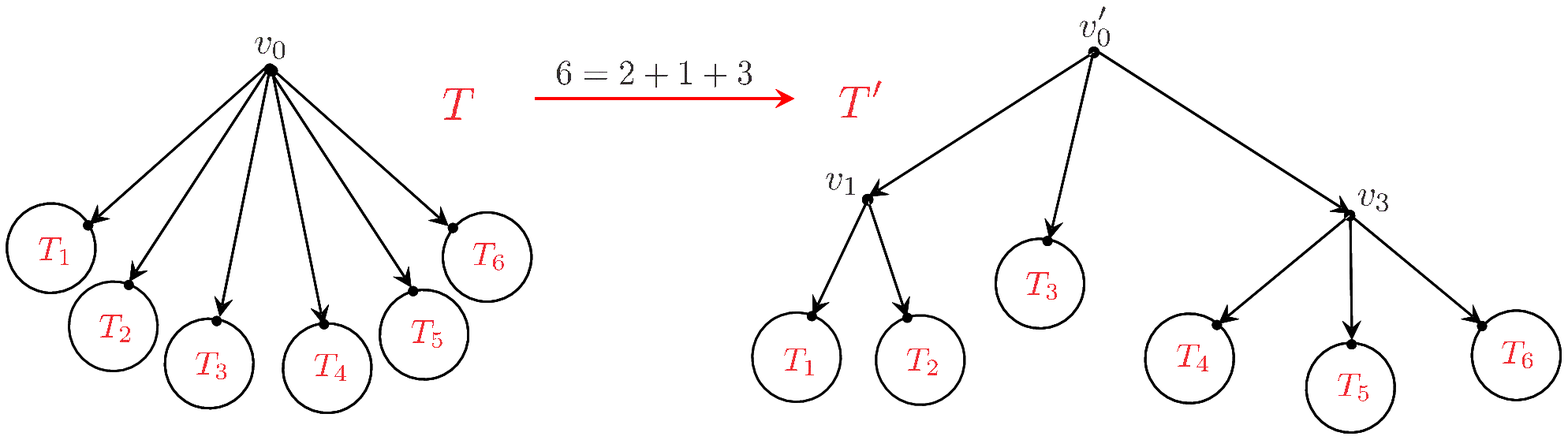}{17.cm}{fig-eff-part}{-1.3cm}

Note that each set $\cI=\{j_1,\dots,j_{m-1}\}$ provides an ordered partition of $n=n_1+\cdots n_m$ with $n_k=j_{k}-j_{k-1}$.
Then according to the Proposition we have
\be
\intPhi_1\tPhi_n=\sum_{n_1+\cdots n_m=n \atop n_k\ge 1}  \intPhi_m(\bfx')
\prod_{k=1}^m \tPhif_{n_k}(x_{j_{k-1}+1},\dots,x_{j_k}),
\label{Phipart}
\ee
where
\be
\tPhif_{n}(\bfx)=\frac{1}{2^{n-1}}\prod_{k=1}^{n-1}\sgn(\bfu_k,\bfx-\sqrt{2\tau_2}\,\beta \ptt)
=(-y)^{-\sum_{i<j} \gamma_{ij} }\gfr_n(\{\gama_i,c_i\},y).
\ee
The contribution of the trivial partition ($m=1$) combined with $\Phif_{v_0}$ in \eqref{totker-ref3}
is equivalent to $\intPhi_1\trPhi_{n}$ and thus already has the required form \eqref{totker-ref}.
For non-trivial partitions, the effect of substitution of \eqref{Phipart} into \eqref{totker-ref3} can be interpreted
as a replacement of tree $T$ by a new tree $T_{k_1,\cdots,k_m}\equiv T' $ with $k_1+\cdots k_m=k_{v_0}$, which is constructed as follows.
Group all children of the original root $v_0$ according to the decomposition of $k_{v_0}$ under consideration.
If $k_j>1$, all children in the $j$th group are connected to a vertex $v_j$
which is itself connected to the root $v'_0$ of the new tree.
Otherwise the corresponding child is connected directly to the root (see Fig. \ref{fig-eff-part}).
The contribution assigned to the new tree is then given by
\be
(-1)^{n_{T'}-1}\intPhi_m\prod_{v\in \Ch^{\rm new}}(-\tPhif_v)\prod_{v\in \Ch^{\rm old}}\Phif_v
\prod_{v\in V_{T'}\setminus{\{v'_0\cup \Ch(v'_0)\}}}\Phif_{v},
\ee
where $\Ch^{\rm new}$ denotes the set of the added vertices $v_j$, whereas $\Ch^{\rm old}$ is the set of those children of the root
which have already been such children before the above operation and are not the leaves.
It is clear that the sum over trees $T$ and partitions is equivalent to the sum over trees $T'$ supplemented by the sum over all possible
assignments of  `new' and `old' to the children of the root. The latter sum can easily be evaluated and one obtains
\be
\sum_{T'\in\IT_n^{\rm S}}(-1)^{n_{T'}-1}\intPhi_{v'_0}\prod_{v\in Ch(v'_0)\setminus L_{T'}} \(\Phif_{v}-\tPhif_{v}\)
\prod_{v\in V_{T'}\setminus{\{v'_0\cup \Ch(v'_0)\}}}\Phif_{v}.
\label{totker-ref5}
\ee
This result precisely coincides with the contribution of non-trivial partitions to \eqref{totker-ref}. To see this,
it is enough to split $T'$ into subtrees corresponding to descendants of the root which then correspond to the trees in the formula \eqref{kerPhiintr},
whereas the effect of the root is captured by the sum over partitions.
This completes the proof of \eqref{treeFh-r}.

\section{Proof of the truncation theorem}
\label{ap-theorem}

In this appendix we prove Theorem \ref{theor} from section \ref{sec-holan}.
Our first step is to establish some useful properties of the orthogonal projections $\bfv_{\ell \perp k}$
appearing as parameters of the generalized error functions in \eqref{derivPhiE},
the factor assigned to the root vertex of each Schr\"oder tree contributing to the anomaly coefficient \eqref{solJn}.

\begin{proposition}
For collinear charges $p^a_i=N_i p_0^a$ with $N=\sum_{i=1}^n N_i$, one has
\be
\bfv_{\ell \perp k}= \left\{
\begin{array}{ll}
\frac{N}{\sum_{i=1}^{k} N_i}\,\bfv^{1k}_\ell, & \ell<k,
\\
\frac{N}{\sum_{i=k+1}^n N_i}\,\bfv^{k+1,n}_\ell,\quad & \ell>k,
\end{array}
\right.
\qquad \mbox{where}\qquad
\bfv^{ij}_\ell=\sum_{k=i}^\ell\sum_{k'=\ell+1}^j \bfv_{kk'},
\quad
i\le \ell< j.
\nn
\ee
\label{prop-orthcolin}
\end{proposition}

\begin{proof}
This Proposition is a direct analogue of Lemma \ref{lemma-orthu} and their proofs are identical provided
one uses the following dictionary: $\bfu\to\bfv$, $t^a\to p_0^a$, $(p_it^2)\to N_i$.
Note however that in contrast to the Lemma this Proposition holds only for collinear charges.
\end{proof}

This Proposition shows that after the orthogonal projection the set of vectors $\bfv_\ell$ is split into two sets of
mutually orthogonal vectors, $\{\bfv_{\ell\perp k}\}_{\ell<k}$ and $\{\bfv_{\ell\perp k}\}_{\ell>k}$ (of course, for $k=1$ and $k=n-1$
one of these sets is empty).
At the same time, the generalized error functions are known to possess the property that if the vectors defining them are split into two
mutually orthogonal sets, the function is given by a product of two generalized error functions
of lower ranks evaluated on the respective sets of vectors.
Therefore, we can rewrite \eqref{derivPhiE} as
\bea
\p_{\tau_2}\cEr_n
&=&\frac{(-y)^{\sum_{i<j} \gamma_{ij} }}{2^{n-1}\tau_2}\sum_{k=1}^{n-1}\frac{(\bfv_k,\check\bfx)}{|\bfv_k|}\,
e^{-\pi\,\frac{(\bfv_k,\bfx)^2}{\bfv_k^2} }
\Phi^E_{k-1}(\{ \bfv^{1k}_{\ell}\}_{\ell=1}^{k-1};\bfx)
\Phi^E_{n-k-1}(\{ \bfv^{k+1,n}_{\ell}\}_{\ell=k+1}^{n-1};\bfx)
\nn\\
&=&\sum_{k=1}^{n-1} A_k(\{\gama_i\}_{i=1}^{n},y)\,
\cEr_{k}(\{\gama_i\}_{i=1}^{k},y)\,\cEr_{n-k}(\{\gama_i\}_{i=k+1}^{n},y),
\label{derivPhiEsplit}
\eea
where we set $\Phi^E_{0}=1$, whereas in the last line we used the definition \eqref{Erefsim} of $\cEr_n$, the fact that
\be
\sum_{1\le i<j\le n} \gamma_{ij}=\sum_{1\le i<j\le k} \gamma_{ij}+\sum_{k+1\le i<j\le n} \gamma_{ij}+\sum_{i=1}^k\sum_{j=k+1}^n\gamma_{ij},
\ee
and introduced
\be
A_k(\{\gama_i\}_{i=1}^{n},y)=J\(\sum_{i=1}^k\gama_i,\sum_{i=k+1}^n\gama_i;y\),
\qquad
J(\gama_1,\gama_2;y)=\frac{(-y)^{\gamma_{12}}}{2\tau_2}\, \frac{(\bfv_{12},\check\bfx)}{|\bfv_{12}|}\,e^{-\pi\,\frac{(\bfv_{12},\bfx)^2}{\bfv_{12}^2} }\!.
\ee

As a result, each Schr\"oder tree produces a sum of contributions given by a product of two Schr\"oder trees,
obtained by cutting the original tree at the root between the $k$th and the $(k+1)$th children,
for which every vertex carries a factor of $\cEr_{v}$, and of the factor $A_{j_k}$ where $j_k$ is the number of leaves in the first subtree.
The latter number can also be expressed as $j_k=n(T_1)+\cdots +n(T_k)$ where $n(T)$ is the number of leaves of a rooted tree $T$ and
$T_i$ are the subtrees growing from descendants of the root $v_0$.
Then it is easy to see that for each such product there are four Schr\"oder trees which produce it,
and the resulting four contributions cancel each other as shown in Fig. \ref{fig-Acancel}.
In fact, if $k=1$, the second and fourth trees do not exist (they spoil the definition of a Schr\"oder tree)
and the cancelation happens just between two trees.
Similarly, if $k=k_{v_0}$, the third and fourth do not exist.

\lfig{Combination of contributions to the anomaly coming from four Schr\"oder trees. The green dashed lines show
how each tree is cut to produce the contribution at the bottom.}
{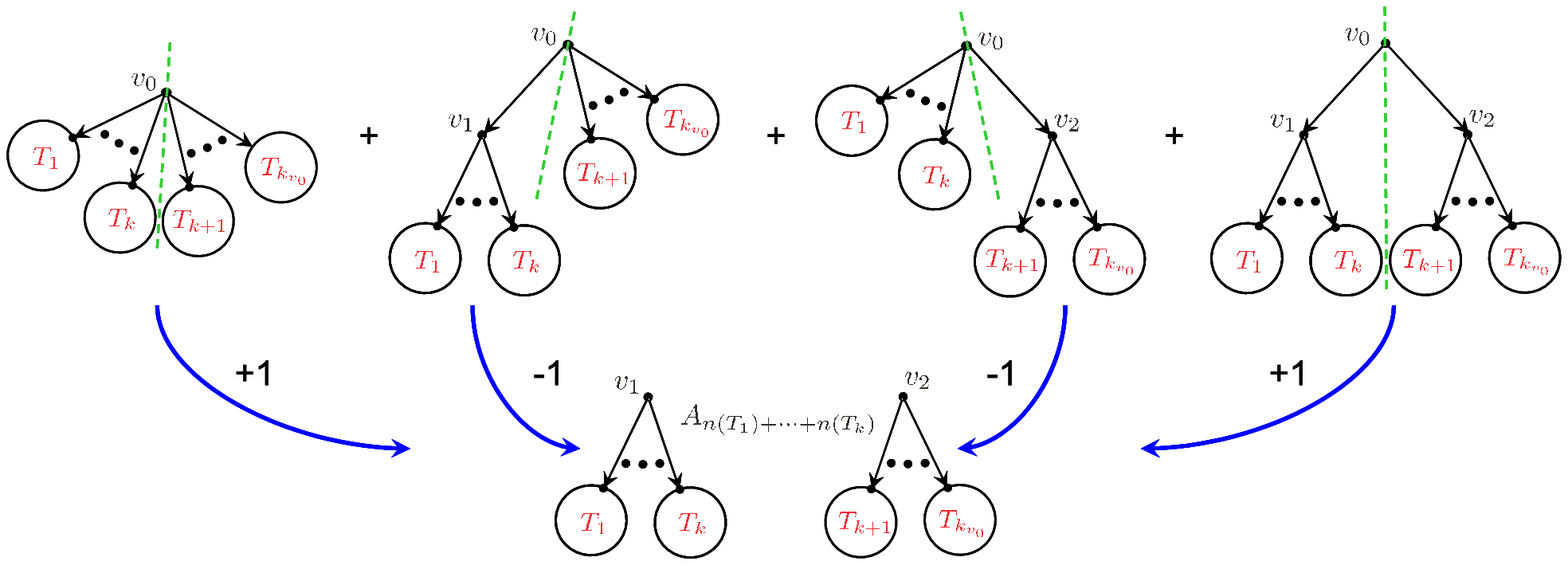}{18.cm}{fig-Acancel}{-1.2cm}

The only special case when three of the four shown trees do not exist is $n=2$.
Then only the first tree contributes and there is no cancelation giving\footnote{Note that for $n=2$, $\ptt=\bfv_{12}$.}
\be
\begin{split}
\cJr_2=&\,\frac{\I}{2}\, \Sym J(\gama_1,\gama_2;y)
\\
=&\,\frac{\I (-y)^{\gamma_{12}}}{4\sqrt{2\tau_2}}\, \frac{\gamma_{12}-\beta(pp_1p_2)}{\sqrt{(pp_1p_2)}}\,
e^{-2\pi\tau_2\,\frac{(\gamma_{12}+\beta(pp_1p_2))^2}{(pp_1p_2)} }+(1\leftrightarrow 2),
\end{split}
\label{valJr2}
\ee
where we evaluated all contractions with the bilinear form.
For all other $n$'s, all contributions cancel and $\cJr_n$ vanish.
This proves the statement of the Theorem.

\section{Geometric data for Hirzebruch and del Pezzo surfaces}
\label{sec_geodata}

In this appendix, we provide the geometric data for the complex surfaces used in
constructing local Calabi-Yau threefolds in \S\ref{sec_locallim}.
\begin{itemize}
\item
For the Hirzebruch surface $S=\mathbb{F}_k$ (also known as ruled rational surface), defined as the projectivization of the
$\cO(k)\oplus \cO(0)$ bundle over $\IP^1$, one has $b_2(S)=2$ and $\chi(S)=4$.
Using the same basis as in \cite[\S 4.1.1]{Alexandrov:2017mgi}
(see also \cite[\S A.2]{Alim:2010cf}), we get
\be
\begin{split}
&
C_{\alpha\beta} =\begin{pmatrix} 0 & 1 \\ 1 & k \end{pmatrix},
\qquad
C^{\alpha\beta} = \begin{pmatrix} -k & 1 \\ 1 & 0 \end{pmatrix},
\\
c_1^\alpha =  ( 2-k , 2 ),&
\qquad
c_{2,e} = 92,
\qquad
c_{2,\alpha}=12C_{\alpha\beta}c_1^\beta = 12( 2, 2+k ),
\end{split}
\ee
hence
\be
K_S^2 = [S]^3 = 8,
\qquad
[S] \cap c_2(T\CY) = -4.
\ee
\item
For the del Pezzo surface $S=\mathbb{B}_k$, defined as the
blow-up of $\IP^2$ over $k$ generic points,
one has $b_2(S)=k+1$, $\chi(S)=k+3$.
Using the same basis as in  \cite[\S 4.1.2]{Alexandrov:2017mgi}, we get
\be
\begin{split} &\qquad
C_{\alpha\beta} = \begin{pmatrix} 0 & 1 & 1 & \dots & 1 \\
1 & 0 & 1 &\dots & 1  \\
\vdots &  & 0 & 1 & 1 \\
1 & \dots & 1 & 0 & 1\\
1 & \dots & 1 & 1 & 1
\end{pmatrix},
\qquad
C^{\alpha\beta} = \begin{pmatrix} -1 & 0 & 0 & \dots & 1 \\
0 & -1 & 0 &\dots & 1  \\
\vdots &  & -1 & 0 & 1 \\
0 & \dots & 0 & -1 & 1\\
1 & \dots & 1 & 1 & 2-k
 \end{pmatrix},
\\
c_1^\alpha = &\, ( 1,\dots, 1, 3-k ),
\quad\
c_{2,e} = 102-10 k,
\quad\
c_{2,\alpha}=12C_{\alpha\beta}c_1^\beta = 12( 2,\dots, 2, 3 ),
\label{dataBk}
\end{split}
\ee
hence
\be
K_S^2 = [S]^3 = 9-k,
\qquad  [S] \cap c_2(T\CY) = 2k-6.
\ee
\end{itemize}
Note that $\mathbb{B}_1=\mathbb{F}_1$, whereas $\mathbb{B}_0=\IP^2$.
Smooth elliptic fibrations for these two cases have
been discussed in detail in  \cite{Klemm:2012sx}. For $k=9$, $\mathbb{B}_9$ is almost Fano
and known as the rational elliptic surface or half-K3. Vafa-Witten invariants on $\mathbb{B}_9$
were studied in \cite{Minahan:1998vr,yoshioka1999euler,Klemm:2012sx}.

\section{Modular completion of Vafa-Witten invariants on $\IP^2$}
\label{ap-vwp2}

In this appendix we provide a detailed comparison of the modular completion of the generating function
of refined VW invariants on $S=\IP^2$ for ranks $N=2$ and 3 known in the literature  \cite{Manschot:2017xcr}
with the prediction of our general formula \eqref{exp-whhr},
and spell out prediction for $N=4$.

Applying the general formulae of \S\ref{subsec-local} to the case at hand supplemented by the data
in \eqref{dataBk} with $k=0$, one obtains that the D4-brane charge for the divisor $[\IP^2]$ is $p_0^a=(1,-3)$ and $p_0^3=9$,
such that the Dirac product \eqref{gam12} becomes
\be
\label{gam12P2}
\begin{split}
\langle\gamma_1,\gamma_2\rangle=&\,
3 (N_1 q_2 -N_2 q_1).
\end{split}
\ee
Since $b_2(\IP^2)=1$, the choice of  $J$ inside the \kahler cone is irrelevant,
and the first Chern class $\mu$ is an integer modulo the rank $N$.
It will be convenient to define
\be
\hrVW_{N,\mu}(\tau,w) = g_{N,\mu}(\tau,w)\, \(\hrVW_{1,0}(\tau,w)\)^N,
\label{defgVW}
\ee
where $\hrVW_{1,0}$ was given in \eqref{h10anySref},
and similarly for  the modular completion $\whhrVW_{N,\mu}$, so that  $\widehat{g}_{N,\mu}$
transforms as a vector valued Jacobi form of weight  $\frac12(N-1)$ and
index $-\frac32(N^3-N)$.
The identification \eqref{relhh} implies
\be
\label{gtogp}
g'_{N,\mu}(\tau,w)=
g_{N,\mu+\frac12 N(N-1)}(\tau,w+\haf),
\ee
where $g'_{N,\mu}$ is the function defined by $\hr_{Np_0,\mu}$ similarly to \eqref{defgVW}.
Below we verify that once this relation is satisfied, it continues to hold for the respective modular completions.
To this end, we borrow the results for $\widehat g_{N,\mu}$ at $N=2$ and $N=3$ from \cite{Vafa:1994tf,Manschot:2017xcr}.

\subsection{Rank 2}

For $N=2$, the generating functions of refined VW invariants were computed in \cite{Yoshioka:1994,Yoshioka:1995},
generalizing the unrefined case in \cite{Klyachko:1991}. They are closely
related to the generating function of Hurwitz class numbers \cite{Zagier:1975}.
The modular completion is given by~\cite{Zagier:1975,Vafa:1994tf,Bringmann:2010sd}
\be
\label{ghat2}
\widehat g_{2,\mu}=g_{2,\mu}
+\frac{1}{2}\sum_{\ell\in  \mathbb{Z}+\mu/2}\Bigl[E_1(\sqrt{\tau_2}(2\ell+6\beta))-\sgn(\ell)\Bigr]\,\q^{-\ell^2}y^{6\ell}.
\ee
This should be compared with the result
of our general prescription which, after extracting the square of $\hr_{1,0}(\tau,w)$ as in \eqref{defgVW}, reads
\be
\widehat g'_{2,\mu}=g'_{2,\mu}+
\sum_{q_1+q_2=\mu}
\frac{(-y)^{\gamma_{12}}}{2}\(E_1\(\frac{\sqrt{2\tau_2}(\gamma_{12}+\beta(pp_1p_2)) }{\sqrt{(pp_1p_2)}}\)-\sgn(\gamma_{12})\)
e^{\pi\I\tau Q_2(\gama_1,\gama_2)},
\label{corresp2}
\ee
where $\gama_1=(p_0^a,(0,q_1))$, $\gama_2=(p_0^a,(0,q_2))$.
For this set of charges \eqref{gam12P2} gives $\gamma_{12}=3(q_2-q_1)$, whereas $(pp_1p_2)=2p_0^3=18$ and $ Q_2(\gama_1,\gama_2)=-\hf\, (q_2-q_1)^2$.
According to \eqref{quant-q}, both charges are decomposed as $q_i=\eps_i+\hf$ where $\epsilon_i\in \IZ$.
Therefore, $q_2-q_1=2\eps_2-\mu+1$ and if we set $q_2-q_1\equiv 2\ell$, then $\ell\in \IZ+(\mu+1)/2$.
Substituting all these quantities into \eqref{corresp2}, one finds perfect agreement with \eqref{ghat2}
provided one uses the identification \eqref{gtogp}.

\subsection{Rank 3}

For $N=3$, the generating functions of refined VW invariants were
computed in \cite[Eqs.(6.18), (6.22)]{Manschot:2017xcr}, generalizing
results in the unrefined case in \cite{weist2011torus, Kool2015}.
The modular completion is given by \cite[Eqs.(6.20), (6.28)]{Manschot:2017xcr}
\be
\label{g3complete}
\begin{split}
\widehat g_{3,0}
=&\, g_{3,0}-\sum_{\mu=0,1}g_{2,\mu}\,R_{1,\mu/2}(\tau,6w)-\frac{1}{4}\,R_{2,0}(\tau,6w)-\frac{1}{12}\, ,
\\
\widehat g_{3,\pm 1}=&\, g_{3,1}
-\frac{g_{2,0}}{2} \[ R_{1,\frac13}(\tau,6w) - R_{1,\frac13}(\tau,-6w) \]
-\frac{g_{2,1}}{2}\[ R_{1,\frac 16}(\tau,6w) - R_{1,\frac16}(\tau,-6w) \]
\\
&\,-\frac{1}{4}\, R_{2,(-\frac13,\frac13)}(\tau,6w),
\end{split}
\ee
where
\bea
\label{R1alpha}
R_{1,\alpha}(\tau,w) &=& \sum_{\ell\in\IZ+\alpha} \Bigl[ \sign(\ell) -
E_1\(\sqrt{3\tau_2}(2\ell+\beta)\) \Bigr]\, y^{6\ell}\, \q^{-3\ell^2},
\\
R_{2,\alpha}(\tau,w) &=&  \sum_{(k_3,k_4)\in\IZ^2+\alpha}
\biggl[ \sign( k_3)\, \sign (k_4)\,
- E_2\left( \frac{1}{\sqrt3}; \sqrt{\tau_2}(2k_3-k_4+\beta), \sqrt{3\tau_2}(k_4+\beta)\right)
\nn\\
&&
- \left( \sign (k_4) - E_1\left(\sqrt{3\tau_2}(k_4+\beta)\right) \right) \sign(2k_3-k_4)
\label{defR2}\\
&&
- \left( \sign (k_3) - E_1\left(\sqrt{3\tau_2}(k_3+\beta)\right) \right) \sign(2k_4-k_3)\biggr]
y^{k_3+k_4}\, \q^{-k_3^2-k_4^2+k_3 k_4},
\nn
\eea
and we used the parametrization \eqref{defE23} for the generalized error function of rank 2.
Note  the identity
\be
\label{R1inv}
R_{1,-\alpha}(\tau,-w)=-R_{1,\alpha}(\tau,w).
\ee

\lfig{Schr\"oder trees contributing to the completion for $N=3$. We indicated in red the $N_i$'s assigned to the leaves of the trees.}
{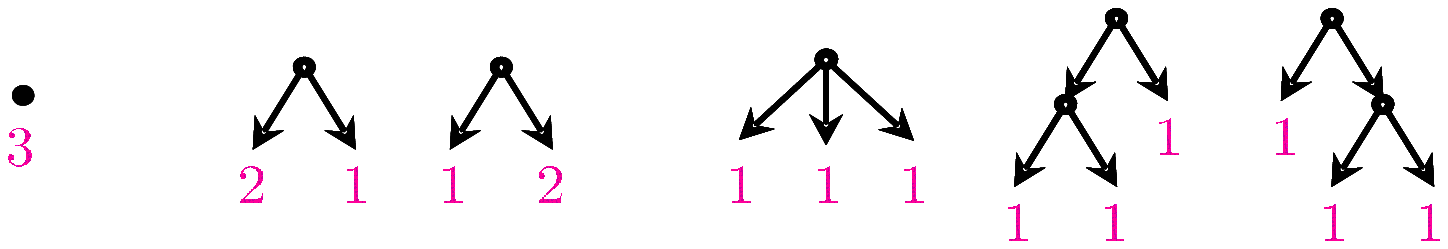}{12cm}{fig-N3}{-1.6cm}

This should be compared with the result of our general prescription, where each term originates from
one of the Schr\"oder trees shown in Fig. \ref{fig-N3},
\bea
\label{h3pred}
\widehat g'_{3,\mu} &=& g'_{3,\mu}
+\hf\sum_{q_1+q_2=\mu+\frac32}g'_{2,\mu_1}\,\expe{-\frac{\tau}{12}\,(2q_2-q_1)^2}
\\
&& \ \times\[(-y)^{3(2q_2-q_1)}
\(E_1\(\sqrt{\frac{\tau_2}{3}}(2q_2-q_1+18\beta)\)-\sgn(2q_2-q_1)\)
\right.
\nn\\
&&\left.\quad
+(-y)^{3(q_1-2q_2)}
\(E_1\(\sqrt{\frac{\tau_2}{3}}(q_1-2q_2+18\beta)\)-\sgn(q_1-2q_2)\)\]
\nn\\
&&
+\frac14   \sum_{q_1+q_2+q_3=\mu+\frac32}(-y)^{6(q_3-q_1)}
\,\expe{-\frac{\tau}{6}\((q_2-q_1)^2+(q_3-q_2)^2+(q_3-q_1)^2\)}
\nn\\
&&\ \times \biggl[\Phi^E_2(\bfv_1,\bfv_2,\bfx)-
\sgn(q_2+q_3-2q_1) \, \sgn(2q_3-q_1-q_2)-\frac13\, \delta_{q_1=q_2=q_3}
\nn\\
&&\quad
-\(E_1\(\sqrt{\frac{\tau_2}{3}}(2q_3-q_1-q_2+18\beta)\)-\sgn(2q_3-q_1-q_2)\)\sgn(q_2-q_1)
\nn\\
&&\quad
-\(E_1\(\sqrt{\frac{\tau_2}{3}}(q_2+q_3-2q_1+18\beta)\)-\sgn(q_2+q_3-2q_1)\)\sgn(q_3-q_2)
\biggr],
\nn
\eea
where
\be
\Phi^E_2(\bfv_1,\bfv_2,\bfx)=E_2\(\frac{1}{\sqrt{3}},\sqrt{\tau_2}(q_2-q_1+6\beta),\sqrt{\frac{\tau_2}{3}}(2q_3-q_1-q_2+18\beta)\).
\ee

In the second term on the r.h.s. of (\ref{h3pred}) the charges are decomposed as $q_1=2\eps_1+\mu_1$ and $q_2=\eps_2+\hf$.
Therefore $6\ell\equiv q_1-2q_2=6\eps_1-2\mu+3(\mu_1+1)$, which implies that $\ell\in\IZ-\mu/3+(\mu_1+1)/2$. The
second term then reads
\be
\begin{split}
&\frac{1}{2}\sum_{\mu_1=0,1} g'_{2,\mu_1+1}(\tau,w) \sum_{\ell\in \IZ-\frac{1}{3}\mu+\frac{1}{2}\mu_1}
\Bigl[ (-y)^{-18\ell} \(E_1(\sqrt{3\tau_2}(-2\ell+6\beta)-\sgn(-\ell)\)
\\
&\qquad + (-y)^{18\ell} \(E_1(\sqrt{3\tau_2}(2\ell+6\beta)-\sgn(\ell)\)\Bigr]\,\q^{-3\ell^2}
\\
=&\,\frac{1}{2}\sum_{\mu_1=0,1} g_{2,\mu_1}(\tau,w+\haf)\left[
R_{1,-\frac{1}{3}\mu+\frac{1}{2}\mu_1}(-6(w+\haf)) - R_{1,-\frac{1}{3}\mu+\frac{1}{2}\mu_1}(6(w+\haf))\right]
\\
=&\,-\frac{1}{2}  \,g_{2,0}(\tau,w+\haf)\[R_{1,\frac{1}{3}\mu}(6(w+\haf)) - R_{1,\frac{1}{3}\mu}(-6(w+\haf))\]
\\
&\,
-\frac{1}{2}  \,g_{2,1}(\tau,w+\haf) \[R_{1,-\frac{1}{3}\mu+\frac{1}{2}}(6(w+\haf))- R_{1,-\frac{1}{3}\mu+\frac{1}{2}}(-6(w+\haf)) \],
\end{split}
\ee
where we used \eqref{gtogp} at the second step and \eqref{R1inv} to get the last line.
Taking into account that for $\mu=0$ the functions in the brackets can be added by using again (\ref{R1inv}) and $R_{1,\alpha}=R_{1,\alpha+1}$,
this term agrees precisely with the corresponding contributions in \eqref{g3complete} with shifted $w$.

\medskip

Next we move to the third term in (\ref{h3pred}) where all charges are decomposed as $q_i=\eps_i+\hf$, $\epsilon_i\in \IZ$.
We set $k_1=-\eps_1+\mu/3$ and $k_2=-(\eps_1+\eps_2)+2\mu/3$. This term then reads
\bea
&&
\frac14  \sum_{k_1\in \IZ+\frac{\mu}{3}}\sum_{k_2\in \IZ+\frac{2\mu}{3}} (-y)^{6(k_1+k_2)}
\,\q^{-k_1^2-k_2^2+k_1k_2}
\biggl[E_2\(\frac{1}{\sqrt{3}},\sqrt{\tau_2}(2k_1-k_2+6\beta),\sqrt{3\tau_2}(k_2+6\beta)\)
\nn\\
&&\qquad
-\sgn(k_1) \, \sgn(k_2)
-\(E_1\(\sqrt{3\tau_2}(k_2+6\beta)\)-\sgn(k_2)\)\sgn(2k_1-k_2)
\label{h3rd}\\
&&\qquad
-\(E_1\(\sqrt{3\tau_2}(k_1+6\beta)\)-\sgn(k_1)\)\sgn(2k_2-k_1)
\biggr]
-\frac{\delta_{\mu=0}}{12}.
\nn
\eea
Note that the generalized error function $E_2$ appearing in \eqref{h3rd}
is invariant under $k_1\leftrightarrow k_2$ by \cite[Corollary
3.10]{Alexandrov:2016enp}. Therefore, identifying $(k_1,k_2)$ with $(k_4,k_3)$ in \eqref{defR2},
one finds that (\ref{h3rd}) equals
\be
-\frac14\, R_{2,\frac{\mu}{3}(-1,1)}(\tau,6(w+\tfrac12))
-\frac{\delta_{\mu=0}}{12},
\ee
thus reproducing the last terms in \eqref{g3complete}.
Given that for $N=3$ the shift of $\mu$ in \eqref{gtogp} is irrelevant,
we conclude that the completion \eqref{h3pred} perfectly agrees with the one found in \cite{Manschot:2017xcr}.

\subsection{Rank 4}
\label{ap-rank4}

\lfig{Schr\"oder trees contributing to the completion for $N=4$. In the last row all $N_i=1$ and are not indicated.
For each group of trees with the same set of $N_i$ we also show the decomposition of charges $q_i$
and the definition of the independent set of summation variables.}{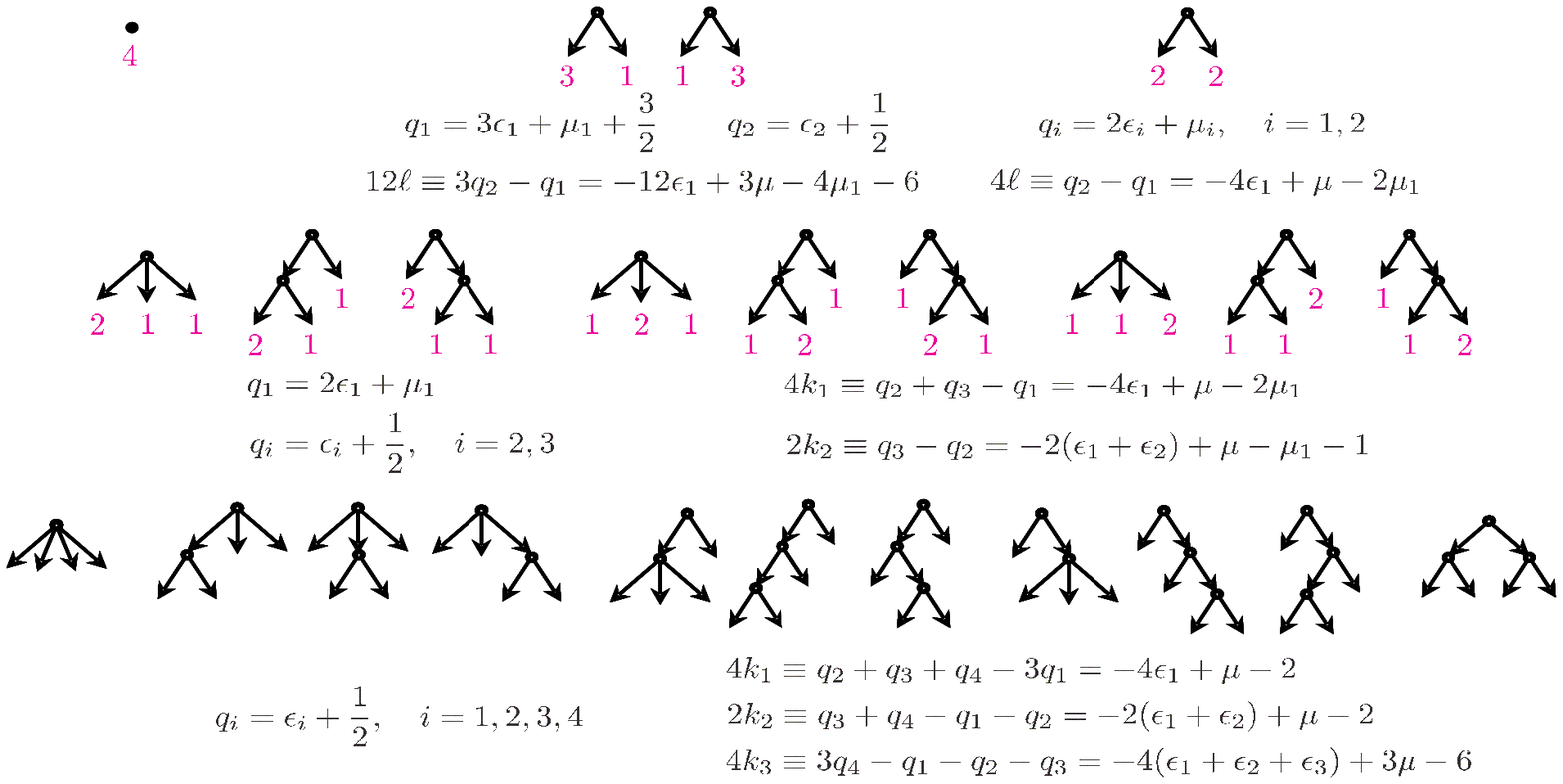}{17.5cm}{fig-N4}{-1.4cm}

For $N=4$, the generating functions of refined VW invariants were computed
in \cite{Manschot:2014cca}, as an example of a general procedure valid for arbitrary $N$.
Our general prescription \eqref{exp-whhr} predicts that the
modular completion should be given by a sum over the  trees shown in Fig. \ref{fig-N4},
where we also indicate the decomposition of charges and the definition of the variables to be summed up.
Taking into account the relation \eqref{gtogp}, one then arrives at the following prediction
\bea
\label{h4pred-final}
\widehat g_{4,\mu}&=&g_{4,\mu}
+\hf \sum_{\mu_1=0,1,2}g_{3,\mu_1}\sum_{\ell\in \IZ+\frac{\mu}{4}-\frac{\mu_1}{3}}\q^{-12\ell^2}
\Bigl[y^{36\ell}\, \tM(2\sqrt{6},3;\ell)
+y^{-36\ell}\,\tM(2\sqrt{6},3;-\ell)
\Bigr]
\nn
\eea
\bea
\quad &+&\hf\sum_{\mu_1=0,1}g_{2,\mu_1}\,g_{2,\mu-\mu_1}\,\sum_{\ell\in \IZ+\frac{\mu}{4}-\frac{\mu_1}{2}}\q^{-2\ell^2}
y^{24\ell}\,\tM(2\sqrt{2},6;\ell)
\nn\\
&+&\frac14  \sum_{\mu_1=0,1} g_{2,\mu_1}\sum_{k_1\in\IZ+\frac{\mu}{4}-\frac{\mu_1}{2}}\,
\sum_{k_2\in\IZ+\frac{1}{2}\,(\mu-\mu_1)}
\,\q^{-4k_1^2-k_2^2}
\nn\\
&&
\times \left\{y^{6(4k_1+k_2)}
\biggl[E_2\(\frac{1}{\sqrt{2}};2\sqrt{\frac{\tau_2}{3}}\,(2k_1-k_2+9\beta),2\sqrt{\frac{2\tau_2}{3}}(k_1+k_2+9\beta)\)
-s(k_1,k_1+k_2)
\right.
\nn\\
&&\quad
-\tM\(2\sqrt{\frac{2}{3}}\, ,9;k_1+k_2\)\sgn(2k_1-k_2)-\tM\(2\sqrt{2} ,6;k_1\)\sgn(2k_1+k_2)
\biggr]
\nn\\
&&
+y^{18k_2}
\biggl[-E_2\(\frac{1}{2\sqrt{2}};2\sqrt{\frac{\tau_2}{3}}\,(2k_1-k_2-18\beta),2\sqrt{\frac{2\tau_2}{3}}(k_1+k_2+9\beta)\)
-s(k_2-k_1,k_1+k_2)
\nn\\
&&\quad
+\tM\(2\sqrt{\frac{2}{3}}\, ,9;k_1+k_2\)\sgn(2k_1-k_2)
-\tM\(2\sqrt{\frac{2}{3}}\, ,9;k_2-k_1\)\sgn(2k_1+k_2)
\biggr]
\nn\\
&&
+y^{-6(4k_1+k_2)}
\biggl[E_2\(\frac{1}{\sqrt{2}};2\sqrt{\tau_2}(k_2-3\beta),2\sqrt{2\tau_2}\,(k_1-6\beta)\)
-s(k_1,k_1+k_2)
\nn\\
&&\left. \quad
-\tM\(2\sqrt{\frac{2}{3}}\, ,-9;k_1+k_2\)\sgn(2k_1-k_2)
-\tM\(2\sqrt{2},-6;k_1\)\sgn(k_2)
\biggr]\right\}
\nn
\\
&+&\frac18\sum_{k_1\in\IZ+\frac{\mu}{4}}\,
\sum_{k_2\in\IZ+\frac{\mu}{2}} \, \sum_{k_3\in\IZ-\frac{\mu}{4}}
\,\q^{-\(k_1^2-k_1k_2+k_2^2-k_2k_3+k_3^2\)}\,y^{6(k_1+k_2+k_3)}
\nn\\
&&
\times \biggl\{E_3\(\frac{1}{\sqrt3}, \frac{1}{\sqrt6}, \frac{1}{2\sqrt2} ;
\sqrt{\tau_2}(2k_1-k_2+6\beta),\sqrt{\frac{\tau_2}{3}}\,(3k_1-k_3+18\beta),
2\sqrt{\frac{2\tau_2}{3}}\,(k_3+9\beta)\)
\nn\\
&&\quad
-\sgn(k_1) \, \sgn(k_2)\, \sgn(k_3)
-\frac13\bigl(\delta_{k_2=k_3=0}\sgn(k_1)
+\delta_{k_1=k_3=0}\sgn(k_2)
+\delta_{k_1=k_2=0}\sgn(k_3)\bigr)
\nn
\\
&&
-\(E_2\(\frac{1}{\sqrt{2}};\sqrt{\frac{\tau_2}{3}}\,(3k_2-2k_3+18\beta),2\sqrt{\frac{2\tau_2}{3}}\,(k_3+9\beta)\)
-s(k_2,k_3)\) \sgn(2k_1-k_2)
\nn\\
&&
-\(E_2\(\frac{1}{2\sqrt{2}};\sqrt{\frac{\tau_2}{3}}\,(3k_1-k_3+18\beta),2\sqrt{\frac{2\tau_2}{3}}\,(k_3+9\beta)\)
-s(k_1,k_3)\) \sgn(2k_2-k_3-k_1)
\nn\\
&&
-\(E_2\(\frac{1}{\sqrt{2}};\sqrt{\tau_2}(2k_1-k_2+6\beta),\sqrt{2\tau_2}\,(k_2+12\beta)\)
-s(k_1,k_2)\)\sgn(2k_3-k_2)
\nn
\\
&&\quad
-\tM\(2\sqrt{\frac{2}{3}}\, ,9;k_3\)
\biggl[\sgn(3k_1-k_3)\sgn(3k_2-2k_3)+\frac13\,\delta_{k_1=\frac{k_2}{2}=\frac{k_3}{3}}
\nn\\
&& \qquad
-\sgn(2q_3-q_1-q_2)\sgn(q_2-q_1)
-\sgn(q_2+q_3-2q_1)\sgn(q_3-q_2)\biggr]
\nn\\
&&\quad
-\tM\(2\sqrt{\frac{2}{3}}\, ,9;k_1\)
\biggl[\sgn(3k_2-2k_1)\sgn(3k_3-k_1)+\frac13\,\delta_{\frac{k_1}{3}=\frac{k_2}{2}=k_3}
\nn\\
&&\qquad
-\sgn(3k_2-2k_1)\sgn(2k_3-k_2)
-\sgn(3k_3-k_1)\sgn(2k_2-k_3-k_1)\biggr]
\nn
\\
&&
+\tM\(\sqrt{2},12;k_2\)\sgn(2k_1-k_2)\sgn(2k_3-k_2)
\biggr\},
\label{resN4}
\eea
where we introduced convenient notations
\bea
\tM(a,b;\ell)&=& E_1\(a\sqrt{\tau_2}(\ell+b\beta)\)-\sgn(\ell),
\eea
\be
s(k_1,k_2)=\sgn(k_1)\,\sgn(k_2)+\frac13\, \delta_{k_1=k_2=0}\, .
\ee

\providecommand{\href}[2]{#2}\begingroup\raggedright\endgroup

%\bibliography{combined}
%\bibliographystyle{utphys}

\end{document}